\documentclass[letterpaper,twocolumn,10pt]{article}
\usepackage{usenix2019_v3}

\usepackage[english]{babel}
\usepackage{balance}
\usepackage{subcaption}
\usepackage[font=footnotesize]{caption}
\usepackage{amsthm, amsfonts, amssymb, bm, amsmath}
\usepackage{tikz}
\usepackage[numbers]{natbib}
\usetikzlibrary{positioning,shapes.geometric,shapes, arrows.meta,fit,calc,backgrounds,quotes}
\usepackage{ifthen}
\usepackage{booktabs}
\pgfdeclarelayer{background}
\pgfdeclarelayer{foreground}
\pgfsetlayers{background,main,foreground}
\usepackage{graphicx}
\usepackage{xspace}
\usepackage{enumitem}
\usepackage[normalem]{ulem}

\setlist{labelindent=0pt,labelwidth=0.5\parindent,labelsep=0.25\parindent,itemindent=0.5\parindent,leftmargin=*}

\newcommand{\bx}{{\boldsymbol{x}}}
\newcommand{\bN}{{\boldsymbol{N}}}
\newcommand{\bc}{{\boldsymbol{c}}}
\newcommand{\DS}{\mathcal{C}}
\newcommand{\AS}{\mathcal{A}}
\newcommand{\bv}{{\boldsymbol{v}}}
\newcommand{\bz}{{\boldsymbol{z}}}

\newtheorem{theorem}{Theorem}
\newtheorem{lemma}{Lemma}

\newtheorem{proposition}{Proposition}
\newtheorem{assumption}{Assumption}
\newtheorem{definition}{Definition}

\newcommand{\gslb}{MNLB\xspace}
\newcommand{\rif}{RIF\xspace}

\newcommand{\rifrouting}{\textbf{\rif}\xspace}
\newcommand{\latencyrouting}{\textbf{Lat}\xspace}
\newcommand{\gradientrouting}{\textbf{Grad}\xspace}
\newcommand{\discreterouting}{\textbf{Disc}\xspace}
\newcommand{\flowrouting}{\textbf{Flow}\xspace}
\newcommand{\randomrouting}{\textbf{Rand}\xspace}
\newcommand{\gslbrouting}{\textbf{\gslb}\xspace}

\usepackage{etoolbox}
\makeatletter
\patchcmd{\@maketitle}{\vbox to 2.5in}{\vbox}{}{\PackageError{paper-arxiv}{Unable to adjust the title height}{}}
\makeatother

\newenvironment{acks}{\section*{Acknowledgments}}{}

\title{\Large\bfseries DLB: Distributed Load Balancing at Scale for Generative AI Inference}
\date{}

\author{
{\rm Santiago R. Balseiro\thanks{Equal contribution.}}\\
Google Research\\and Columbia University\\
{\upshape\href{mailto:sbalseiro@google.com}{sbalseiro@google.com}}
\and
{\rm Bartek Wydrowski\footnotemark[1]}\\
Google Research\\
{\upshape\href{mailto:bwydrowski@google.com}{bwydrowski@google.com}}
\and
{\rm Sameer Agarwal\thanks{Alphabetical order.}}\\
Google Research\\
{\upshape\href{mailto:sameeragarwal@google.com}{sameeragarwal@google.com}}
\and
{\rm David Applegate\footnotemark[2]}\\
Google Research\\
{\upshape\href{mailto:dapplegate@google.com}{dapplegate@google.com}}
\and
{\rm Aaron Archer\footnotemark[2]}\\
Google Research\\
{\upshape\href{mailto:aarcher@google.com}{aarcher@google.com}}
\and
{\rm Soheil Hassas Yeganeh\footnotemark[2]}
\and
{\rm Alex Iriza\footnotemark[2]}\\
Google Research\\
{\upshape\href{mailto:iriza@google.com}{iriza@google.com}}
\and
{\rm Bobby Kleinberg\footnotemark[2]}\\
Google Research\\and Cornell University\\
{\upshape\href{mailto:rkleinberg@google.com}{rkleinberg@google.com}}
\and
{\rm Balasubramanian Sivan\footnotemark[2]}\\
Google Research\\
{\upshape\href{mailto:balusivan@google.com}{balusivan@google.com}}
\and
{\rm Pranav Vaish\footnotemark[2]}\\
Google Research\\
{\upshape\href{mailto:pranavvaish@google.com}{pranavvaish@google.com}}
\and
{\rm Oscar Zegarra\footnotemark[2]}\\
Google Research\\
{\upshape\href{mailto:ozegarra@google.com}{ozegarra@google.com}}
\and
{\rm Wenxin Zhang\footnotemark[2]}\\
Google Research\\
{\upshape\href{mailto:wenxinzhang@google.com}{wenxinzhang@google.com}}
\and
{\rm Vahab Mirrokni}\\
Google Research\\
{\upshape\href{mailto:mirrokni@google.com}{mirrokni@google.com}}
\and
{\rm Amin Vahdat}\\
Google Research\\
{\upshape\href{mailto:vahdat@google.com}{vahdat@google.com}}
}

\hypersetup{
  pdftitle={DLB: Distributed Load Balancing at Scale for Generative AI Inference},
  pdfauthor={Santiago R. Balseiro; Bartek Wydrowski; Sameer Agarwal; David Applegate; Aaron Archer; Soheil Hassas Yeganeh; Alex Iriza; Bobby Kleinberg; Balasubramanian Sivan; Pranav Vaish; Oscar Zegarra; Wenxin Zhang; Vahab Mirrokni; Amin Vahdat}
}

\begin{document}
\begingroup
\hypersetup{linkcolor=black,urlcolor=black}
\makeatletter
\let\@footnotemark\H@@footnotemark
\let\@footnotetext\H@@footnotetext
\makeatother
\maketitle
\endgroup

\begin{abstract}
The reliance on scarce and expensive accelerators such as GPUs and TPUs in modern datacenters places unprecedented demands on backend infrastructure. For workloads characterized by heterogeneous service times and complex multi-stage processing, such as
Generative AI, conventional load balancing techniques are often inadequate, relying heavily on costly overprovisioning to maintain service level objectives. This paper introduces DLB, the Distributed Load Balancer, a novel system designed to minimize end-to-end user latency for large-scale, heterogeneous workloads.

DLB employs a scalable, distributed design with peer-to-peer probing to maintain real-time visibility into server capacity across large-scale, geographically distributed infrastructure. The system continuously learns latency models to estimate the latency impact of routing decisions, allowing it to effectively manage heterogeneous hardware and diverse model architectures. We provide a novel theoretical analysis of our routing algorithms that establishes their stability and global performance guarantees over time. We also evaluate DLB through extensive simulations, which show substantial gains compared to state-of-the-art load balancing algorithms. Finally, following a 22-month deployment of DLB at Google, where it facilitates large-scale Generative AI inference for thousands of different machine learning models and millions of requests per second, we detail the design choices and practical experiences gained from the system in production. Analysis of production migrations demonstrates that DLB yields statistically significant latency reductions compared to the legacy baseline, including a 17\% decrease in median latency and a 13\% decrease at the p95 tail.

\end{abstract}

\section{Introduction}
 To meet the surging demand in modern machine learning (ML), especially GenAI inference, hyperscalers often deploy model \emph{endpoints} across globally distributed \emph{cells}~\cite{verma2015large}, where inference servers execute requests on machine-learning accelerators (\emph{MLAs}) like GPUs or TPUs. Given the substantial cost of MLAs, load imbalance is exceptionally wasteful: requests may queue in one cell while compute remains idle in another. Therefore, a \emph{global load balancer}---which routes incoming requests across cells---acts as a critical lever to ``flatten the latency-utilization curve'': at a fixed capacity, it reduces queueing latency; at a fixed latency target, it reduces the required headroom and hence overall capacity.

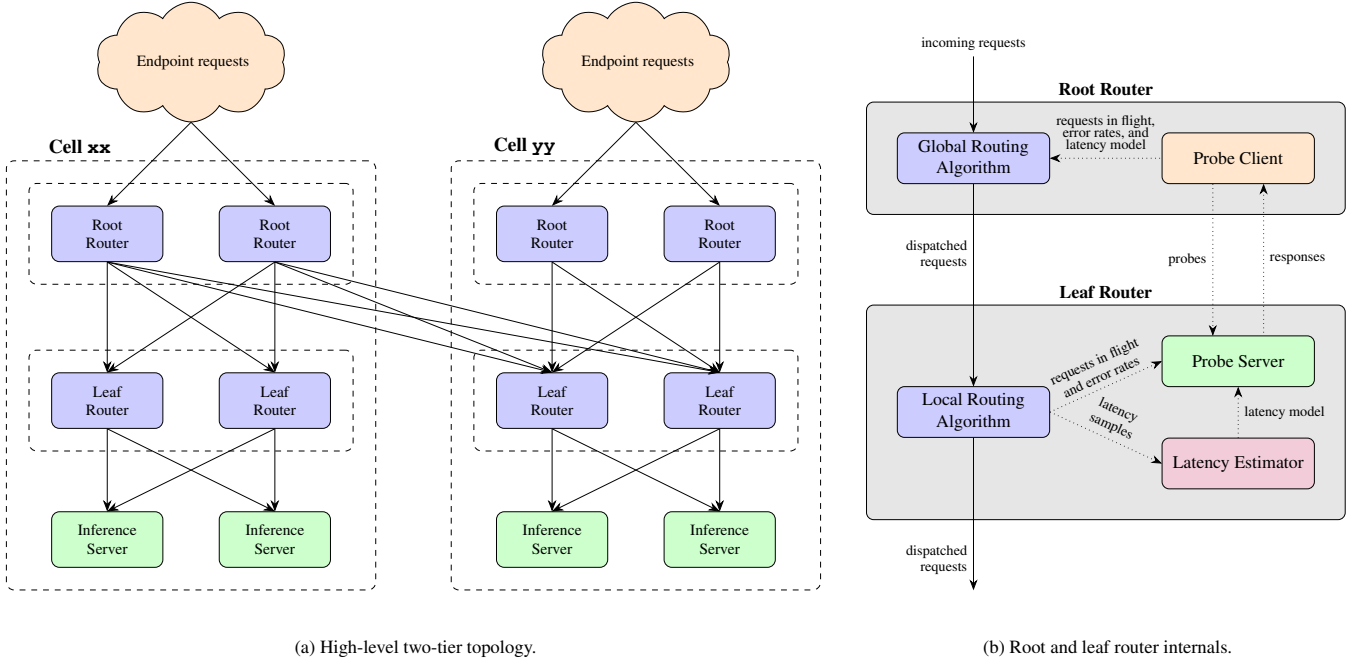
\begin{figure*}[t]
    \centering
    \begin{subfigure}[b]{0.61\textwidth}
        \centering
        \centering
\resizebox{\linewidth}{!}{%
\begin{tikzpicture}[
    node distance=1.5cm and 2cm,
    component/.style={rectangle, draw, minimum width=2cm, minimum height=1cm, align=center, rounded corners,font=\footnotesize},
    router/.style={component, fill=blue!20},
    server/.style={component, fill=green!20},
    client/.style={cloud, draw,cloud puffs=10,cloud puff arc=120, aspect=2,fill=orange!20,font=\footnotesize},
    estimator/.style={component, fill=purple!20,font=\footnotesize},
    layer/.style={draw, dashed, rounded corners, inner sep=0.4cm}
]

\begin{scope}[node distance=1cm and 1cm]
    \node[router] (rr1) {Root\\ Router};
    \node[router, right=of rr1] (rr2) {Root\\ Router};
\end{scope}

\begin{scope}[node distance=1cm and 1cm, xshift=8cm]
    \node[router] (rr3) {Root\\Router};
    \node[router, right=of rr3] (rr4) {Root\\Router};
\end{scope}

\begin{scope}[node distance=1cm and 1cm, yshift=-3cm]
    \node[router] (lr1) {Leaf\\Router};
    \node[router, right=of lr1] (lr2) {Leaf\\Router};
\end{scope}

\begin{scope}[node distance=1cm and 1cm, yshift=-3cm, xshift=8cm]
    \node[router] (lr3) {Leaf\\Router};
    \node[router, right=of lr3] (lr4) {Leaf\\Router};
\end{scope}

\begin{scope}[node distance=1cm and 1cm, yshift=-5.5cm, xshift=8cm]
    \node[server] (is3) {Inference\\Server};
    \node[server, right=of is3] (is4) {Inference\\Server};
\end{scope}

\begin{scope}[node distance=1cm and 1cm, yshift=-5.5cm]
    \node[server] (is1) {Inference\\Server};
    \node[server, right=of is1] (is2) {Inference\\Server};
\end{scope}

\node[layer, fit=(rr1) (rr2)] (root_layer_box1) {};

\node[layer, fit=(lr1) (lr2)] (leaf_layer_box1) {};

\node[layer, fit=(root_layer_box1) (leaf_layer_box1) (is1)] (data_center_box1) {};

\node[above,xshift=-2cm] at (data_center_box1.north) {\textbf{Cell \texttt{xx}}};

\node[layer, fit=(rr3) (rr4)] (root_layer_box2) {};

\node[layer, fit=(lr3) (lr4)] (leaf_layer_box2) {};

\node[layer, fit=(root_layer_box2) (leaf_layer_box2) (is3)] (data_center_box2) {};
\node[above,xshift=-2cm] at (data_center_box2.north) {\textbf{Cell \texttt{yy}}};

\foreach \x in {1,...,4}
    \foreach \y in {1,...,4}
    {
        \ifthenelse{\x < 3 \OR \y > 2}{
        \draw[-{Stealth[length=2mm]}] (rr\x.south) -- (lr\y.north);
        }{}
    }

\draw[-{Stealth[length=2mm]}] (lr1.south) -- (is1.north);
\draw[-{Stealth[length=2mm]}] (lr1.south) -- (is2.north);
\draw[-{Stealth[length=2mm]}] (lr2.south) -- (is1.north);
\draw[-{Stealth[length=2mm]}] (lr2.south) -- (is2.north);

\draw[-{Stealth[length=2mm]}] (lr3.south) -- (is3.north);
\draw[-{Stealth[length=2mm]}] (lr3.south) -- (is4.north);
\draw[-{Stealth[length=2mm]}] (lr4.south) -- (is3.north);
\draw[-{Stealth[length=2mm]}] (lr4.south) -- (is4.north);

\node (client1) [client, above=of $(rr1.north)!0.5!(rr2.north)$] {\ Endpoint requests};
\node (client2) [client, above=of $(rr3.north)!0.5!(rr4.north)$] {\ Endpoint requests};

\draw[-{Stealth[length=2mm]}] (client1.south) -- (rr1.north);
\draw[-{Stealth[length=2mm]}] (client1.south) -- (rr2.north);
\draw[-{Stealth[length=2mm]}] (client2.south) -- (rr3.north);
\draw[-{Stealth[length=2mm]}] (client2.south) -- (rr4.north);

\end{tikzpicture}
}

        \vspace{\baselineskip}
        \caption{High-level two-tier topology.}
        \label{fig:system}
    \end{subfigure}
    \hfill
    \begin{subfigure}[b]{0.36\textwidth}
        \centering
        \centering
\resizebox{\linewidth}{!}{%
\begin{tikzpicture}[
    node distance=1.5cm and 2cm,
    component/.style={rectangle, draw, minimum width=3cm, minimum height=1cm, align=center,rounded corners},
    router/.style={component, fill=blue!20},
    server/.style={component, fill=green!20},
    client/.style={component, fill=orange!20},
    estimator/.style={component, fill=purple!20},
    layer/.style={draw, rounded corners, inner sep=0.6cm,fill=black!10}
]

\begin{scope}[node distance=1cm and 1cm]
    \node[router, align=center] (rr1) {Global Routing\\Algorithm};
    \node[client, right=of rr1, xshift=1.2cm] (pc) {Probe Client};
\end{scope}

\begin{scope}[node distance=1cm and 1cm, yshift=-5cm]
    \node[router, align=center] (lr1) {Local Routing\\Algorithm};

    \node[server, right=of lr1, xshift=1.2cm, yshift=1cm] (ps) {Probe Server};
    \node[estimator, below=of ps] (le) {Latency Estimator};
\end{scope}

\begin{pgfonlayer}{background}
\node[layer, behind path, fit=(rr1) (pc)] (root_layer_box) {};
\node[above] at (root_layer_box.north) {\textbf{Root Router}};

\node[layer, fit=(lr1) (ps) (le)] (leaf_layer_box) {};
\node[above] at (leaf_layer_box.north) {\textbf{Leaf Router}};
\end{pgfonlayer}

\draw[-{Stealth[length=2mm]}] ($(rr1.south)+(0,0pt)$) -- (lr1.north) node[pos=0.35, left, align=right] {\footnotesize dispatched\\[-0.1cm]\footnotesize requests};

\draw[-{Stealth[length=2mm]}] ($(lr1.south)+(0,0pt)$) -- ($(lr1.south)+(0,-3cm)$) node[pos=0.8, left, align=right] {\footnotesize dispatched\\[-0.1cm]\footnotesize requests};

\draw[-{Stealth[length=2mm]}, dotted] ($(pc.south)-(0.5cm,0)$)  -- ($(ps.north)-(0.5cm,0)$) node[midway, left,font=\footnotesize] {probes};

\draw[{Stealth[length=2mm]}-, dotted] ($(pc.south)+(0.5cm,0)$)  -- ($(ps.north)+(0.5cm,0)$) node[midway, right,font=\footnotesize] {responses};

\draw[{Stealth[length=2mm]}-, dotted] (le.west) -- (lr1.east) node[midway, above, sloped,font=\footnotesize,align=center] {latency\\samples};

\draw[{Stealth[length=2mm]}-, dotted] (ps.west) -- (lr1.east) node[midway, above,sloped,font=\footnotesize,align=center] {requests in flight\\and error rates};

\draw[-{Stealth[length=2mm]}, dotted] (le.north) -- (ps.south) node[midway, right,font=\footnotesize] {latency model};

\draw[{Stealth[length=2mm]}-, dotted] (rr1.east) -- (pc.west) node[midway, above, align=center,font=\footnotesize] {requests in flight,\\[-0.1cm]error rates, and\\[-0.1cm]latency model};

\node (traffic) [above=of rr1, font=\footnotesize] {incoming requests};
\draw[-{Stealth[length=2mm]}] (traffic) -- (rr1.north);
\end{tikzpicture}
}

        \vspace{\baselineskip}
        \caption{Root and leaf router internals.}
        \label{fig:routers}
    \end{subfigure}
    \caption{\textbf{DLB design.} (a) Requests for an endpoint enter root routers, which select an eligible cell and forward each request to a leaf router there; the leaf then selects a local inference server. Each cell runs multiple root and leaf router replicas. Traffic-residency constraints can restrict eligible cells (here, traffic from cell \texttt{yy} cannot be routed to cell \texttt{xx}). (b) A root router runs the global routing algorithm and a probe client, while a leaf router runs the local routing algorithm, a probe server, and a latency estimator. At the leaf, the local routing algorithm exports latency samples and RIF/error state to the estimator and probe server; at the root, the probe client supplies the aggregated response to the global routing algorithm. Solid arrows show request routing; dotted arrows show telemetry and probing.}
    \label{fig:design}
\end{figure*}

\subsection{Global Load Balancing for GenAI: Motivation and Challenges}
\label{sec:old-arch}
\label{sec:challenges}
{}
 Conventional approaches---such as weighted round-robin \cite{wang2014evaluating}, fewest-connections \cite{choi2010improvement}, or centralized global schemes~\cite{bethea2018managingload}---often struggle to capture fine-grained system dynamics. These policies either rely on coarse proxies such as active connection counts that correlate poorly with actual serving latency or use long control loops.

At Google, the primary global load balancer, which we will refer to as the \emph{minimum network latency balancer} (\gslb), works as follows~\citep{gcp_service_lb_policy}. At intervals on the order of tens of seconds, a centralized optimizer computes routing weights that minimize network latency subject to serving-capacity constraints; traffic then follows those weights until the next update. Because \gslb accounts for serving latency only implicitly via capacity, operators must maintain conservative headroom to avoid the steep region of each cell's utilization--latency curve. This architecture has proven exceptionally robust for traditional web services with ample headroom and short, homogeneous ``grains of sand'' requests.

Modern ML serving challenges these premises. First, workloads exhibit diverse attributes, with service times spanning orders of magnitude: from milliseconds for high-throughput recommendation systems (``sand'') to several minutes for complex reasoning tasks (``boulders''). For long-running queries, serving latency significantly exceeds network latency, making it beneficial to route requests remotely to alleviate any local compute backlog. Second, capacity headroom is scarce and costly: provisioning additional MLA capacity on demand is severely limited in the face of resource constraints and overprovisioning is often prohibitively expensive. Consequently, inference servers are forced to operate at high utilization (running ``hot''), leaving minimal headroom to absorb traffic bursts.

Therefore, the load balancer must jointly optimize network and serving latency and actively track how busy each cell is. However, doing this at a global scale introduces three challenges:
\begin{itemize}
    \item \textbf{Heterogeneous cells and evolving serving stacks.} Cells can vary widely in size and MLA hardware, with raw hardware performance differing by more than an order of magnitude. Furthermore, serving stacks evolve through techniques such as continuous batching \cite{yu2022orca} and disaggregated prefill-decode \cite{zhong2024distserve}. Because these operations intertwine sequential and highly parallelizable steps, identical in-flight request counts or utilization can yield drastically different serving latencies across endpoints and cells. %

    \item \hspace{0pt}\textbf{Tracking system state at scale.} Cell load and available capacity can change much faster than \gslb's tens-of-seconds update cycle. Making intelligent routing decisions requires fresh system state collection and sharing, which can be costly at scale. The system must navigate through a complex design space: \emph{how} and \emph{what} state to collect, \emph{how} to share it, and at \emph{what frequency}.

    \item \textbf{Routing under delayed feedback.} Network latency delays both user queries and system feedback. {}
    Aggressively directing incoming traffic toward a currently favorable cell can overload it before the effects of earlier decisions are observed, and repeated corrections can produce oscillations. This is particularly problematic for short, high-throughput requests where network latencies are significant relative to serving latencies.   %
\end{itemize}

\subsection{DLB: System Design and Theoretical Foundations}
\label{sec:dlb-design}

We introduce DLB, a distributed global load-balancing system with a theoretical foundation for routing under delayed feedback. DLB combines learned cell performance models and scalable state collection with routing mechanisms that account for the delayed effects of sending load.

{}
{}

\begin{itemize}
    \item \hspace{0pt}\textbf{A distributed system for heterogeneous inference.} {}
    DLB operates at Layer 7, the Application Layer, providing the request-level visibility essential for complex ML workloads. DLB separates global traffic control from local inference server selection through two tiers (Figure~\ref{fig:system}). \emph{Root routers} operate globally to choose a target cell using aggregate cell-level state, while \emph{leaf routers} operate locally to select a specific inference server within a cell.
    Both tiers use multiple router replicas, distributing the routing load and providing redundancy.
    {}

    To accommodate diverse MLA generations and evolving serving pipelines, DLB treats the cell's internal architecture as a black box. Leaf routers aggregate the cell's \emph{requests-in-flight} (\rif, requests queued or executing across its servers) and latency observations to continuously learn a cell-specific \emph{latency model}, and root routers use these models to make routing decisions (Figure~\ref{fig:routers}). Peer-to-peer sharing and active probing keep routing state fresh, allowing DLB to respond on a network round-trip timescale to changes in the serving system, including load, cell capacity, and traffic intensity.
    \item \textbf{Routing algorithms with global performance guarantees.}
    DLB supports two routing mechanisms, configured per endpoint. \emph{Discrete routing} immediately directs each request to a
    {}
    most preferred cell; it is well suited to sparse, computationally heavy ``boulders,'' for which serving latency dominates network latency.
    {}
    On the other hand, for high-throughput ``sand'' workloads that are prone to oscillations, DLB applies \emph{flow routing}, which changes routing probabilities gradually, limiting how quickly traffic shifts while the effects of earlier decisions remain unobserved. Both mechanisms support several cost signals so that practitioners can choose based on their infrastructure capabilities and deployment stage: \rif-based costs provide a robust option without latency estimation, while latency- or gradient-based costs incorporate network latency and estimates of serving latency for better routing performance.

    Recent cross-region load balancers such as SkyWalker \cite{skywalker2024} and GORGO \cite{ricci2026gorgo} provide practical mechanisms for routing and state estimation, but they do not establish performance guarantees for their routing dynamics under delayed feedback.
Existing theoretical work either omits network latency \cite{zhang2025distributed} or establishes only local convergence to optimality \cite{balseiro2025load}, leaving global behavior unresolved.
    We give a novel unified theory for global convergence of our routing algorithms in the presence of network latencies using Lyapunov's direct method (Section~\ref{sec:theory}). Our theoretical results provide a new understanding of how network latencies impact the stability of routing decisions and workload dynamics, laying a theoretical foundation for designing and tuning load balancing algorithms in the presence of stale state information. The key analytical innovation is the construction of a Lyapunov function~\cite[ch. 4]{khalil2002nonlinear} that jointly controls routing decisions and workload dynamics.
    {}

    \item \textbf{Deployment and evaluation.} DLB has been deployed at Google since November 2024 and now serves thousands of endpoints, including state-of-the-art Gemini models, and millions of requests per second across hundreds of cells and millions of MLAs.
    We evaluate DLB through production migrations and simulation (Section~\ref{sec:deployment}). Across 68 production endpoints, migration from \gslb to DLB is associated with $17\%$ lower p50, $13\%$ lower mean, $14\%$ lower p90, and $13\%$ lower p95 latency, conditional on request rate and endpoint fixed effects. In production, DLB allows us to run 20\% hotter on average.
    {}
    Figure~\ref{fig:endpoint-migration} shows one illustrative endpoint. In simulations, DLB policies show lower mean and tail latency and fewer queue-overflow errors than \gslb and weighted random routing under heterogeneity, demand bursts, and capacity outages. DLB achieves these performance gains with low system overhead: aggregate router compute cost is about $0.04\%$ of the compute cost of the MLAs it balances; traversing the root and leaf routers adds less than a millisecond of delay to requests that take hundreds of milliseconds and typically seconds or longer to complete.
\end{itemize}

    {}
    {}
{}

{}

\begin{figure}[!htbp]
    \centering
    \includegraphics[width=0.72\linewidth]{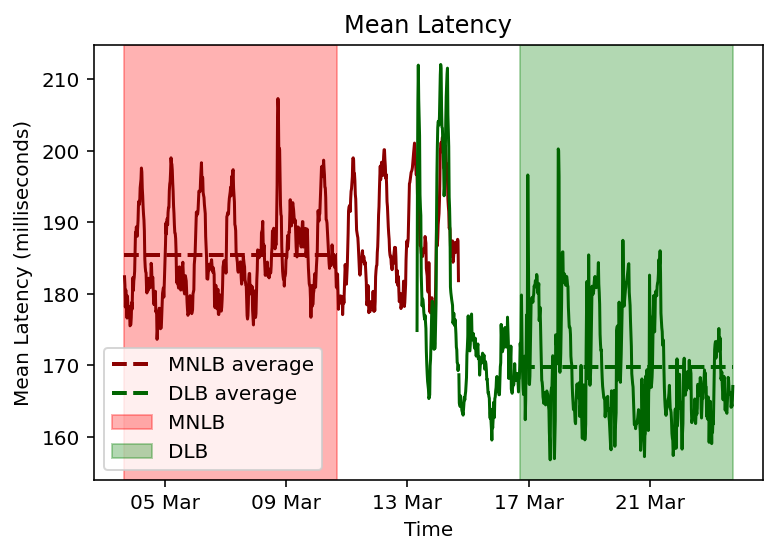}
    \caption{Mean latency of one illustrative high-throughput endpoint while transitioning from the legacy system (MNLB) to DLB in March 2025.}
    \label{fig:endpoint-migration}
\end{figure}

\section{Related Work}\label{sec:related_work}
Our work sits at the intersection of distributed systems, machine learning infrastructure, and network theory. %

\paragraph{Datacenter load balancing} Load balancing is a foundational problem in datacenter networking, typically addressed at either the packet level (Layer 4) or the application level (Layer 7).

Systems like Maglev~\cite{eisenbud2016maglev} and Ananta~\cite{patel2013ananta} provide scalable, layer 4 load balancing using consistent hashing and equal-cost multi-path routing to distribute packets across backend servers. These systems operate on flow tuples (5-tuples) and lack visibility into application-level semantics. They assume request homogeneity, which fails for GenAI workloads where service times vary by orders of magnitude. More granular approaches like Drill~\cite{ghorbani2017drill} perform micro-load balancing at the switch level to reduce tail latency in low-latency datacenter networks. However, Drill operates on microsecond timescales to manage queue buildup for short flows, whereas DLB manages seconds-long inference queries where processing latency dominates network latency.

Layer 7 balancers use content-aware routing. The ``Power of $d$-Choices'' is a classic paradigm where a router probes $d$ servers and selects the least loaded. Prequal~\cite{wydrowski2024prequal} extends this by probing to balance real-time \rif, while C3~\cite{suresh2015c3} focuses on cutting tail latency in storage systems by selecting replicas based on expected service time.
DLB instead employs a predictive model that accounts for the complex relationship between request counts (prefill/decode) and latency. Furthermore, our gradient-based routing minimizes total system latency (social optimum) rather than just balancing queue lengths.

\paragraph{Serving systems for machine learning loads}
The explosion of Large Language Models (LLMs) has spurred the development of specialized serving systems.
Most existing GenAI serving and request-routing systems balance traffic across inference servers within a single cell (e.g., vLLM Production Stack, Ray Serve, AIBrix, SGLang Router, Preble, and DualMap)~\cite{vllmproductionstack2025,rayserve2026,aibrix2024,sglangrouter2024,preble2024,dualmap2024}. This corresponds to DLB's leaf-to-server layer.

The closest systems to DLB are recent multi-region LLM load balancers. SkyWalker's cross-region routers, analogous to DLB roots, select cells based on replica availability and prefer the local cell while using longest-prefix match among eligible destinations~\cite{skywalker2024}. DLB instead uses learned cell performance models to compare heterogeneous cells.
GORGO directly scores candidate replicas using a tunable combination of network latency, queueing, and prefix-cache reuse to optimize time-to-first-token~\cite{ricci2026gorgo}. DLB instead uses cell-level \rif, latency, or gradient-based costs aimed at end-to-end latency, and we provide theoretical guarantees for the flow-routing policies in Section~\ref{sec:theory}.

\paragraph{Theoretical analysis of load balancing algorithms} There is a long stream of literature studying policies such as Join-the-Shortest-Queue~\citep{winston1977optimality,weber1978optimal,gupta2007analysis}, Join-Idle-Queue~\citep{lu2011join}, and backpressure~\citep{tassiulas1990stability}. These policies achieve system stability and attain excellent performance when cells are homogeneous and in the absence of network latencies. Weng et al.~\citep{weng2020optimal} show that variations of these policies can achieve asymptotic optimality with heterogeneous servers when networks are well connected. These papers, however, do not consider network latencies, which can induce oscillations in distributed routing~\cite{mehdian2017join}.

We frame the load balancing problem within the context of algorithmic game theory~\cite{roughgarden2010algorithmic}. Latency-based routing in DLB converges to a Wardrop Equilibrium, typical of selfish routing games. Our gradient-based routing builds upon the Greatest Marginal Service Rate policy~\cite{zhang2025distributed}, which analyzes the case without network latency, and distributed gradient descent approaches~\cite{balseiro2025load}, which provide a local stability analysis with network latencies. Our work ensures global convergence to the system-optimal solution even under delayed feedback.

\section{System Architecture and Algorithms}\label{sec:design}

We detail the transition from our previous centralized load balancing control plane (Section~\ref{sec:old-arch}) to the distributed DLB architecture (Section~\ref{sec:dlb-arch}), outlining the specific routing and probing mechanisms that enable millisecond-scale responsiveness. We undertook this transition to address the latency and heterogeneity challenges posed by GenAI inference.

\subsection{DLB System Architecture}\label{sec:dlb-arch}

DLB is a two-tier distributed system in which each cell has two types of logical routers.
The \emph{root routers} act as the entry point for traffic and are in charge of distributing the traffic across cells. The \emph{leaf routers} are in charge of assigning requests internally to inference servers within the cell. They also gather information about the current state of the inference servers, such as load, availability, and response times, to make informed routing decisions. Each router runs on a different machine and multiple replicas are maintained to increase reliability and scalability. Leaf (and root) routers utilize a peer-to-peer (P2P) communication protocol to exchange information across replicas. Our routing infrastructure supports multi-tenancy where a single router handles multiple endpoints, but for clarity, this discussion focuses on a single endpoint. %

We next describe the main components in detail. Figures~\ref{fig:system} and~\ref{fig:routers} show the fleet-wide architecture and router internals.

\paragraph{\textbf{Root routers.}}
Incoming traffic first arrives at the geographically closest cell, where it is assigned in a round-robin fashion to one of the root routers, which execute two modules: a routing algorithm and a probe client.
\begin{itemize}
    \item \emph{Global routing algorithm.} Incoming requests to a root router are immediately dispatched to a cell using one of the load balancing algorithms described in Section~\ref{sec:routing-algo}. Within the selected cell, the request is handed to one of the cell's leaf routers in a round-robin fashion. Because of traffic residency constraints, cells might not be able to route requests to every other cell (for example, European regulations require personal data to be processed within the European Union). Root routers do not queue requests. Routers implement an error aversion mechanism designed to maximize goodput (the number of requests served successfully) by dynamically shifting traffic away from cells that are producing high error rates.

    \item \emph{Probe client.} Root routers periodically probe the leaf routers of each cell. Each probe response reports the cell's \rif as well as the locally derived latency model and the cell's error rate. Roots probe at a pre-specified frequency (e.g., once every few milliseconds). To decrease network congestion, the root routers within a cell take turns probing leaves and sharing the results with local peer root routers. The probing rate can be adjusted to maintain a balance between acquiring up-to-date information and mitigating network congestion.
\end{itemize}

\paragraph{\textbf{Leaf routers.}} These routers receive requests from root routers and assign them to an inference server within the same cell. Each leaf router executes three modules: a local routing algorithm, a probe server, and a latency estimator.

\begin{itemize}
    \item \emph{Local routing algorithm.} The local routing algorithm selects an inference server within the same cell to process requests. Inference servers are selected using a distributed power-of-$d$-choices paradigm that seeks to balance the number of tokens waiting to be processed in each queue~\citep{wydrowski2024prequal}.%

    \item \emph{Latency estimators.} Leaf routers periodically collect latency samples from all the leaf routers within the cell and estimate a parametric latency model for the cell. To decrease computational burden, the leaf routers within a cell take turns estimating the latency models and sharing the resulting latency model with peer leaf routers within the cell. More details about the estimation process are provided in Section~\ref{sec:latency-estimation}.

    \item \emph{Probe server.} Upon receiving a probe from a root router, the probe server synchronously returns the cell's aggregate state, including the \rif, latency model, and error rate. To derive the \rif, the leaf router asynchronously queries all peer leaf routers within the cell and caches the aggregated results, ensuring that internal data collection does not block the probe response path. Latency is measured directly at the leaf router by timestamping each request's arrival and completion.
\end{itemize}

\subsection{Load Balancing Algorithms}\label{sec:routing-algo}

We implement two different routing mechanisms: \emph{flow routing} and \emph{discrete routing}. In both mechanisms, root routers use the state of the system and the latency models to estimate the ``cost'' of routing an additional request to each cell. The algorithms, however, differ in the way costs are used to make routing decisions. Moreover, we provide three approaches for computing the cost of routing an additional request: \emph{\rif-based}, \emph{latency-based}, and \emph{gradient-based}. In production, the mechanism is configured per endpoint based on workload attributes (Section~\ref{sec:lessons}).

Discrete routing is a fast-reacting algorithm that routes each request to the cell with the lowest cost. This algorithm can lead to fast convergence, but can suffer from oscillatory behavior when network latencies between cells are large. Flow routing is a probabilistic routing algorithm in which each request router maintains a vector of routing probabilities that are adjusted using gradient descent~\cite{balseiro2025load}. Flow routing can guarantee convergence with arbitrary network latencies at the expense of slower convergence times compared to discrete routing. Both routing algorithms are distributed and do not require estimating arrival rates.

In the \rif-based routing algorithm, the cost is set to be the ratio between the \rif and the number of inference servers in a cell. This algorithm can be understood as an instantiation of the weighted least connections policy, where the weights are set to be inversely proportional to the capacity of each cell~\cite{zhang2000linux}, and seeks to equalize the utilization across cells under the assumption that inference servers have similar processing rates. In latency-based routing, the cost is set to the estimated latency of an average request based on the cell state.  This routing algorithm is similar to the Shortest Expected Delay policy in \cite{banawan1989load} or the C3 policy of \cite{suresh2015c3}, which route requests based on the queue length and processing rate of each cell. Finally, in gradient-based routing, the cost is set to a gradient that measures the marginal impact on the total system latency of routing an additional unit of flow (in requests per second) to a cell. This routing algorithm extends the Greatest Marginal Service Rate policy of \cite{zhang2025distributed} to incorporate network latencies.

The choice of cost function dictates the equilibrium to which the system converges. This choice presents a fundamental trade-off between latency optimality and robustness. In Section~\ref{sec:theory}, we prove that gradient-based routing can achieve near-optimal system latency. However, this pursuit of optimality requires aggressive consumption of latency estimates, rendering the policy sensitive to estimation error. Conversely, \rif-based routing requires minimal information but may yield sub-optimal results by ignoring network and serving latencies. We observe, however, that \rif-based routing can lead to good performance when network latencies are small compared to serving latencies (the "boulder" case) and cells have homogeneous hardware. Latency-based routing strikes a balance between these two extremes: the system converges to the so-called Wardrop Equilibrium of the nonatomic routing game in which self-interested players route traffic through a network~\citep[Definition 18.1]{roughgarden2010algorithmic}. Because requests do not internalize the delay they impose on future requests, the total system latency can be suboptimal.

Combining the routing mechanism and the cost function yields six distinct algorithms. We next provide a formal definition of the algorithms.

\subsubsection{Formal Description of the Global Routing Algorithms}
We model the network as a directed graph $(\DS, \AS)$, where $\DS = \{1, \ldots, n\}$ represents the set of cells and $\AS \subseteq \DS \times \DS$ denotes the set of directed arcs. This topology captures the allowable source-destination cell pairings, rather than the physical links of the underlying communication infrastructure. Due to traffic residency constraints, the graph is not necessarily symmetric; that is, $(i, j) \in \AS$ does not imply $(j, i) \in \AS$. We denote by $\DS^+(i) = \{ j \in \DS : (i,j) \in \AS\}$ and $\DS^-(j) = \{ i \in \DS : (i,j) \in \AS\}$ the out-neighbors and in-neighbors of a cell, respectively. %
We denote by $N_i(t)$ the total number of \rif in the inference servers of cell $i \in \DS$ at time $t$. For every arc $(i,j)\in\AS$, let $\tau_{ij}\geq0$ denote its fixed one-way network latency.

We encode the routing decisions of the algorithm using $x_{ij}(t)$, which denotes the probability that a request arriving at cell $i\in\DS$ at time $t$ is routed to cell $j \in \DS$. We let $\bx_i(t) = (x_{ij}(t))_{j \in \DS}$ be a routing vector for cell $i\in\DS$ and assume that $\bx_i(t) \in \Delta_i = \left\{ \bz \in \mathbb R_{\geq 0}^n : \sum_{j \in \DS} z_j = 1, z_j = 0 \text{ for } j \not\in \DS^+(i) \right\}$, that is, cells can only route requests to valid destinations.

\subsubsection{Routing Mechanisms} \label{sec:routing-mechanisms}
For each admissible arc $(i,j)$, let $c_{ij}:\mathbb R_{\geq0}\to\mathbb R_{\geq0}$ be the route cost evaluated at the destination workload. We next describe the global routing algorithms.

\begin{itemize}
    \item \emph{Discrete routing.} Each root router in cell $i\in\DS$ sends requests to the cell with the least cost, i.e.,
\begin{align}\label{eq:discrete-routing}
    \bx_i(t) \in \arg\min_{\bz \in \Delta_i} \sum_{j\in\DS^+(i)}  c_{ij}\big(N_j(t - \tau_{ij})\big) \cdot z_j\,.
\end{align}
Ties can be broken arbitrarily. Note that costs are evaluated at the delayed system state $N_j(t - \tau_{ij})$ to account for network latencies between cells. We omit the delays introduced by discrete probing intervals in our description, as these intervals are negligible relative to network latency. %

\item\emph{Flow routing.} Each root router maintains a vector of routing probabilities $\bx_i(t)$ updated using gradient descent. Then
\begin{align}\label{eq:flow-routing}
    \bx_{i}(t+\delta t) = \Pi_{\Delta_i}\left(\bx_i(t)-\eta_i \cdot \delta t\cdot \boldsymbol{c}_i(t)\right)\,,
\end{align}
with $\delta t > 0$ denoting the time between updates and $\eta_i > 0$ denoting the stepsize used by the cell. The cost vector $\boldsymbol{c}_i(t)\in\mathbb{R}^{n}$ has $j$th component $c_{ij}(N_j(t-\tau_{ij}))$ for $j\in\DS^+(i)$. The time between updates is set to match the probing rate. In \eqref{eq:flow-routing} $\Pi_{\Delta_i}$ denotes the Euclidean projection to the set of feasible routing vectors, $\Pi_{\Delta_i}\left(\boldsymbol{z}\right) = \arg\min_{\boldsymbol{v} \in \Delta_i} \| \boldsymbol{v} - \boldsymbol{z}\|_2$. The projection can be done in $O(n \log n)$ time using standard algorithms~\citep{condat2016fast}.
\end{itemize}

\subsubsection{Cost Functions}
We use these cost functions.
\begin{itemize}
    \item \emph{\rif-based routing.} This cost is the \rif divided by the number of replicas $k_j \in \mathbb N$ in cell $j \in \DS$: $c_{ij}(N_j) = N_j / k_j$. The policy aims to equalize the ratio of \rif to replicas across cells.

    \item \emph{Latency-based routing.} The cost function is the expected end-to-end latency at workload $N_j$, $c_{ij}(N_j) = \ell_j(N_j) + 2 \tau_{ij}$, which is the sum of the serving and round-trip network latency. Here, we denote by $\ell_i(N_i)$ the expected latency across requests in cell $i \in \DS$ as a function of its \rif. The algorithm seeks to route requests to the cell with least latency.

    \item \emph{Gradient-based routing.} The cost function captures the marginal impact on the system latency of routing an additional unit of flow $c_{ij}(N_j) = 1/\mu_j'(N_j) + 2\tau_{ij}$, where the first term measures the marginal impact on the serving latency of all requests and the second on the network latency of this request. Here, $\mu_i(N_i)$ denotes the processing rate of cell $i \in \DS$, which captures the number of queries processed per unit of time as a function of the \rif. The term $1/\mu_j'(N_j)$ measures the sensitivity to server saturation. When the marginal service rate $\mu_j'(N_j)$ is high, pushing more flow to the cell does not significantly impact processing times. However, when $\mu_j'(N_j)$ is low, the cell is nearing saturation and pushing more flow causes a backlog because the server cannot speed up to match it. Under this choice, we prove that the system converges to the optimal latency-minimizing solution.

\end{itemize}

\subsection{Latency Estimation}\label{sec:latency-estimation}

Effective routing requires accurate estimates of serving latency and processing rate. While both can be derived from observing request arrivals and departures, we focus on latency as it is more readily instrumented in production. We specifically target average latency because it serves as a proxy for the typical user's experience and is mathematically more tractable to optimize, as discussed in Section~\ref{sec:challenges}.

We define a latency model for cell $i \in \DS$ as a function $\ell_i: \mathbb{R}_{\geq 0} \to \mathbb{R}_{\geq 0}$ that maps the number of \rif $N_i$ to an expected serving latency $\ell_i(N_i)$. Leaf routers periodically fit these models and communicate the parameters (introduced later) to root routers, which evaluate the local latencies or processing rate derivatives $\mu_i'(N_i)$ based on real-time probing. To derive processing rates from latency models, we employ the approximation $N_i \approx \ell_i(N_i) \cdot \mu_i(N_i)$. Intuitively, this captures the cell's \emph{drain time}: given a backlog of $N_i$ requests, the expected latency is roughly the time required to clear the pending work at the current rate. This relationship mirrors Little's Law, with the caveat that we apply it here as a proxy for the instantaneous state rather than a long-term average.

We opt to communicate latency models as opposed to instantaneous point estimates of latency for multiple reasons. First, latency models tend to be more stable and can be communicated with lower frequency. Second, different routing algorithms use latency models in different ways, so it is advantageous to separate the routing logic from the estimation process.

We adopt a parametric model of the latency function $\ell_i(N_i;\theta_i)$ with $\theta_i \in \mathbb R^d$ a $d$-dimensional vector representing the parameters for cell $i\in\DS$. Parametric models allow us to impose shape constraints and can be more data efficient than nonparametric models.

Figure~\ref{fig:estimation} plots our estimator together with the exponentially smoothed averages for a cell with multiple inference servers. For our parametric models, we used a softplus functional approximation, which provides a good fit in practice. Our latency models capture that latency is approximately constant when workload is low and then increases linearly when workload is high as requests start queueing.

\begin{figure}
    \centering
    \includegraphics[width=0.85\linewidth]{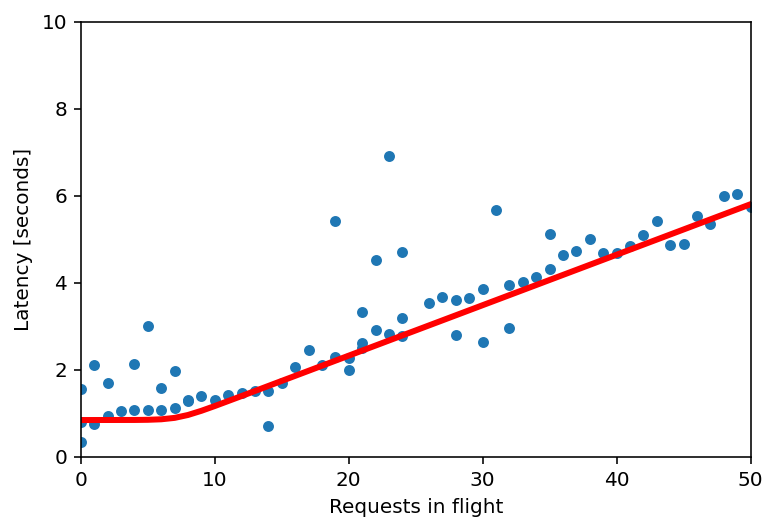}
    \caption{Latency estimation using parametric modeling. The blue dots represent exponentially smoothed average latency measurements from production traffic. The red solid line shows the fitted parametric model{}.}
    \label{fig:estimation}
\end{figure}

\subsection{Error Aversion}

We implement an error aversion mechanism designed to maximize goodput (the number of requests served successfully) by dynamically shifting traffic away from cells that are producing a high rate of errors. This is crucial for maintaining service reliability when some cells may be unhealthy, overloaded, or experiencing other issues. Error aversion allows the system to react gracefully when cells are generating errors, without manual intervention.

Each leaf router tracks the error rate of requests it handles using an exponential moving average. This error fraction is reported back to the root routers via the probing mechanism. Each root router  maintains an ``aversion rate'' for each cell, which is a value between 0 and 1. This rate represents the probability that the root router will avoid sending a request to that particular destination due to its error history. The aversion rate is periodically updated using a control loop based on how each cell compares to the others. The core idea is to penalize cells performing worse than the best one while incorporating a recovery detection mechanism that allows the system to detect if the cell becomes healthy again.

\subsection{Affinity Routing}\label{sec:affinity_routing}
DLB routes requests from the same session to the same inference server to exploit \emph{prefix caching}, which reuses the expensive prefill computation across requests that share context. We observe that 60\% to 70\% of requests can benefit from prefix caching, and reuse cuts prefill computation by 40\% to 50\%.
However, exploiting affinity could work against load balance. Static assignment schemes like consistent hashing suffer from well-known pitfalls: hash collisions can concentrate sessions on a single server, and a few ``hot sessions'' can overload that server.
DLB instead makes stickiness \emph{conditional}. A new session is routed as usual; subsequent queries return to the same cell  as long as the expected latency minus expected prefix-caching savings there is lower than that of the best alternative cell. Otherwise, the request is re-routed to the best alternative cell and the session re-affinitized.

The two tiers hold affinity state differently: root routers keep a map from session scheduling hash to cell, while leaf routers are nearly stateless---the inference servers maintain the session-to-server mappings and return them in probe responses, so every leaf router in a cell sees the same mapping without leaf-to-leaf synchronization.

\subsection{\texorpdfstring{{Fault Tolerance}}{Fault Tolerance}}
\emph{Root routers} store only session affinity caches and cell-level load statistics. When a root router crashes, clients fail over to peer replicas; missing affinity mappings temporarily fall back to load-based routing and are relearned upon request completion.
\emph{Leaf routers} are nearly stateless, storing only transient router-local counters and local latency models. When a leaf router crashes, root routers detect probe timeouts and redirect queries to healthy peers. The lost local \rif\ quickly self-corrects as peer leaves probe servers and sync via P2P broadcasts. Upon restart, the leaf recovers peer models over P2P and re-enters the active routing pool as soon as root probes succeed.

\section{Theoretical Analysis}\label{sec:theory}

{We consider a fluid model in which requests arrive at cell $i \in \DS$ at a rate $\lambda_i > 0$. In a fluid model, every request is infinitesimal and a continuous flow of queries moves deterministically through the system. For a cell $i \in \DS$, we refer to $N_i(t)$ as the \emph{workload} at time $t\ge0$ to reflect that the number of requests is now a continuous rather than discrete quantity. Fluid models are prevalent in the congestion control literature as they provide tractable approximations that are accurate for large-scale systems with high arrival rates and a large number of servers~\citep{dai1995positive, srikant2004mathematics}.}

We conduct our analysis under the following assumptions.
\begin{assumption}[Processing rate functions]
\label{assume:fsr-service-draft}
For every $j\in\DS$, the processing rate function
$\mu_j:[0,\infty)\to[0,\infty)$ is strictly increasing, continuous, and bounded, with $\mu_j(0) = 0$. Define $\bar\mu_j=\sup_{N\geq0}\mu_j(N)<\infty$.
\end{assumption}
Monotonicity is natural for work-conserving systems, while boundedness captures finite hardware throughput. Our parametric models, which provide a good fit in practice, satisfy these assumptions.
\begin{assumption}\label{assume:stability}
There exists a finite feasible pair
$(\bN',\bx')\in\mathbb R_{\geq0}^{n}\times
\prod_{i\in\DS}\Delta_i$ such that, for every $j\in\DS$,
\begin{equation*}
    \sum_{i\in\DS^-(j)}\lambda_i x_{ij}'
       =\mu_j(N_j')<\bar\mu_j.
\end{equation*}
\end{assumption}
We need Assumption \ref{assume:stability} to ensure that the system is stable, i.e., the service capacity of the cells is enough to serve all arriving jobs without allowing the queue sizes to explode.

 \begin{assumption}[Separable costs with additively slowly varying workload components]
\label{assume:fsr-cost-draft}
For every arc $(i,j)\in\AS$,
\begin{equation}
    c_{ij}(N)=f_j(N)+g_{ij},\quad N\geq0,
    \label{eq:fsr-raw-cost-draft}
\end{equation}
where $g_{ij}\geq0$ and
$f_j:[0,\infty)\to[0,\infty)$ is continuous, strictly increasing, and
coercive.  Moreover, for every fixed $H\geq0$,
\begin{equation}
    \lim_{N\to\infty}\frac{f_j(N+H)}{f_j(N)}=1.
    \label{eq:fsr-assumption-draft}
\end{equation}
\end{assumption}

In the implemented latency- and gradient-based policies, $g_{ij}=2\tau_{ij}$ is the round-trip network latency; the analysis allows any fixed nonnegative per-arc term $g_{ij}$.

We say that $f_j$ is \emph{additively slowly varying} when it
satisfies~\eqref{eq:fsr-assumption-draft}.  Thus, at large workloads, any fixed
additive workload increment changes $f_j$ by a vanishing fraction of its
current value.
\rif-based routing satisfies Assumption \ref{assume:fsr-cost-draft} automatically.
For latency-based routing, measured latency increases as requests compete for finite serving resources. At high load, additional backlog mainly adds queueing delay, motivating approximately linear growth: latency becomes unbounded, while any fixed backlog increment has a vanishing relative impact.
For gradient-based routing, $1/\mu_j'$ reflects diminishing returns to concurrency: near saturation, increasingly more backlog is needed to obtain additional throughput. The additive condition further assumes that a fixed backlog increment changes the already-small marginal throughput gain by only a vanishing fraction.

We now define \emph{equilibrium points} of our algorithms, which depend on the choice of cost functions $\bc = (c_{ij}(\cdot))_{(i,j)\in\AS}$ used to quantify the impact of routing a request.

\begin{definition}[Equilibrium point]\label{defn:eq}
A finite pair
$(\bN^*,\bx^*)\in\mathbb R_{\geq0}^{n}\times
\prod_{i\in\DS}\Delta_i$ is an \emph{equilibrium} for cost functions $\bc$ if it satisfies (1) flow balance,
\begin{equation}
    \sum_{i\in\DS^-(j)}\lambda_i x_{ij}^*=\mu_j(N_j^*),
    \qquad j\in\DS,
    \label{eq:equilibrium-flow-balance}
\end{equation}
and (2) complementary slackness, i.e., for every origin $i\in\DS$,
\begin{equation}
    \sum_{j\in\DS^+(i)}c_{ij}(N_j^*)
       (z_{ij}-x_{ij}^*)\geq0,
    \qquad \bz_i\in\Delta_i.
    \label{eq:equilibrium-vi}
\end{equation}
\end{definition}
The complementary slackness condition states that, for each cell $i\in\DS$, every arc with positive flow has the same cost. Otherwise, we could improve the solution by pushing an extra unit of flow through the arc with the lowest cost.

\begin{lemma}\label{lemma:existence-eq-points}
Suppose Assumptions~\ref{assume:fsr-service-draft}--
\ref{assume:fsr-cost-draft} hold. Then, there exists an equilibrium point. Moreover, the equilibrium workloads $\bN^*$ are unique.
\end{lemma}

Lemma~\ref{lemma:existence-eq-points} maps every set of cost functions satisfying our assumptions to unique workload levels, but, as we now discuss, different cost functions lead to distinct equilibrium points. For gradient-based routing,  Lemma~\ref{lemma:existence-eq-points} implies that equilibrium points are global minimizers of steady state average latency subject to flow balance constraints at each cell.
See Appendix~\ref{app:equilibrium-points} for discussion of this point and a proof of the lemma itself.
\cite{balseiro2025load} shows that this optimal steady-state solution constitutes a uniform lower bound on the performance of any load balancing policy.
In latency-based routing, the complementary slackness condition ensures a Nash equilibrium since no request has a unilateral incentive to deviate and select a cell with lower cost.
It is well known from the congestion games literature that the Price of Anarchy of selfish routing can be unbounded, i.e., latency-based routing can lead to arbitrarily poor system-wide latency compared to gradient-based routing~\cite[Example 18.3]{roughgarden2010algorithmic}. In addition, \rif-based routing can theoretically lead to arbitrarily poor latency in heterogeneous networks as it ignores network and serving latencies. In practice, however, both policies perform remarkably well. (See Figure~\ref{fig:simulation_body}.)

\subsection{Global Performance Guarantee}
In this section we present  our main result, a global performance guarantee that measures the cumulative distance between the workloads of our flow routing algorithms and the equilibrium levels. The most convenient distance measure for stating our results is the following semi-metric:
\begin{align}\label{eq:semimetric-workload}
    D_j(N_j, N_j^*) = \left( f_j(N_j) - f_j(N_j^*) \right) \cdot \left( \mu_j(N_j) - \mu_j (N_j^*) \right)\,,
\end{align}
where $f_j(N_j)$ is the workload-dependent component in the additive decomposition of the cost function $c_{ij}(N_j)$ given in Assumption~\ref{assume:fsr-cost-draft}.
Because the cost and processing rate functions are strictly increasing, we have that $D_j(N_j, N_j^*) \ge 0$ and $D_j(N_j, N_j^*) = 0$ if and only if $N_j = N_j^*$. While $D_j$ is symmetric, it is only a semi-metric because it does not necessarily satisfy the triangle inequality.

\begin{theorem}[Cumulative guarantee for flow routing]
\label{thm:fsr-raw-cumulative-draft}
Let $\bar\tau = \max_{(i,j) \in \AS} \tau_{ij}$. Suppose Assumptions~\ref{assume:fsr-service-draft}--\ref{assume:fsr-cost-draft} hold, together with a mild admissibility condition on the prescribed workload history over $[-\bar\tau,0]$ (Assumption~\ref{assume:fsr-history-draft} in Appendix~\ref{app:flow-fluid}). Fix an equilibrium $(\bN^*,\bx^*)$, network latencies $\boldsymbol\tau=(\tau_{ij})_{(i,j)\in\AS}$, and an admissible workload history $\left.\bN\right|_{[-\bar\tau,0]}$. Write $\underline\eta=\min_{i\in\DS}\eta_i$ and
$\bar\eta=\max_{i\in\DS}\eta_i$. There exist constants
\[
\eta_0(\boldsymbol\tau)\in(0,\infty],\qquad
G_1(\boldsymbol\tau)<\infty,\qquad
G_2(\boldsymbol\tau,\left.\bN\right|_{[-\bar\tau,0]})<\infty,
\]
independent of $T$ and the stepsizes, such that the following holds. For $\bar\eta\leq\eta_0(\boldsymbol\tau)$ and every $T\geq2\bar\tau$,
\begin{align*}
\frac1T\sum_{j\in\DS}\int_0^T D_j(N_j(t), N_j^*)\,dt
&\leq
\frac{2\sum_{i\in\DS}\lambda_i}{T\underline\eta}
+G_1(\boldsymbol\tau)\bar\eta\\
&\quad+
\frac{G_2(\boldsymbol\tau,
\left.\bN\right|_{[-\bar\tau,0]})}{T}.
\end{align*}
\end{theorem}
Theorem \ref{thm:fsr-raw-cumulative-draft} shows a finite-horizon trade-off: larger stepsizes reduce the adaptation term $1/(T\underline\eta)$ but increase the delay-dependent penalty $G_1(\boldsymbol\tau)\bar\eta$. For fixed nonzero stepsizes, taking $T\to\infty$ leaves an $O(\bar\eta)$ bound on the time-average distance. If all stepsizes are set to $\eta=T^{-1/2}$, for horizons large enough that $\eta\leq\eta_0(\boldsymbol\tau)$, the right-hand side is $O(T^{-1/2})$ and the time-average distance converges to zero.

\paragraph{{Main challenges and key proof ideas.}}
The analysis of distributed load balancing algorithms is difficult because routing decisions and workload dynamics interact, and because every router acts on a stale view of the system. Our proof relies on a Lyapunov function with two components. The first component, adapted from the classical analysis of gradient descent, measures the squared distance between the algorithm's routing decisions and their equilibrium values. The core novelty lies in the second component, built from the routing costs themselves, which explicitly controls the workload dynamics induced by request arrivals and departures. Delay leaves error terms that pair costs and routing decisions measured a round trip apart, where cost functions can be unbounded. We therefore split each cost into a bounded part and a tail, and analyze them separately. Appendix~\ref{app:global-performance} contains the full proof.

\begin{figure}
    \centering
    \includegraphics[width=0.7\linewidth]{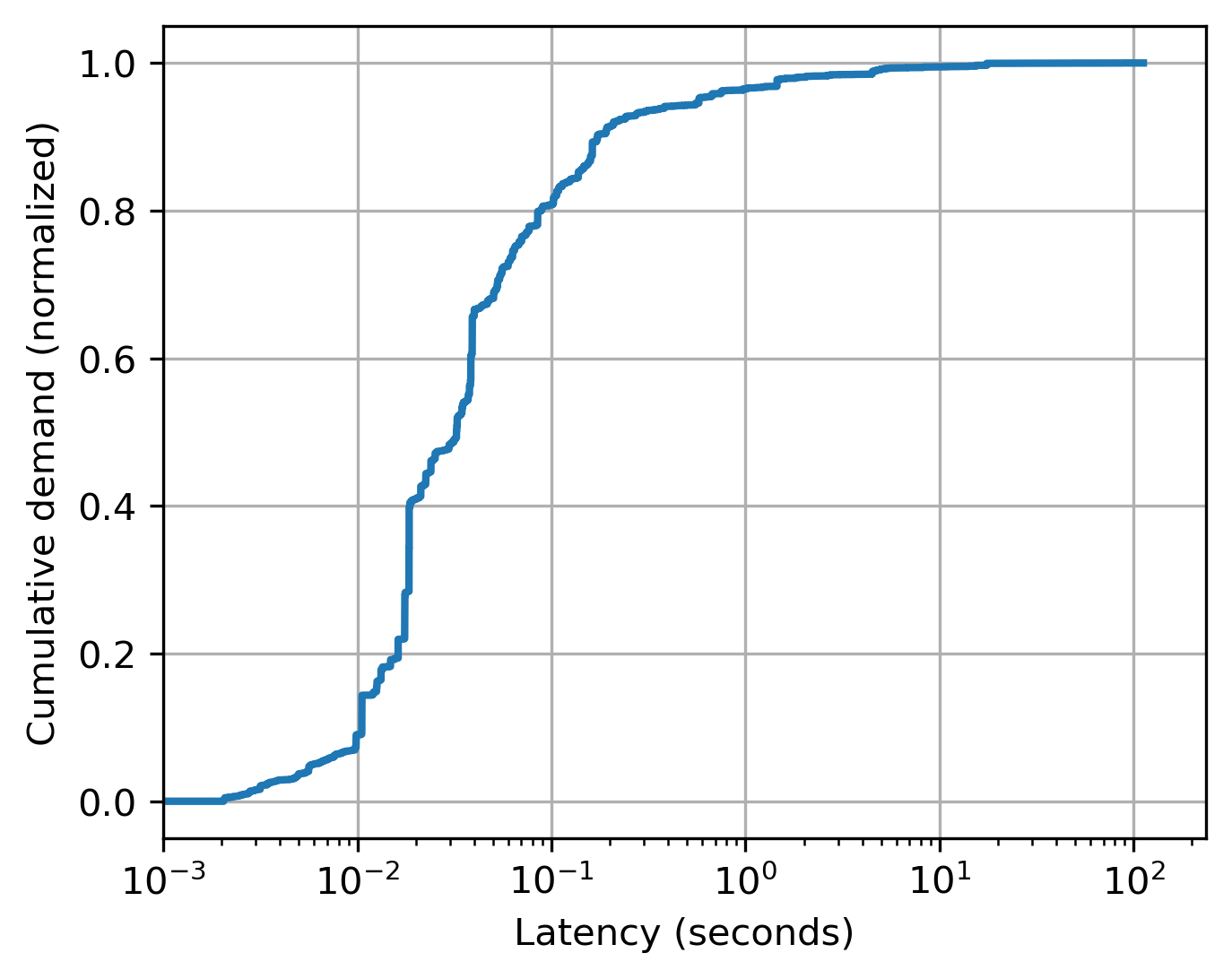}
    \caption{Cumulative distribution of latencies across our fleet of endpoints, weighted by their demand in requests per second. }
    \label{fig:fleet-latency}
\end{figure}

\begin{figure*}
    \centering
    \subcaptionbox{Demand and \rif.\label{fig:endpoint-demand}}{\includegraphics[height=4cm]{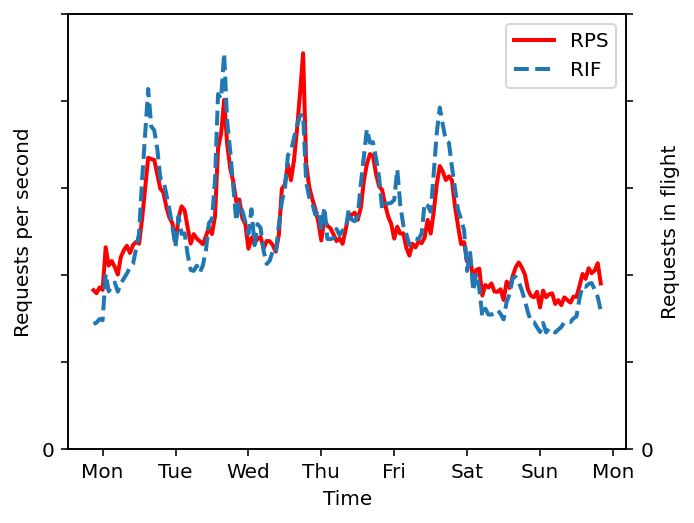}}
    \subcaptionbox{Utilization and inference servers.\label{fig:endpoint-utilization}}{\includegraphics[height=4cm]{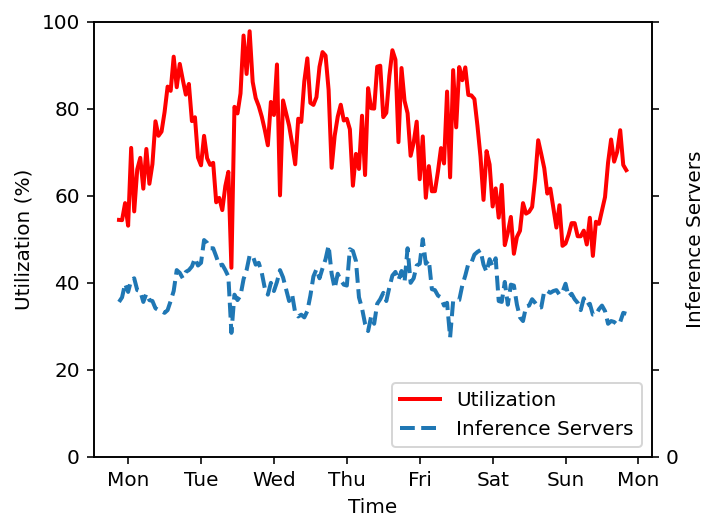}}
    \subcaptionbox{Mean and p90 latency.\label{fig:endpoint-latency}}{\includegraphics[height=4cm]{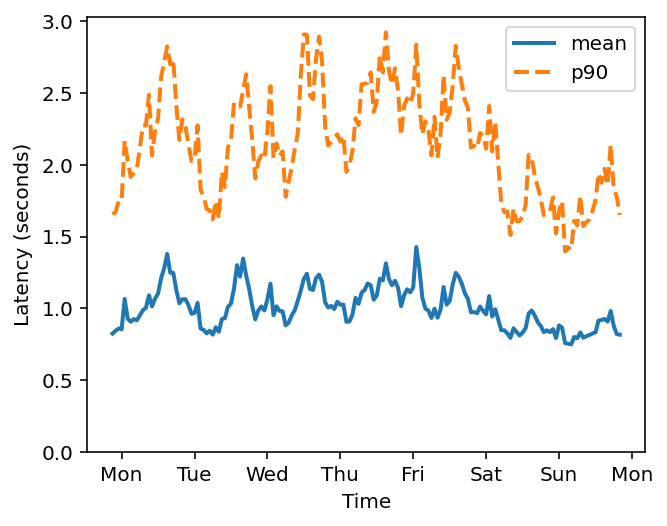}}
    \caption{Typical production data for a GenAI endpoint across one week. Axes in (a) are normalized to protect proprietary operational data.}
    \label{fig:endpoint}
\end{figure*}

\section{Deployment and Evaluation}\label{sec:deployment}

We deployed DLB starting in November 2024 as the primary load balancer for Google's GenAI serving stack, and it has since expanded to serve other types of models. While specific usage figures remain confidential, the system's adoption has tracked the rapid expansion of GenAI and machine learning workloads at Google. Notably, the number of MLAs balanced by DLB has increased tenfold since launch. DLB now supports thousands of distinct endpoints and balances millions of requests per second (RPS) across millions of MLAs in Google's global data center infrastructure.

Figure~\ref{fig:fleet-latency} illustrates the cumulative distribution of serving latencies across our fleet, weighted by request volume. Workloads are extremely heterogeneous, with latencies spanning multiple orders of magnitude: from milliseconds for lightweight CPU models, to seconds for some LLMs, to several minutes for complex reasoning requests. Although the latter are a small fraction of requests, they consume a much larger fraction of MLA chips.

\subsection{\texorpdfstring{System overhead}{System overhead}}
Each cell runs between 5 and 200 software router tasks depending on its query rate and MLA fleet size.
The aggregate compute footprint of all DLB routers is roughly $2500\times$ smaller than the cost of the MLAs they balance (i.e., about $0.04\%$). Probing (root$\rightarrow$leaf, leaf$\rightarrow$server) and P2P gather-scatter account for approximately $40\%$ of that router CPU consumption. The routing delay added by traversing the root and leaf routers is on the order of a few hundred microseconds to one millisecond; for a representative deployment of a frontier LLM, this amounts to $0.04\%$ of end-to-end latency in the mean and $0.14\%$ at p99: orders of magnitude smaller than the migration-associated latency changes reported below.

\subsection{Migration Analysis} \label{sec:migration}
This subsection compares DLB's production performance against \gslb. Note that \gslb is a \emph{strong baseline}, having routed most of Google's production traffic for 20+ years, while undergoing continuous improvement.

Figure~\ref{fig:endpoint-migration} illustrates the raw time-series dynamics of a single high-throughput endpoint before and after its live transition from \gslb to DLB. The week before the transition is red, and the week after is green. For this specific endpoint trace, raw observational latency dropped by 9\% under comparable request arrival rates (Figure~\ref{fig:endpoint-migration-qps} in Appendix~\ref{app:migration_analysis}). %

To quantify the impact associated with the \gslb to DLB migration more systematically across our entire infrastructure, we conducted an interrupted time series analysis on 68 production endpoints that migrated in 2025. We employed a fixed-effects panel regression to isolate the treatment effect from confounding factors like demand elasticity and production model heterogeneity. We specified a log-log regression model to normalize differences in scale across endpoints, allowing us to interpret regression coefficients as percentage changes~\cite{Wooldridge2010}. Our analysis is based on 56,663 observations, allowing us to estimate regression coefficients with high confidence. We defer a detailed description of the econometric analysis to Appendix~\ref{app:migration_analysis}, and provide a summary here.

Our analysis confirms that DLB yields statistically significant reductions in latency across all metrics ($p < 0.01$). Specifically, after controlling for load fluctuations across the 68-endpoint fleet, the regression estimates a relative causal latency reduction of roughly \textbf{17\% for p50, 13\% for mean, 14\% for p90, and 13\% for p95} \footnote{The 95\% confidence intervals on these estimates are $[-17.26\%, -16.30\%]$ for p50, $[-13.81\%, -12.97\%]$ for mean, $[-14.37\%, -13.32\%]$ for p90, and $[-13.23\%, -12.03\%]$ for p95 latency.} compared to the legacy baseline (\gslb), with results weighted by compute-seconds. This 17\% p50 / 13\% mean causal reduction across 68 endpoints demonstrates that while individual endpoints experience varying raw drops depending on their specific compute-to-network bottleneck (e.g., 9\% for the specific endpoint in Figure~\ref{fig:endpoint-migration}), DLB delivers substantial performance improvement, validating the efficacy of its distributed, state-aware routing in a production environment. The intra-cell load balancer (Prequal) was identical across the MNLB/DLB migration, so we attribute these gains to better inter-cell load balancing.

\subsection{\texorpdfstring{Performance of a Production Endpoint Deployment}{Performance of a Production Endpoint Deployment}}

Figure~\ref{fig:endpoint} plots a one-week performance profile of a production GenAI workload under DLB, revealing the interplay between demand, capacity, and latency.
\begin{itemize}
    \item \emph{Demand and scaling.} Demand exhibits a strong diurnal pattern with daily weekday peaks and lower peaks on the weekend  (Fig~\ref{fig:endpoint-demand}). The model management system dynamically adjusts the replica count (Fig~\ref{fig:endpoint-utilization}, dashed line), but it operates on a slower timescale than traffic bursts and is constrained by available capacity. Thus, utilization frequently spikes above 90\% during the onset of demand surges before new capacity comes online.
    \item \emph{Latency stability.} Despite the utilization spikes, the mean serving latency remains remarkably stable at around 1 second (Fig~\ref{fig:endpoint-latency}, solid line). Hence, \rif tracks RPS almost perfectly, as predicted by Little's Law.
    \item \emph{Tail effects.} Unlike the mean, the p90 latency is highly sensitive to load, fluctuating between 1.5 and 3.0 seconds in correlation with daily peaks. Queue buildup during peak hours disproportionately impacts the tail. %
    Our interrupted time series analysis in Section~\ref{sec:migration} strongly suggests that the elevation of p90 latency during peak hours would be even more severe if we had not switched from MNLB to DLB.
\end{itemize}

\subsection{Simulation-based Evaluation}\label{sec:experiments}
\begin{figure*}
    \centering
    \subcaptionbox{Mean latency.\label{fig:simulation_body_utilization_mean}}{\includegraphics[height=4cm]{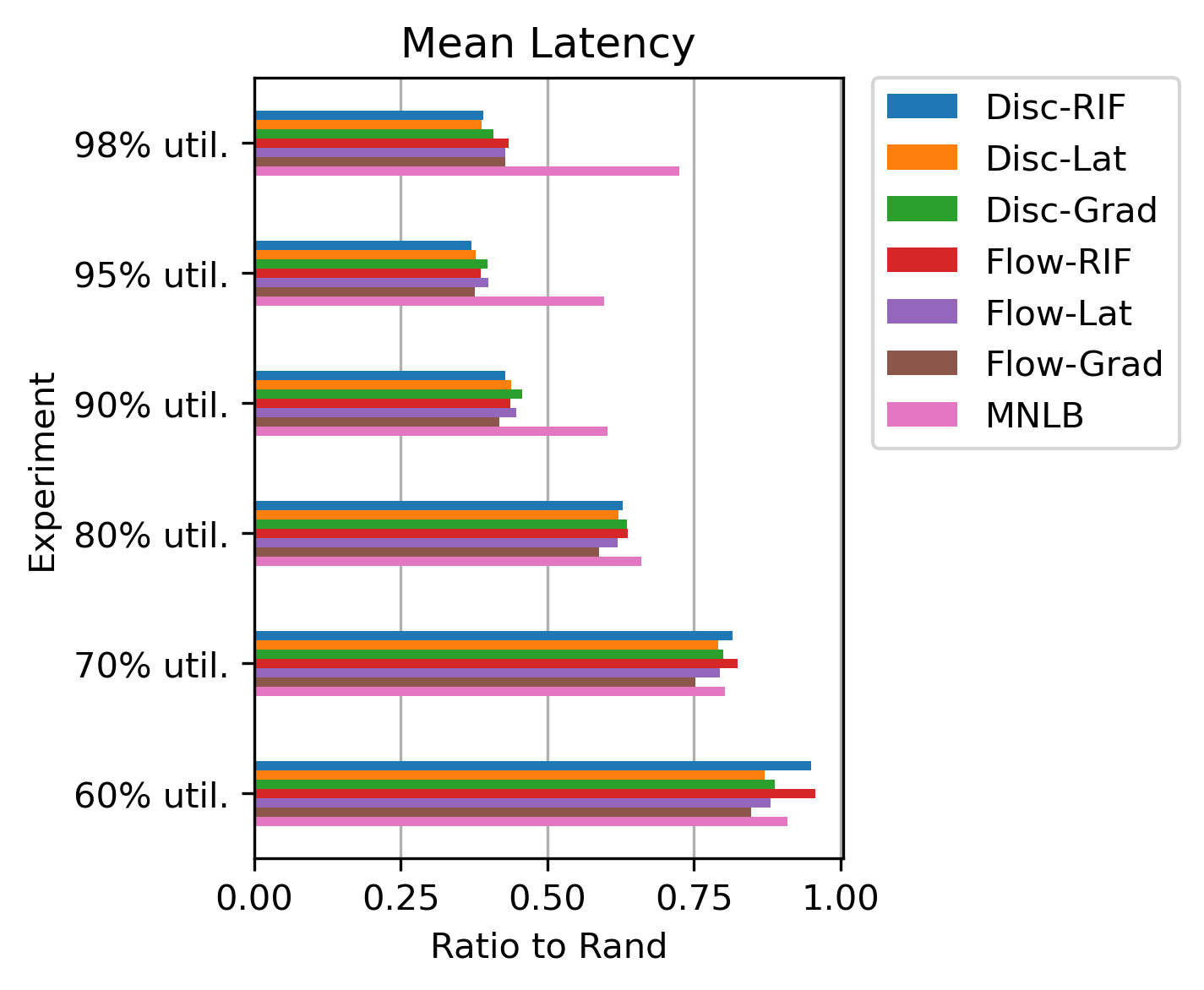}}
    \subcaptionbox{p90 latency.\label{fig:simulation_body_utilization_p90}}{\includegraphics[height=4cm]{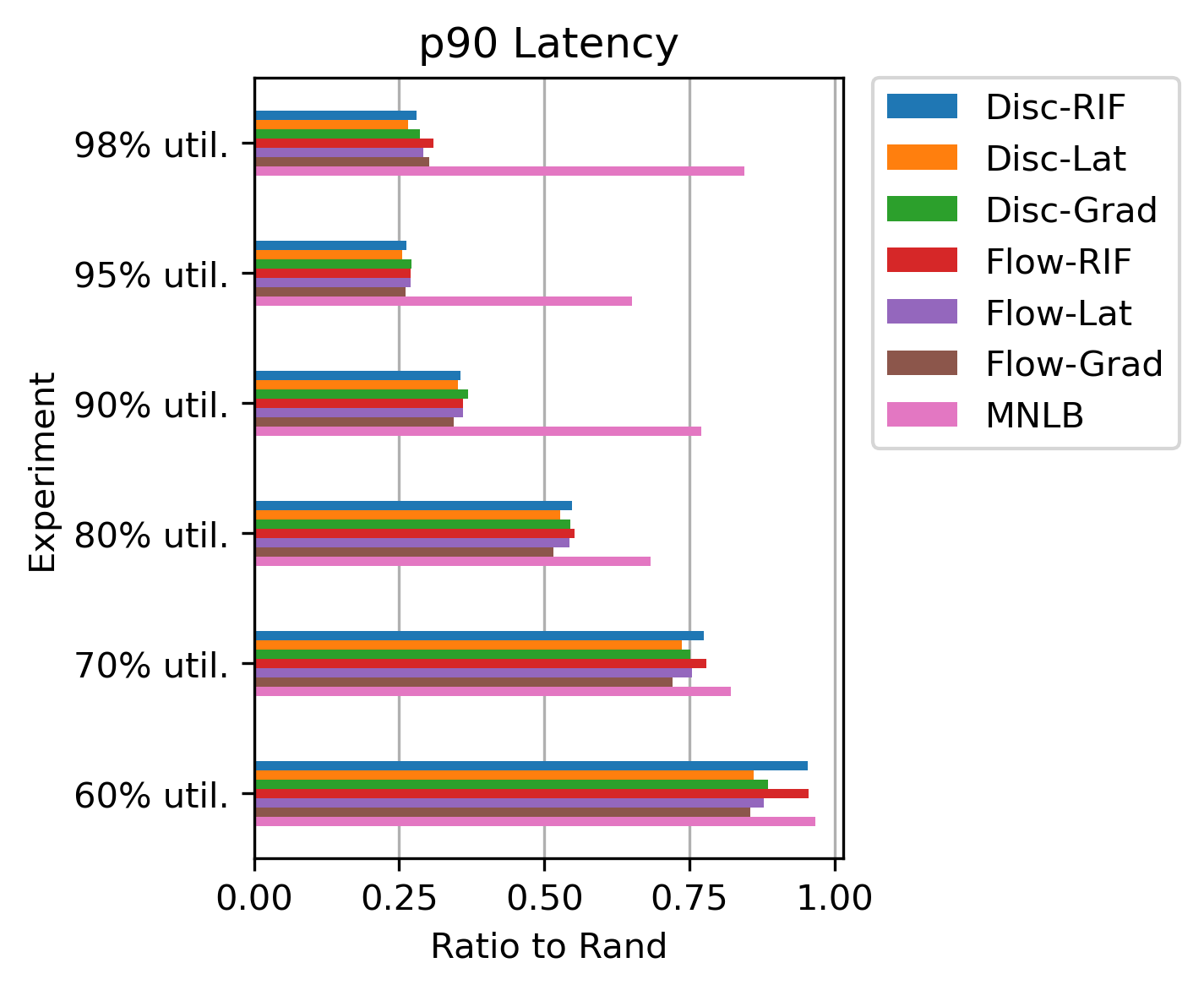}}
    \subcaptionbox{Error rates.\label{fig:simulation_body_utilization_error}}{\includegraphics[height=4cm]{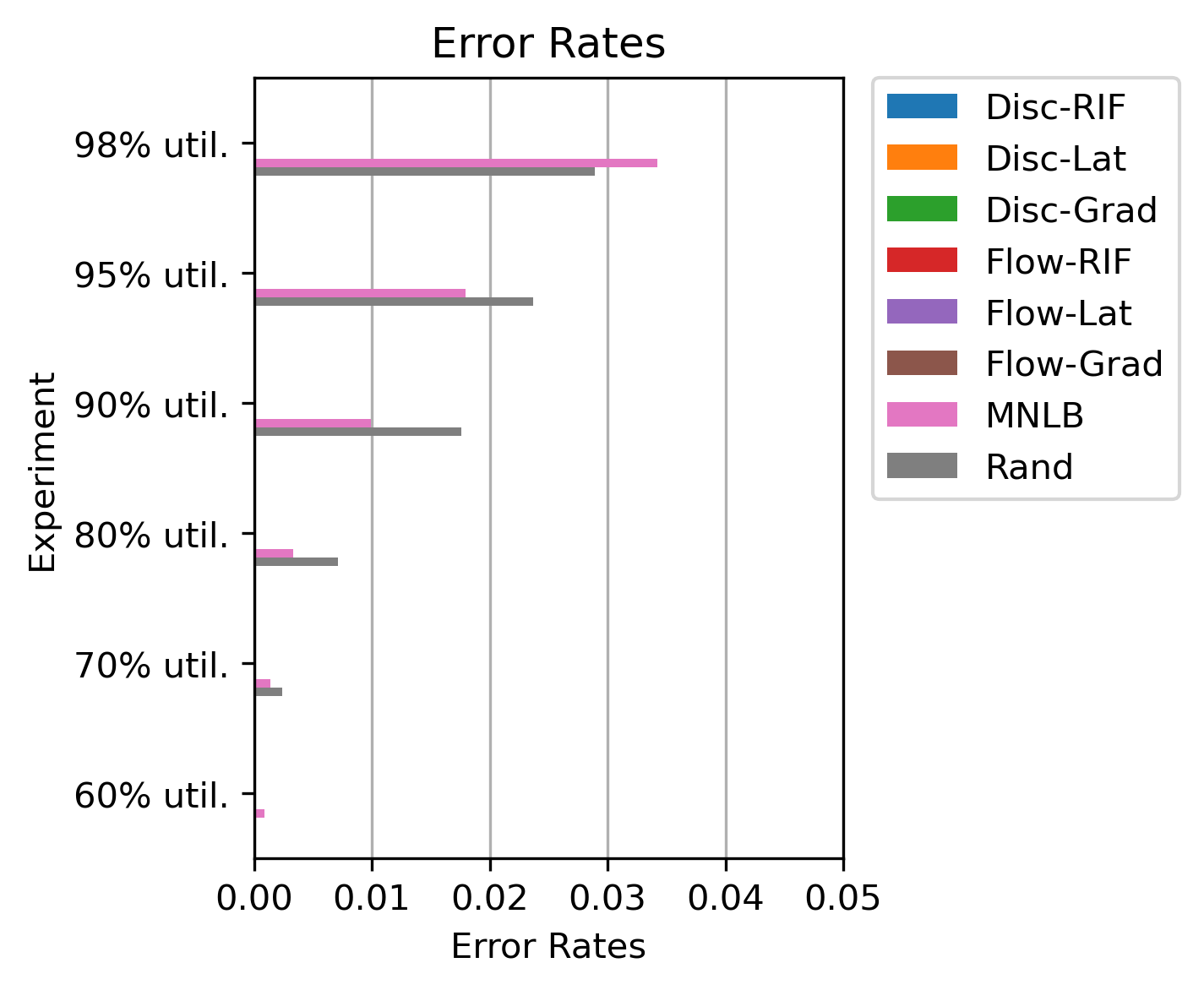}}
    \caption{Mean latency, p90 latency, and error rates of different routing algorithms across target utilization levels. Mean and p90 latency are normalized by the corresponding latency of weighted random routing (\randomrouting). }
    \label{fig:simulation_body}
\end{figure*}

\begin{figure*}
    \centering
    \subcaptionbox{Mean latency.\label{fig:estimation_error_mean_body}}{\includegraphics[height=4cm]{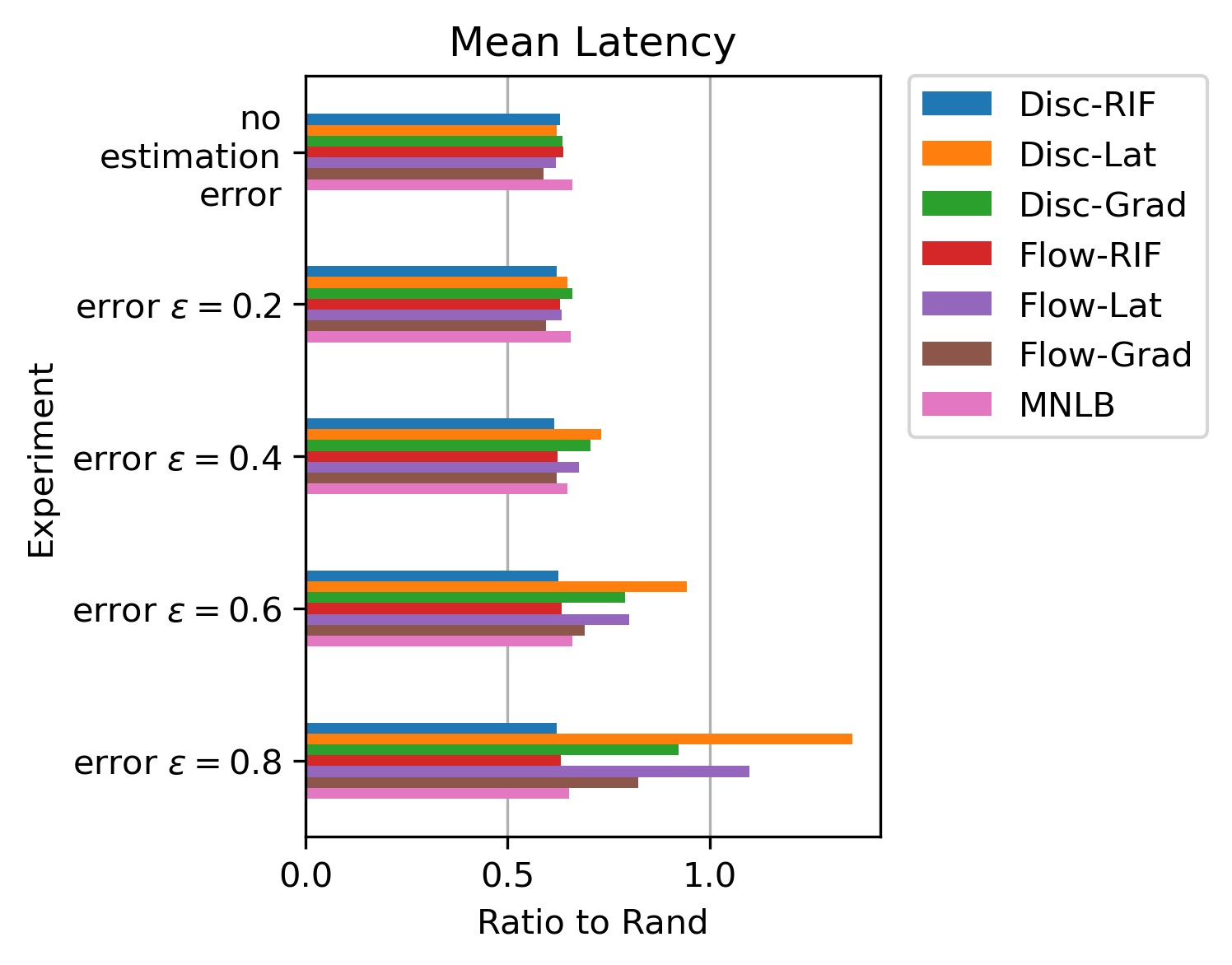}}
    \subcaptionbox{p90 latency.\label{fig:estimation_error_p90_body}}{\includegraphics[height=4cm]{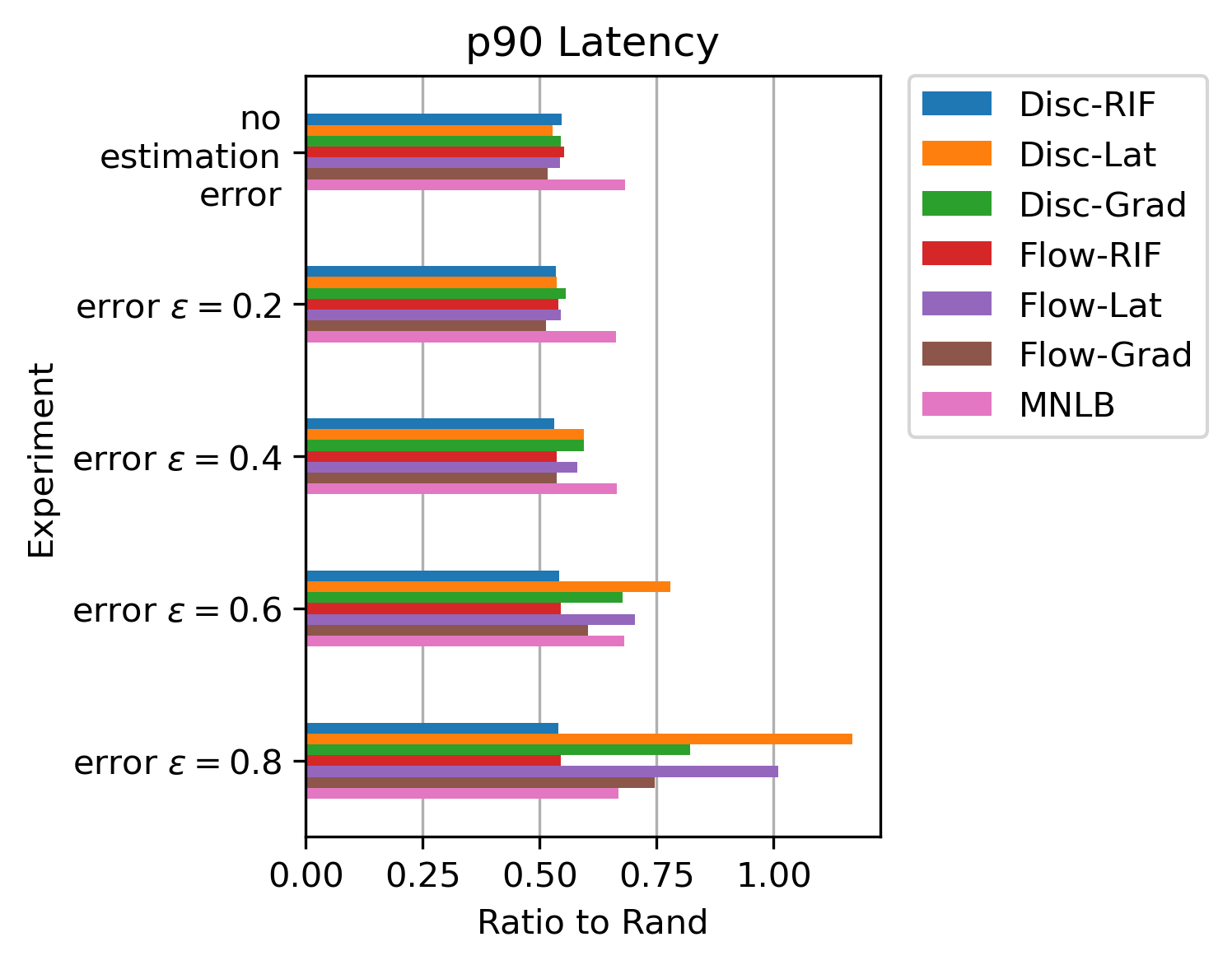}}
    \subcaptionbox{Error rates.\label{fig:estimation_error_error_body}}{\includegraphics[height=4cm]{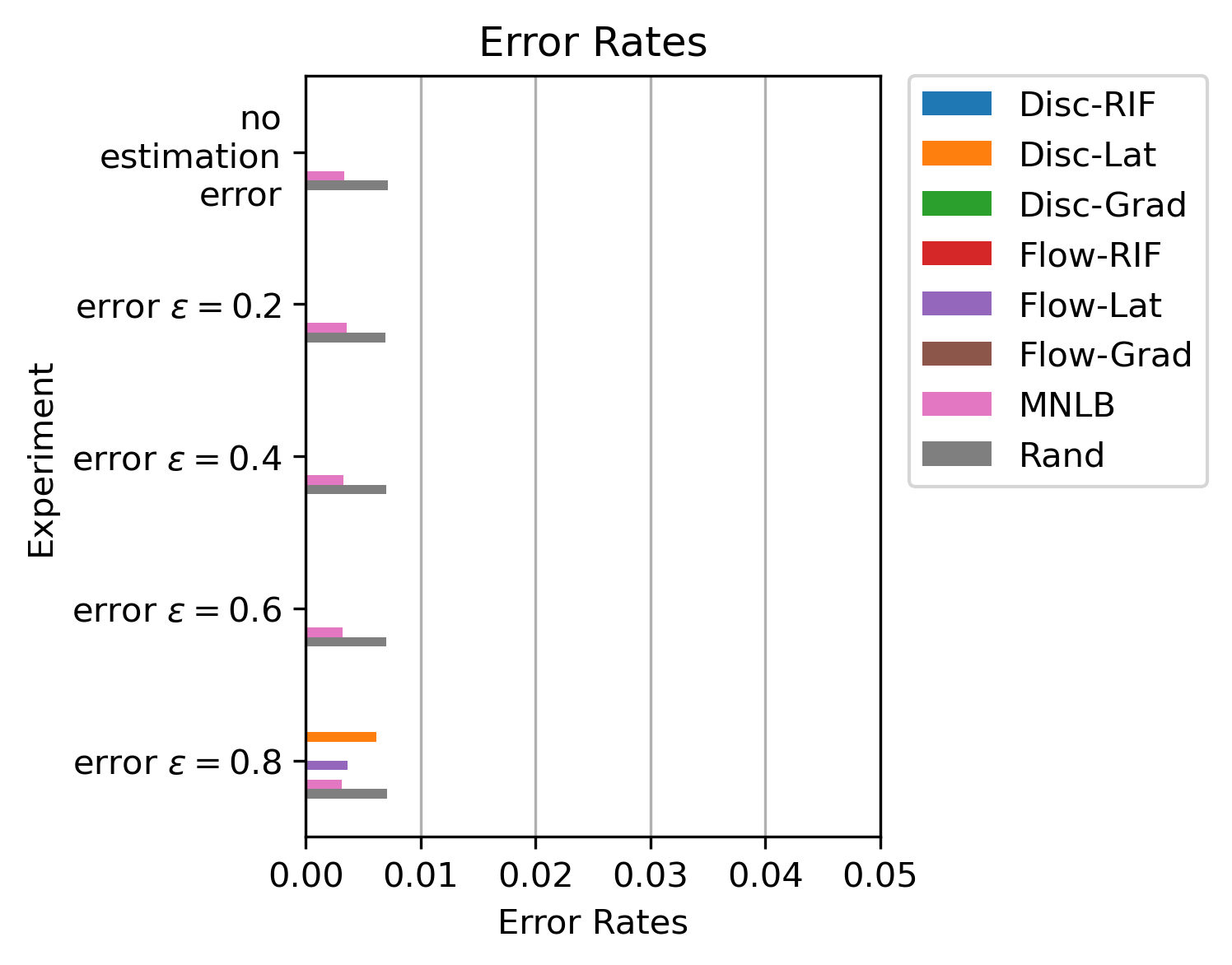}}
    \caption{Mean latency, p90 latency, and error rates of different routing algorithms subjected to varying amounts of latency estimation error. Mean and p90 latency are normalized by the corresponding latency of weighted random routing (\randomrouting).}
    \label{fig:estimation_error_body}
\end{figure*}

We evaluated DLB using a high-fidelity simulator that integrates the C++ production codebase to ensure all modules behave exactly as they do in production. The simulations consider disaggregated prefill-decode architectures across geographically distributed cells. Our simulations evaluate the routing algorithms in isolation and do not model affinity routing: requests are independent and receive no prefix-cache benefits. We tested the system under various conditions, including hardware and processing time heterogeneity, adaptability to demand bursts and capacity outages, and robustness against latency function estimation errors. We compare DLB's discrete (\discreterouting) and flow-based (\flowrouting) routing algorithms using \rif (\rifrouting), latency (\latencyrouting), and gradient (\gradientrouting) cost functions against weighted random routing (\randomrouting), which routes requests proportionally to the number of replicas in each cell~\cite{hajek1985extremal}, and the legacy centralized load balancer (\gslbrouting) described in Section~\ref{sec:old-arch}. We defer the full description of the experiments to Appendix~\ref{app:simulation} and summarize the results here.

The simulations demonstrate that DLB's  policies substantially outperform the baselines in heterogeneous scenarios. While gradient-based flow routing provided the best overall performance, particularly for short ``sand'' requests where network latency is significant, the simpler \rif-based routing proved highly competitive at high utilization levels. Additionally, DLB exhibited superior robustness compared to the legacy \gslbrouting system, significantly reducing tail latencies and error rates during demand bursts and capacity outages. %

Figure~\ref{fig:simulation_body} shows how performance varies with the target utilization level. We report \emph{normalized latency} (the ratio to the random routing baseline) to standardize comparisons across diverse scenarios. Naturally, the relative improvements of our stateful policies over \randomrouting and \gslbrouting are smaller when utilization is low, as queuing delays are less severe. This regime highlights a critical limitation of \rif-based routing. When the system is lightly loaded, the optimal strategy is to route exclusively to faster or closer cells (depending on query cost). \rif-based routing results in sub-optimal performance by spilling requests to slower cells despite the availability of faster alternatives. When utilization is high, our policies yield substantially lower latencies and error rates compared with \gslbrouting (and \randomrouting), highlighting DLB's ability to drive higher utilization and keep error rates and latencies in check.

To evaluate robustness against latency estimation error, we inject multiplicative noise into the model parameters. Given a maximum perturbation magnitude $\epsilon > 0$, we define the perturbed parameters for cell $i$ as $(1 + \epsilon \cdot \xi_i) \cdot \hat{\theta}_i$, where $\xi_i \sim \text{Unif}[-1,1]$ and $\hat{\theta}_i$ are the original parameters. We report mean latencies, p90 latencies, and error rates in Figure~\ref{fig:estimation_error_body}.

These results highlight an advantage of \rif-based routing: because routing decisions are estimation agnostic, it is immune to latency estimation errors. The latency-estimation-dependent policies \discreterouting-\latencyrouting, \flowrouting-\latencyrouting, \discreterouting-\gradientrouting, and \flowrouting-\gradientrouting degrade gracefully; while their performance declines as $\epsilon$ increases, they maintain competitive performance even in the presence of substantial parameter noise. \gslbrouting and \randomrouting do not use the latency estimator, so are also immune to errors in the estimates. Note, however, that \gslbrouting depends on other estimated inputs such as query arrival rates.

\section{Lessons Learned}\label{sec:lessons}

We highlight the specific benefits, unexpected behaviors, and hard-earned lessons from this deployment.

\paragraph{Hyperscaling to production.} DLB scaled rapidly alongside the dramatic increase in demand for GenAI models at Google. We faced the challenge of supporting a vast assortment of endpoints with heterogeneous attributes---balancing both ``boulders'' and ``sand'' (Figure~\ref{fig:fleet-latency}).
\rif-based routing is used when an endpoint is onboarded because it requires no latency estimation and robustly handles demand and capacity ramp-ups. It excels at equalizing utilization at the expense of increased network latency and bandwidth. We gradually transition endpoints to latency-aware cost functions. Discrete routing suits ``boulders,'' whose serving time dominates feedback delay, while flow routing suits ``sand,'' where delayed feedback would otherwise induce herding.
Our aim was to make routing a system concern rather than a user configuration problem, and for most endpoints it is; high-value endpoints still receive manual tuning of parameters.

\paragraph{Objectives and configuration.} We found that ``load balancing'' is often an ill-posed optimization problem: objective functions vary drastically across tenants, ranging from throughput maximization to minimizing time-to-first-token. However, exposing complex configuration knobs to accommodate every nuance creates unmanageable toil. Instead, we focused on minimizing end-to-end latency while maintaining high goodput. While this objective does not perfectly align with every stakeholder's  desires, it served as a robust proxy for system health that minimized user complaints.

\paragraph{Integration with other systems.} DLB does not operate in isolation---it is part of a cohesive serving stack. For instance, it complements the horizontal scaling capabilities of the model management system, which dynamically adjusts replica counts based on utilization. However, we found that operating them as independent control systems creates distinct challenges. First, due to the significant latency required to load large model weights into memory, the scaling system operates at a much lower frequency (minutes or hours) compared to DLB. Since DLB operates on a millisecond timescale, it must be highly reactive and absorb demand spikes using only the currently available resources. Second, the systems rely on overlapping but inconsistent sources of truth: the model management system actuates based on infrastructure metrics like utilization, whereas DLB optimizes latency. Moving forward, we hope to couple these feedback loops more tightly, enabling a unified control plane where routing and scaling decisions inform one another in real time.

\paragraph{The gap between analytical models and practice.}
A recurring lesson was that production traffic and inference servers routinely violate standard theoretical assumptions.
For example, we observed endpoints receiving ``synchronized'' traffic, such as thousands of simultaneous requests at the start of each minute, creating demand spikes that defy standard probabilistic modeling. This simultaneously creates a low average utilization with a high error rate during the burst window. High-resolution visualization played a key role in identifying root causes and providing solutions. DLB's millisecond-timescale control loop helps absorb these bursts much better than the 10-second loop of \gslb.

GenAI workloads exhibit complex behavior due to optimizations such as continuous batching and speculative decoding~\cite{leviathan2023fast}. This complexity is compounded by the rapid pace of model innovation and the heterogeneity of the serving fleet (varying in parameters, architecture, and context window). Our system relies on a robust estimation framework that treats the inference server as a black box. This approach proved critical and allowed us to adapt to the heterogeneity of the fleet without requiring model-specific tuning or deep introspection into the inference servers. As complex, multi-stage inference models are adopted, a promising research direction is to design load balancing algorithms that gain visibility into the inference process, utilizing stage-level telemetry to identify and avoid internal bottlenecks.

\paragraph{Estimation and routing form a closed loop.}
Coupling an online estimator to the policy that generates its data created two subtle control challenges:
{First}, our latency models are fit over requests that have \emph{completed} within a sliding window. This naturally censors in-flight requests and systematically biases the fitted curves toward short-lived queries; the skew is especially pronounced for ``boulder'' endpoints dominated by long prefill or generation sequences.
{Second}, and more fundamentally, the router observes a cell's latency only when routing traffic to it: a cell temporarily deemed slow stops receiving traffic, its model grows stale, and the router cannot directly observe if it has recovered. Because the policy dictates its own training distribution, online tuning degrades into an exploration--exploitation dilemma.
In practice, DLB mitigates this starvation loop via out-of-band background probing, regularizing unconfident models toward a fleet-wide prior (via Bayesian shrinkage), and injecting controlled dithering noise into latency estimates. While effective, these solutions remain empirical heuristics. This operational friction underscores why \rif-based routing serves as such a resilient default: it observes raw physical queue occupancy rather than fitting empirical curves, avoiding the additional feedback loop introduced by latency estimation.

\section{Conclusions}
This paper introduces DLB, a distributed global load balancing system designed
to address the unique resource constraints and service heterogeneity of modern workloads, with a focus on GenAI inference at scale. Our Lyapunov analysis establishes a cumulative guarantee for the continuous-time flow-routing model with finite stepsizes and fixed network latencies. In addition, we detail our experiences deploying DLB within Google’s infrastructure, demonstrating its ability to robustly serve millions of requests per second across thousands of distinct endpoints while validating the effectiveness of state-aware policies in complex production environments.

Several interesting research directions stem from this work. Our algorithms and analysis focus on optimizing mean latency; an interesting research direction is to improve our algorithms to explicitly control tail latencies, which can define user experience.
Another interesting direction is to analyze nonmonotone processing rate functions that can arise from complex inference server behavior such as batching and caching.
Finally, supporting genuinely heterogeneous tenant objectives, such as time-to-first-token alongside end-to-end latency, without reintroducing configuration toil would be highly practically relevant.

\begin{acks}
Building and deploying DLB was a massive collaborative effort, and we are deeply grateful to the many individuals who made it possible. We thank Amlan Chakraborty, Toby Davies and David Eisenstat for their engineering support. Integrating DLB into Google's next-generation machine learning serving stack could not have happened without the engineering contributions of Yanghua Huang, Chung il Lee, Matt Miecnikowski, Ken Franko, Anirudh Nambiar, Zuguang Yang, Brian Zhao and Kuan Zhu alongside the leadership support of Abhijit Karmarkar, Li Lao, Abhishek Rajgarhia and Anitha Vijayakumar. We also thank Hawkwood Glazier, who helped build an earlier prototype that provided many valuable lessons. We are grateful to Google Research leadership—specifically John Anderson, Corinna Cortes, Manu Guere, Yossi Matias—for investing in the DLB team, and Jon Orwant for trusting our instincts and supporting this bottom-up project from its inception. Finally, we thank Carla Bromberg and Tyler Russell for providing critical program management support throughout the project, with special thanks to Tyler for his assistance with data collection and the analysis of production migrations.
\end{acks}

\clearpage
\onecolumn
\appendix

\section{Migration Analysis: Data, Specification, and Results}\label{app:migration_analysis}

Many production endpoints migrated from \gslb to DLB, and we wish to
quantify the treatment effect on mean and tail latencies for these endpoints. %
Our methodology uses an interrupted time series analysis within a panel data framework, incorporating unit fixed effects and controlling for varying demand.

The dataset comprises 68 distinct inference endpoints that migrated to DLB during 2025. This dataset excludes units with very low demand or missing more than half of the observations.

We constructed a panel by tracking these units relative to their specific migration events. For each unit, we collected time-series data on mean, p50, p90, and p95 latencies, alongside the requests per second (RPS).

Let $m_i$ denote the migration date for unit $i$. To capture steady-state behavior, we defined a pre-intervention observation window from $m_i - 10$ days to $m_i - 3$ days, and a post-intervention window from $m_i + 3$ days to $m_i + 10$ days. A six-day buffer period $[m_i-3, m_i+3]$ was excluded from the analysis to prevent contamination from transient effects during the system rollout. The seven-day duration for both observation windows ensures that our comparison accounts for weekly and diurnal traffic patterns inherent to both systems. Metrics were aggregated over 20-minute intervals, yielding a set of $7 \times 24 \times 3 \times 2 = 1008$ observations per unit.

We define $\text{mean\_latency}_{it}$ and $\text{rps}_{it}$ as the mean latency and request rate, respectively, for unit $i$ at relative time step $t \in \{1,\ldots,1008\}$. The binary treatment variable $T_{it} \in \{0,1\}$ indicates the active load balancing regime, taking a value of one if DLB is active at time $t$, and zero if the legacy load balancer is active. The time index $t=1$ is normalized to the start of the pre-intervention window for each unit to align the asynchronous migration timelines. %

We adopt a fixed-effects panel regression model rather than modeling individual units separately. This approach enhances statistical power and improves external validity by pooling information across heterogeneous units. To account for time-invariant unobserved heterogeneity specific to each endpoint (e.g., underlying model architecture and computational complexity), we include unit-specific fixed effects. We specify a log-log regression model to estimate the elasticity of latency with respect to demand and the relative impact of the treatment:
$$
\log(\text{mean\_latency}_{it}) = \beta_0 + \beta_1 T_{it} + \beta_2 \log(\text{rps}_{it}) + \alpha_i + \epsilon_{it}\,,
$$

where $\beta_0$ represents the intercept or average baseline level, $\beta_1$ captures the treatment effect (the causal impact of DLB), $\beta_2$ estimates the elasticity of latency with respect to load, $\alpha_i$ is the unit-specific effect, and $\epsilon_{it}$ represents the idiosyncratic error term. Latencies enter the regression in \emph{milliseconds} (unlike some figures in the main body, which report latencies in seconds), so the intercepts $\beta_0$ in Table~\ref{tab:regression} are on a log-millisecond scale. The inclusion of $\log(\text{rps}_{it})$ as a covariate controls for exogenous demand fluctuations, isolating the impact of load on latency. The log-log specification normalizes differences in scale across endpoints, allowing us to interpret coefficients as percentage changes~\cite{Wooldridge2010}.
We use a weighted least squares regression, weighting observations by their MLA usage in seconds, to minimize the influence of low-volume outliers and weight the regression toward observations representing more MLA use.

The regression results are summarized in Table~\ref{tab:regression}. All estimated coefficients are statistically significant ($p < 0.01$), and the high $R^2$ values indicate that the model explains a substantial proportion of the variance in latency. Across all metrics, the treatment effect $\beta_1$ is negative, confirming that DLB yields a statistically significant reduction in latency compared to the baseline. As anticipated, the coefficient for request rate $\beta_2$ is positive, consistent with queuing theory predictions that increased load correlates with higher end-to-end latency.
To quantify the magnitude of the improvement, we convert the log-linear coefficients back to a percentage lift. The 95\% confidence interval for the relative latency impact is calculated as $\exp(\beta_1 \pm 1.96 \cdot \text{SE}_{\beta_1}) - 1$, where $\text{SE}_{\beta_1}$ is the standard error of the treatment coefficient.

\begin{table*}[]
    \centering
    \resizebox{\textwidth}{!}{%
\begin{tabular}{c|cccc}
        \toprule
        & Mean  & p50  & p90  & p95  \\ \midrule
        $\beta_0$ : baseline level &  6.405$^{***}$ &  6.008$^{***}$ & 6.949$^{***}$ & 7.157$^{***}$ \\
        $\beta_1$ : treatment effect &  -0.144$^{***}$ &  -0.184$^{***}$ & -0.149$^{***}$ & -0.135$^{***}$ \\
        $\beta_2$ : rps effect &  0.166$^{***}$ & 0.134$^{***}$ & 0.145$^{***}$ & 0.174$^{***}$ \\
        $R^2$ &  0.980 & 0.970 & 0.970 & 0.964\\
        Latency Impact (95\% CI) &  $[-13.81\%, -12.97\%]$ & $[-17.26\%, -16.30\%]$ & $[-14.37\%, -13.32\%]$ & $[-13.23\%, -12.03\%]$
        \\\bottomrule
    \end{tabular}%
    }
    \caption{Fixed-effects panel regression results estimating the causal impact of DLB migration. We denote by $^{***}$ p-values less than 0.01. The number of observations is 56,663. The latency impact represents the relative percentage change in latency attributable to DLB after controlling for unit heterogeneity. We adopt sum coding ($\sum_i \alpha_i=0$) so that $\beta_0$ captures the grand mean of all units. Latencies are measured in milliseconds.}
    \label{tab:regression}
\end{table*}

\begin{figure}[t]
    \centering
    \includegraphics[width=0.4\linewidth]{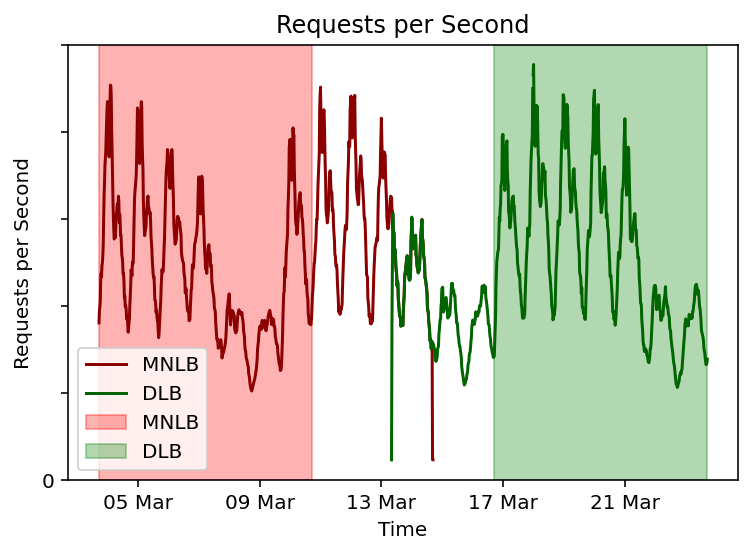}
    \caption{Requests per second for the endpoint shown in Figure~\ref{fig:endpoint-migration} during the migration from
    MNLB to DLB.}
    \label{fig:endpoint-migration-qps}
\end{figure}

\section{Simulation-based Evaluation: Experimental Protocol and Results}\label{app:simulation}

We perform a comprehensive evaluation of DLB via high-fidelity simulations, focusing on three key dimensions: resilience to network and hardware heterogeneity, adaptability to demand bursts and capacity outages, and robustness against latency function estimation errors.

\subsection{Experimental Methodology}
\label{appx:exp-method}

We describe the simulation environment, the scenarios used in our simulations, the policies we evaluate, and the simulation pipeline.

\paragraph{Production-integrated simulation.} To ensure high fidelity, our simulator integrates the production C++ codebase of DLB. This approach ensures that all control plane logic---including probing, latency estimation, and routing updates---behaves exactly as it does in production. We simulate the request arrival process using a time-varying Poisson process, which approximates the volatility observed in production workloads~\cite{robertazzi2000computer}. We model cells as collections of inference servers employing the disaggregated prefill-decode architecture (Figure~\ref{fig:inference}).
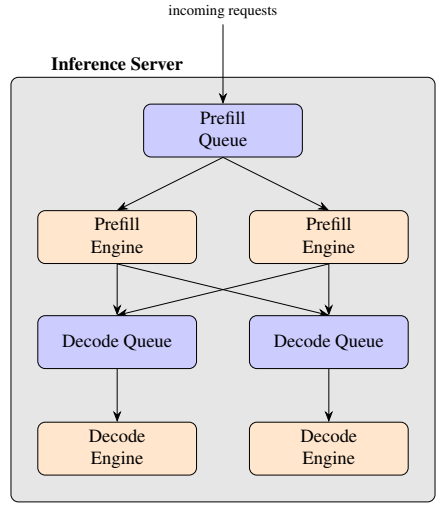
\begin{figure}
\centering
\scalebox{0.7}{
\begin{tikzpicture}[
    node distance=1.5cm and 2cm,
    component/.style={rectangle, draw, minimum width=3cm, minimum height=1cm, align=center,rounded corners},
    router/.style={component, fill=blue!20},
    client/.style={component, fill=orange!20},
    layer/.style={draw, rounded corners, inner sep=0.5cm,fill=black!10}
]

\begin{scope}[node distance=1cm and 1cm]
    \node[router] (pq) {Prefill\\Queue};
    \node[client, align=center, below=of pq, xshift=-2cm] (pe1) {Prefill\\Engine};
    \node[client, align=center, below=of pq, xshift=2cm] (pe2) {Prefill\\Engine};
\end{scope}

\begin{scope}[node distance=1cm and 1cm, yshift=-4cm]
    \node[router, align=center, xshift=-2cm] (dq1) {Decode Queue};
    \node[router, align=center, xshift=2cm] (dq2) {Decode Queue};
    \node[client, align=center, below=of dq1] (de1) {Decode\\Engine};
    \node[client, align=center, below=of dq2] (de2) {Decode\\Engine};
\end{scope}

\begin{pgfonlayer}{background}
\node[layer, behind path, fit=(pq) (de1) (de2)] (root_layer_box) {};
\node[above, xshift=-2cm] at (root_layer_box.north) {\textbf{Inference Server}};

\draw[-{Stealth[length=2mm]}] (pq.south) -- (pe1.north);
\draw[-{Stealth[length=2mm]}] (pq.south) -- (pe2.north);

\draw[-{Stealth[length=2mm]}] (pe1.south) -- (dq1.north);
\draw[-{Stealth[length=2mm]}] (pe2.south) -- (dq2.north);

\draw[-{Stealth[length=2mm]}] (pe1.south) -- (dq2.north);
\draw[-{Stealth[length=2mm]}] (pe2.south) -- (dq1.north);

\draw[-{Stealth[length=2mm]}] (dq1.south) -- (de1.north);
\draw[-{Stealth[length=2mm]}] (dq2.south) -- (de2.north);

\node (traffic) [above=of pq,font=\footnotesize] {incoming requests};
\draw[-{Stealth[length=2mm]}] (traffic) -- (pq.north);
\end{pgfonlayer}
\end{tikzpicture}
}
\caption{Architecture of a disaggregated inference server. The server separates processing into prefill and decode stages. Incoming requests first enter a shared prefill queue servicing multiple parallel prefill engines. Once the prefill stage is complete, the request and its KV cache are transferred to a selected decode engine.}
\label{fig:inference}
\end{figure}

We run our simulations on the Borg cluster management system. Crucially, because the request router's control loops (probing, state estimation) operate on periodic cycles off the critical request path, we execute the simulation in real time (wall-clock time). While this constrains the maximum scale of our experiments, it guarantees that the temporal dynamics between asynchronous control updates and request handling are captured with precision.

\paragraph{Scenario generation.} We employ a hierarchical sampling strategy for the cells to ensure that the simulation reflects the co-located nature of cells in production environments. We group real-world Google cell locations by geographical region, select a random subset of regions, and then randomly sample cells within those regions to yield a total of 2 to 10 cells. Network latencies are based on historical inter-cell round-trip times.

In each cell, the number of inference servers is drawn from a Poisson distribution with mean 10. Each inference server has a disaggregated architecture with one prefill engine and two decode engines with four parallel slots per decode engine. The time required to process a token in each stage is randomly drawn for each cell (but then fixed for the scenario) from a log-normal distribution to replicate the heterogeneity observed in production due to different hardware, with a 20\% standard deviation across cells. Each request has a random number of prefill and decode tokens drawn from different shifted Gamma distributions, which approximate actual workloads well.

We calculate the theoretical request processing rate for each cell based on its replica count (i.e., the number of inference servers), its token processing speed, and the mean prefill and decode token counts of a representative workload distribution.  We sum these rates to obtain a total system processing rate, which we then multiply by a target utilization factor of 80\% to determine the total demand. This total demand is allocated among the cells by sampling uniformly from the probability simplex. %

\paragraph{Baselines and Policies.} We evaluate \rif-based routing (\rifrouting), latency-based routing (\latencyrouting), and gradient-based routing (\gradientrouting) using discrete (\discreterouting) or flow routing mechanisms (\flowrouting). As suggested by our theoretical results in Section~\ref{sec:theory}, we choose the stepsize for flow routing algorithms to be inversely proportional to the mean network latency. We compare our algorithms against two benchmarks: weighted random routing (\randomrouting), which routes requests proportionally to the number of replicas in each cell~\cite{hajek1985extremal}, and the centralized legacy load balancing algorithm (\gslbrouting) described in Section~\ref{sec:old-arch}.
Because \randomrouting is integrated into the DLB stack, it naturally adapts to capacity outages by adjusting replica counts and utilizes our error aversion mechanism to avoid unhealthy cells.

\paragraph{Experimental protocol.} For each experiment, we draw 10 scenarios at random using the generative process described above. Each scenario is simulated for a duration of 10 minutes. %
We run 5 trials per scenario to average out stochastic variations in arrival and service processes. For each scenario-policy tuple, we compute mean and p90 end-to-end latencies (network + serving latencies) of successful requests, and error rates from inference servers' queues overflowing. To facilitate comparison across heterogeneous scenarios, we report the \emph{normalized latency}, defined as the ratio of the algorithm's latency to that of the random routing baseline. We present selected results in Figure~\ref{fig:simulation_body}; the remaining results are presented in Appendix~\ref{app:simulation}.

\subsection{Impact of Network and Hardware Heterogeneity}

We use the scenarios described above as baselines and then explore different variations such as higher/lower utilization levels, longer/shorter requests, and different levels of heterogeneity across cells. Figures~\ref{fig:mean_latency}, \ref{fig:p90_latency}, \ref{fig:error} (in the appendix) present mean latency, p90 latency results, and error rates, respectively, for different experiments.

The baseline scenarios have an $80\%$ utilization and serving latencies on the order of hundreds of milliseconds to seconds (see results in Figure~\ref{fig:mean_latency_main}). Most DLB policies show improvements around 20\% relative to random routing, with \textbf{Flow-Grad} performing the best. These results confirm that our stateful algorithms substantially outperform the random routing (\randomrouting) baseline. While \textbf{Flow-Grad} yields the best performance, it is notable that \rif-based routing—which ignores hardware speed differences—achieves competitive results. At higher utilization, the system is forced to fully saturate all available capacity; thus, optimal routing involves distributing load across all cells regardless of their individual speeds. Finally, our policies outperform \randomrouting and \gslbrouting by actively balancing requests to mitigate the local queue buildups that drive tail latency.

Figure~\ref{fig:mean_latency_utilization} shows how performance varies with the utilization level. Naturally, the relative improvements of our stateful policies over \randomrouting and \gslbrouting are smaller when utilization is low, as the queuing delays we avoid are less severe. This regime highlights a critical limitation of \rif-based routing. When the system is lightly loaded, the optimal strategy is to route exclusively to faster or closer cells (depending on query cost). \discreterouting-\rifrouting and \flowrouting-\rifrouting fail to make this distinction, resulting in sub-optimal performance by spilling requests to slower cells despite the availability of faster alternatives. When utilization is high, our policies yield substantially lower latencies and error rates compared with \gslbrouting (and \randomrouting) in the evaluated scenarios.

We also explore scenarios with different processing times in which we divide the token processing time by a factor of $y$ and multiply the request rate by the same factor $y$ to keep utilization constant. Figure~\ref{fig:mean_latency_speed} shows results for $y \in \{1/10, 1/4, 1/2, 1, 2, 4, 10\}$. When requests are short ($y=10$), network latency has an outsized impact and \rif-based routing performs poorly. \gslbrouting, which was optimized for these sand-type requests, performs remarkably well, beaten only by \flowrouting-\gradientrouting. For longer processing times ($y\leq1$), our algorithms improve upon \gslbrouting.

Finally, we consider homogeneous cells with equal hardware (see Figure~\ref{fig:mean_latency_main}). Because serving latencies are similar across cells and network latencies are small, routing according to the number of replicas in each cell provides good performance when utilization is not too high and none of the algorithms performs drastically better than \randomrouting. %

\subsection{Demand Bursts and Capacity Outages}

To study the impact of variability we consider demand bursts and capacity outages. We report mean latency, p90 latency, and error rate in Figures~\ref{fig:mean_latency_main}, \ref{fig:p90_latency_main} and~\ref{fig:error_main} in the appendix, respectively.

For the demand burst experiment, we pick a cell at random and increase its demand so that overall system utilization increases to $95\%$. The disruption begins $3$ minutes into the $10$-minute simulation and lasts for $4$ minutes. Our policies significantly improve upon \randomrouting as they can better distribute requests across cells and avoid queue buildups. The long cycle times of \gslbrouting lead to higher latency compared with our policies. Moreover, \gslbrouting and \randomrouting have higher error rates, caused by queues overflowing in the inference servers (Figure~\ref{fig:error_main}).

We simulate a capacity outage where multiple inference servers go offline. As before, the disruption begins after $3$ minutes and lasts for $4$ minutes. To model this capacity outage, we iteratively reduce the replica counts of randomly selected cells to a single unit until the total system utilization goes over $95\%$. The performance advantage of our policies over \gslbrouting is slightly narrower in this specific experiment because \gslbrouting is capacity-aware.
That said, \randomrouting and \gslbrouting lead to higher tail latencies and error rates.

\begin{figure*}
    \centering
    \subcaptionbox{Baseline and others.\label{fig:mean_latency_main}}{\includegraphics[height=5cm]{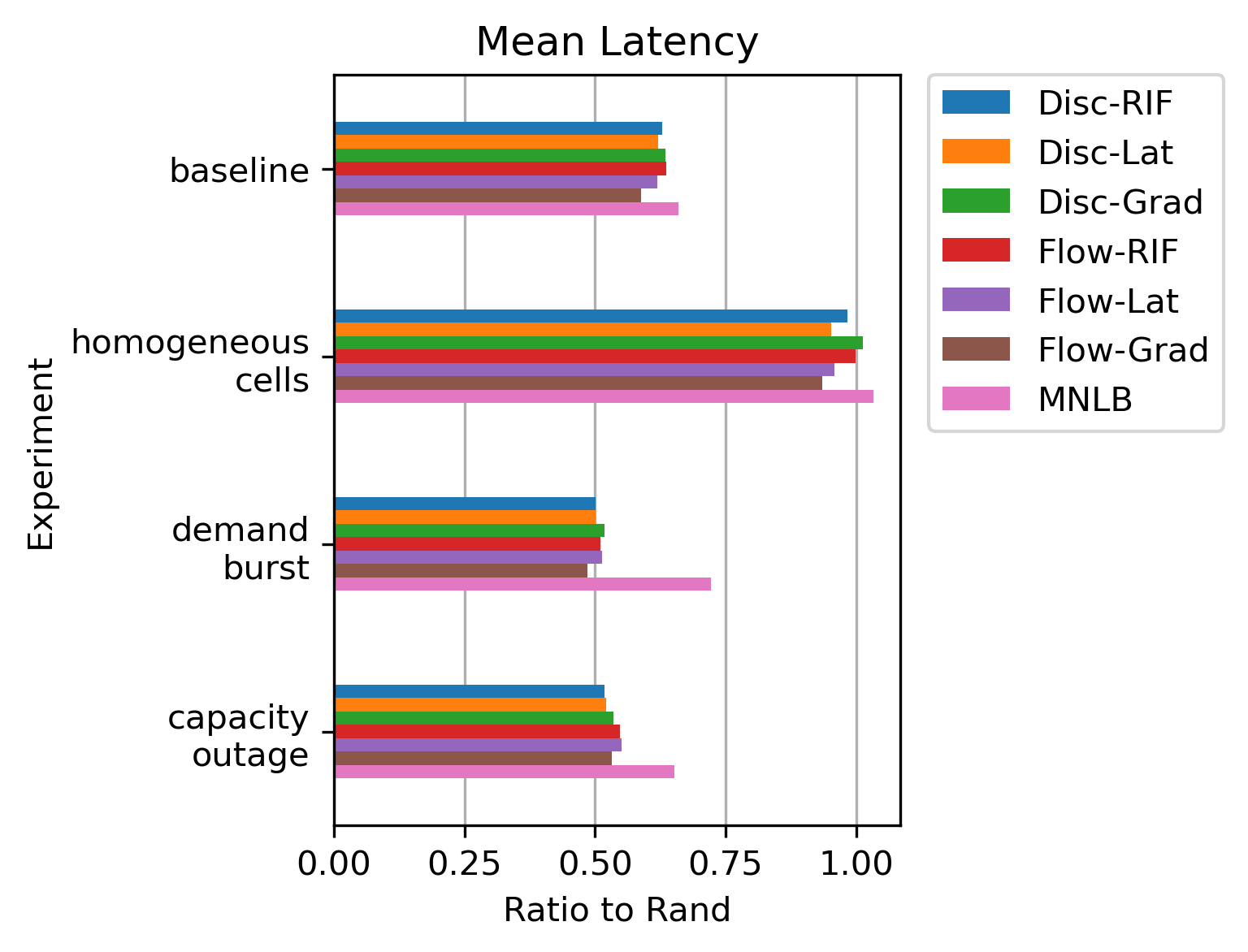}}
    \subcaptionbox{Impact of utilization.\label{fig:mean_latency_utilization}}{\includegraphics[height=5cm]{figures/utilization_disagg_variant_mean_latency.png}}
    \\
    \subcaptionbox{Impact of processing times.\label{fig:mean_latency_speed}}{\includegraphics[height=5cm]{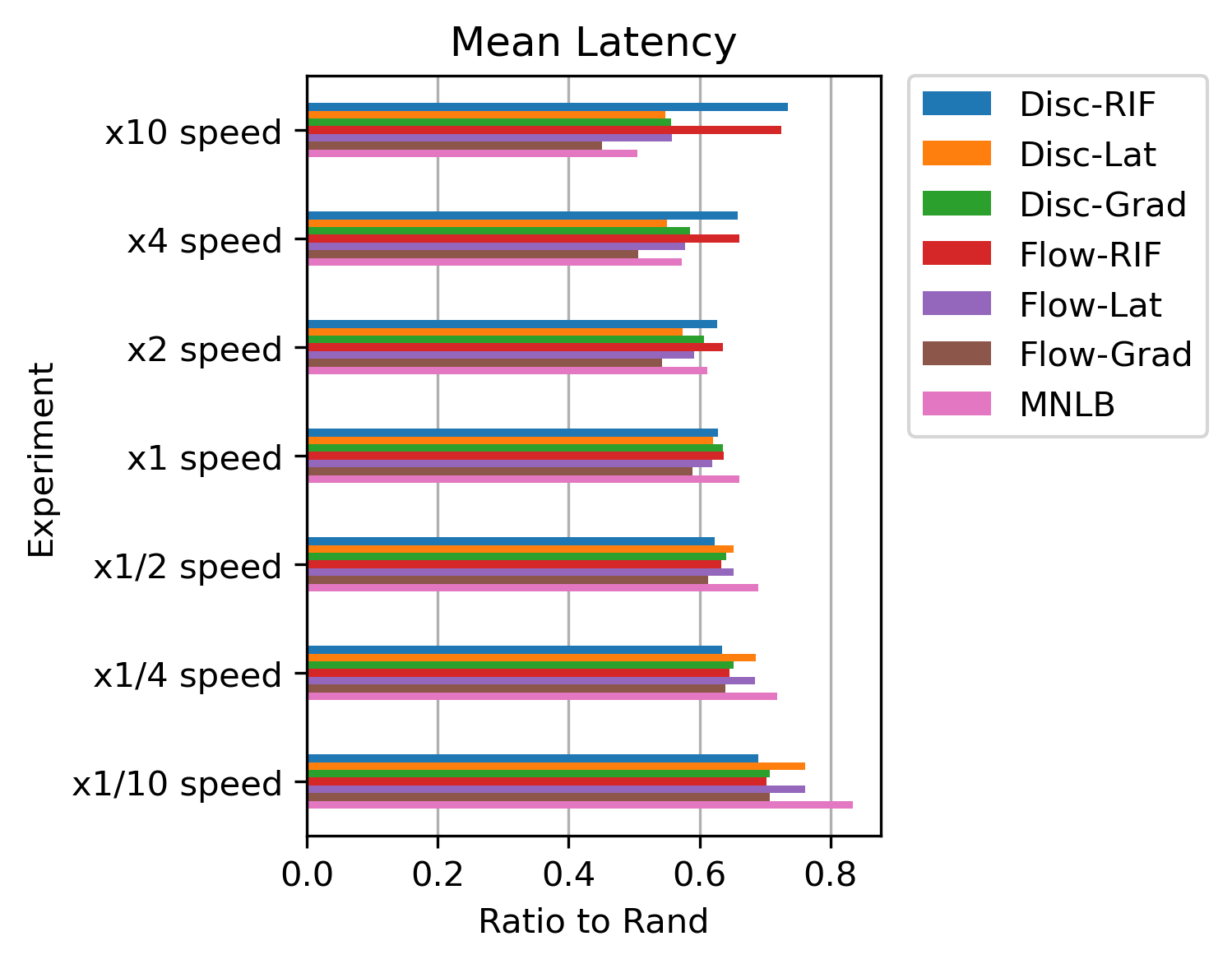}}
    \subcaptionbox{Impact of estimation errors.\label{fig:perturbation_mean_latency}}{\includegraphics[height=5cm]{figures/perturbation_disagg_variant_mean_latency.png}}
    \caption{Mean latency of different routing algorithms compared to random routing across different experiments.}
    \label{fig:mean_latency}
\end{figure*}

\begin{figure}[h]
    \centering
    \subcaptionbox{Baseline and others.\label{fig:p90_latency_main}}{\includegraphics[height=5cm]{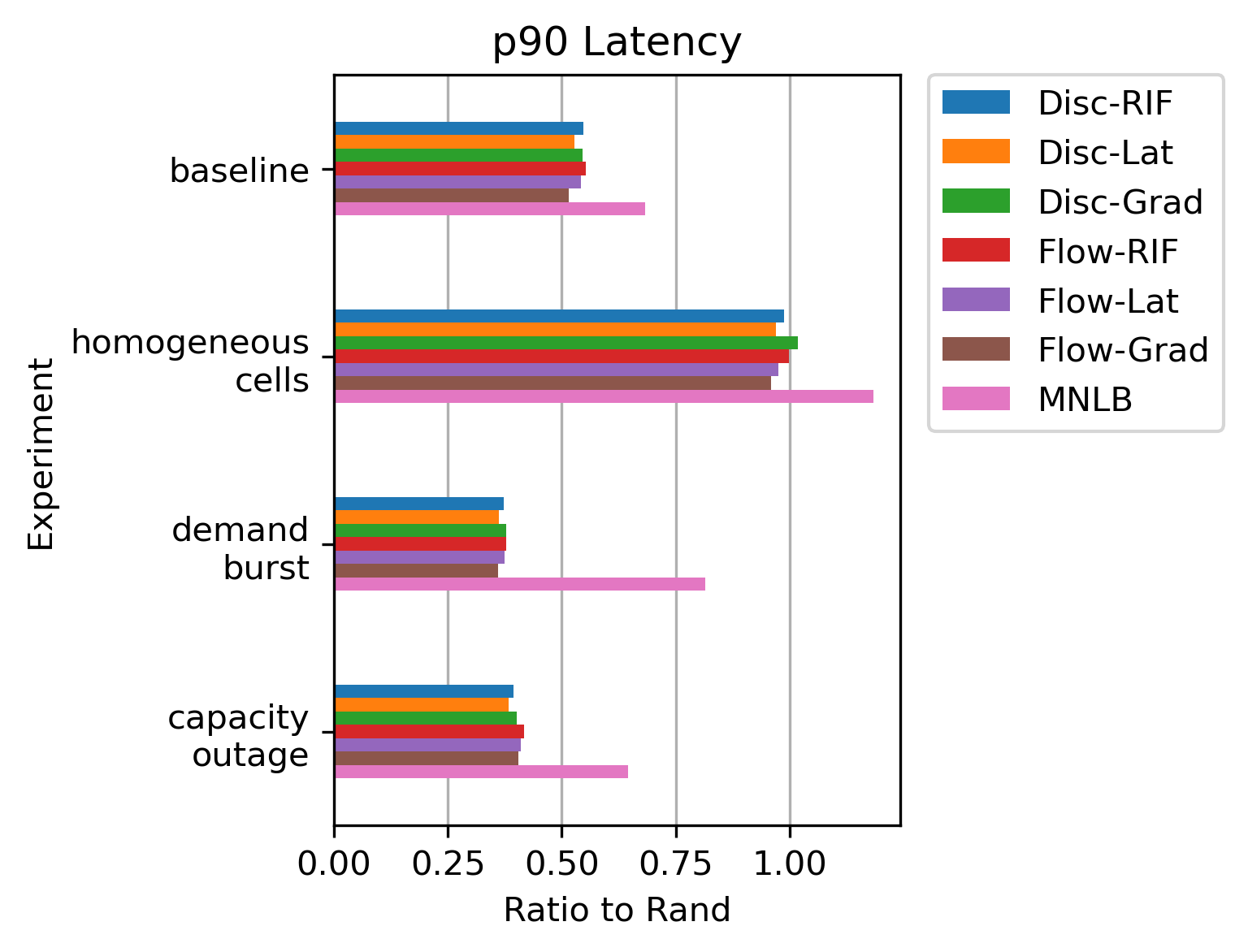}}
    \subcaptionbox{Impact of utilization.\label{fig:p90_latency_utilization}}{\includegraphics[height=5cm]{figures/utilization_disagg_variant_p90_latency.png}}
    \\
    \subcaptionbox{Impact of processing times.\label{fig:p90_latency_speed}}{\includegraphics[height=5cm]{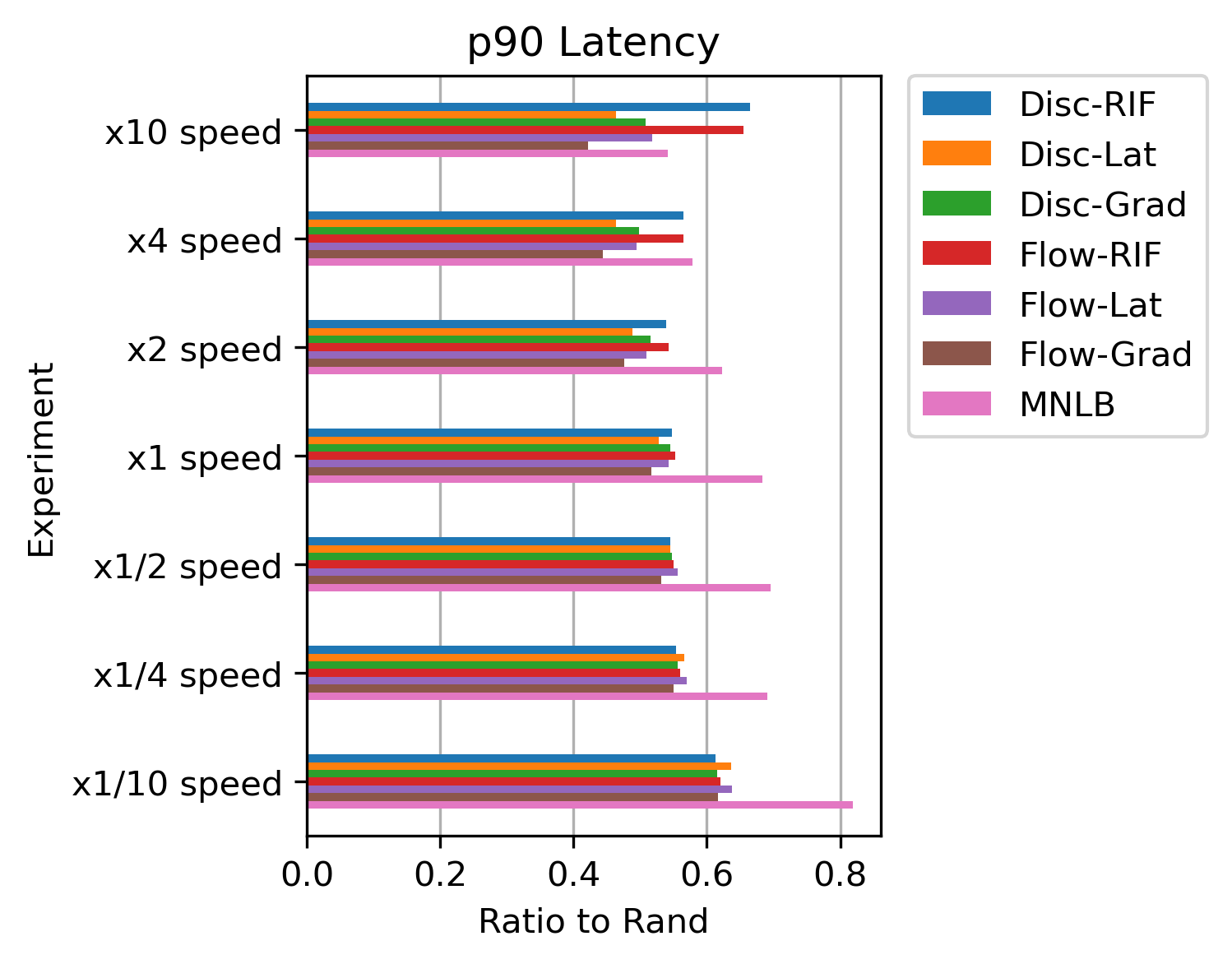}}
    \subcaptionbox{Impact of estimation errors.\label{fig:perturbation_p90_latency}}{\includegraphics[height=5cm]{figures/perturbation_disagg_variant_p90_latency.png}}
    \caption{p90 latency of different routing algorithms compared to random routing across experiments.}
    \label{fig:p90_latency}
\end{figure}

\begin{figure}[h]
    \centering
    \subcaptionbox{Baseline and others.\label{fig:error_main}}{\includegraphics[height=5cm]{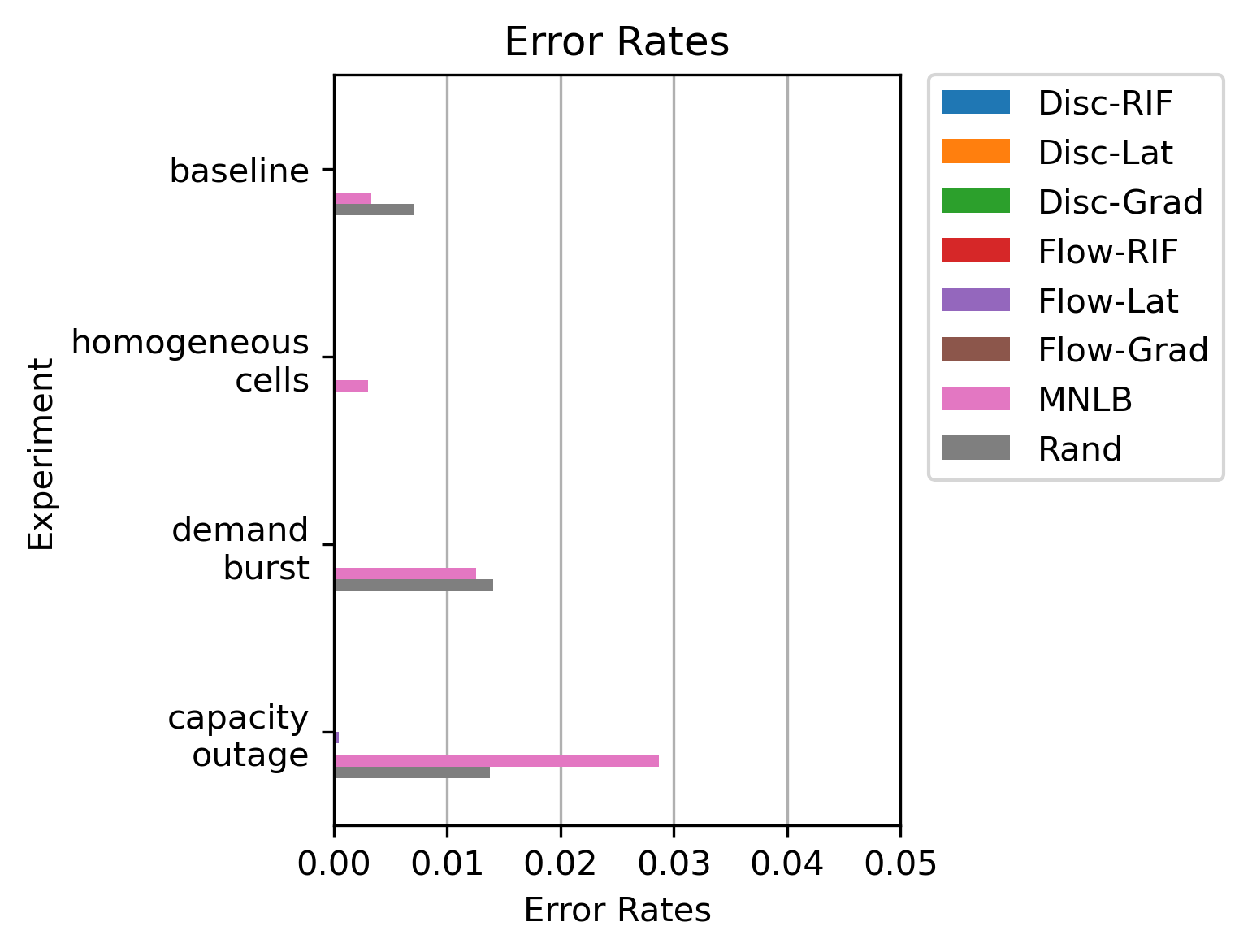}}
    \subcaptionbox{Impact of utilization.\label{fig:error_utilization}}{\includegraphics[height=5cm]{figures/utilization_disagg_variant_error.png}}
    \\
    \subcaptionbox{Impact of processing times.\label{fig:error_speed}}{\includegraphics[height=5cm]{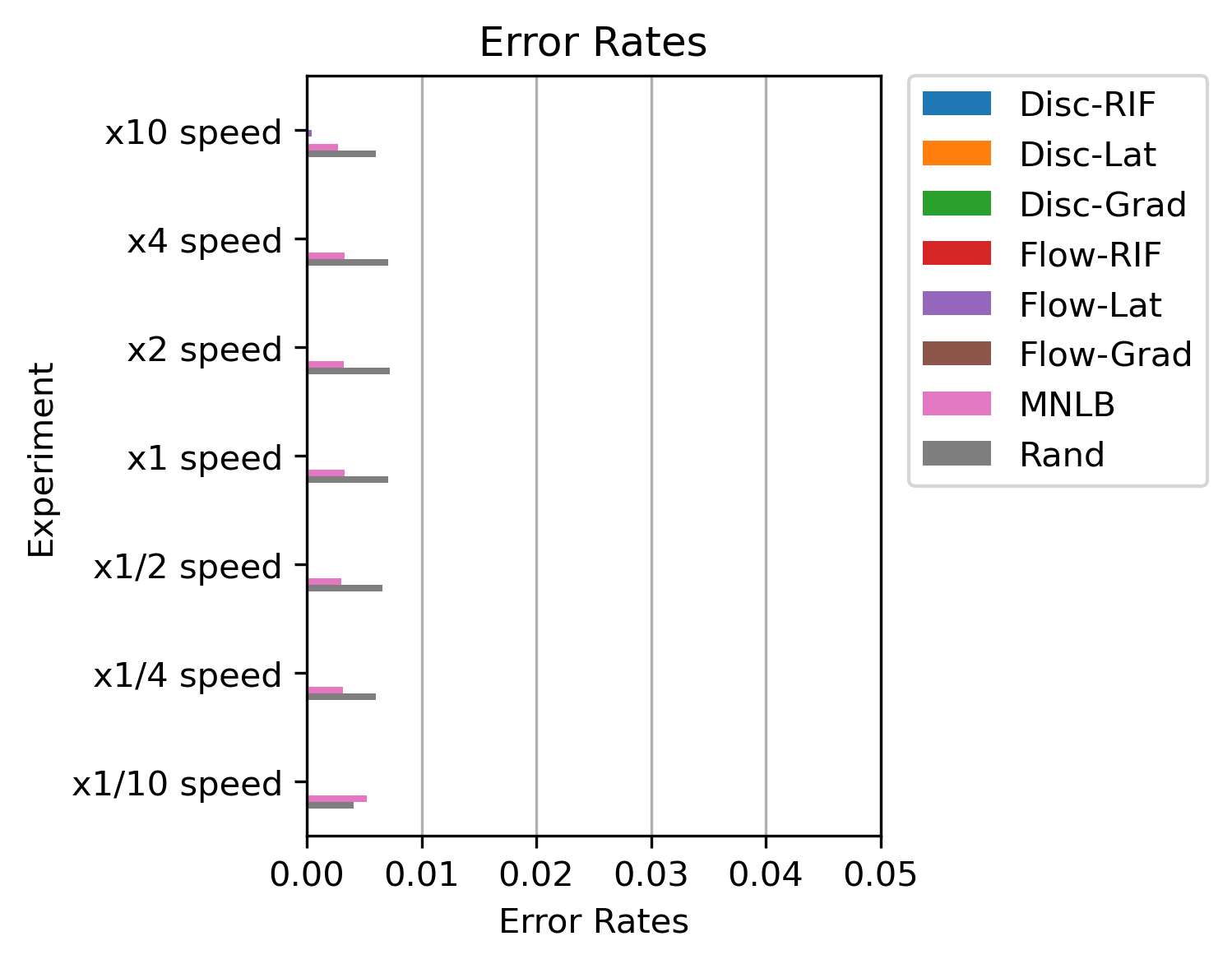}}
    \subcaptionbox{Impact of estimation errors.\label{fig:perturbation_error}}{\includegraphics[height=5cm]{figures/perturbation_disagg_variant_error.png}}
    \caption{Error rates of different routing algorithms across experiments.}
    \label{fig:error}
\end{figure}

\section{Details of Theoretical Analysis} \label{app:theory}
\subsection{Continuous-Time Formulation}\label{app:flow-fluid}

In the fluid model, the workload of cell $j \in \DS$ at time $t > 0$ evolves according to the following delay differential equation:
\begin{equation}\label{eq:dynamics-workloads}
\begin{split}
     \frac {d} {dt} N_j(t) &= \sum_{i\in\DS^-(j)} \lambda_i x_{ij}(t - \tau_{ij}) - \mu_j (N_j(t))\,.
\end{split}
\end{equation}
The first term captures the inflow of requests arriving to the cell, which is determined by the routing decisions and network latencies, and the second term gives the outflow of requests as determined by the processing rate function of the cell.

To fully specify the dynamics of the workloads, we need to determine the evolution of the routing probabilities $x_{ij}(t)$. In the case of discrete routing, the routing probabilities are determined using equation~\eqref{eq:discrete-routing}. In the rest of this section, we describe the continuous-time limit of flow routing as the time between updates converges to zero.  Continuous-time limits of recursive optimization algorithms have received considerable attention in the last decades as they provide tractable approximations using ordinary differential equations (see, e.g., \cite{schropp2000dynamical,su2016differential}). The exposition here follows \cite{balseiro2025load}.

We define $T_{\Delta_i}(\bx_i)$ to be the tangent cone of $\Delta_i$ at $\bx_i$, which is given by
\[
 T_{\Delta_i}(\bx_i) = \left\{\bv \in \mathbb{R}^{|\DS|}: \sum_{j\in \DS} v_j = 0,  v_j \geq 0  \text{ if } x_{ij}  = 0, v_j = 0 \text{ for all } j \not\in \DS^+(i) \right\}\,.
 \]
The tangent cone captures directions along which the cell can update the routing probabilities while maintaining feasibility. The components of a feasible direction $\bv \in T_{\Delta_i}(\bx_i)$ should sum up to zero to satisfy the constraint that probabilities sum up to one. Moreover, for cells whose probabilities are at zero, the corresponding component of the direction should be non-negative to preserve the non-negativity constraint.

Equation \eqref{eq:flow-routing} gives the discrete update of routing probabilities in flow routing when the difference between time updates is $\delta t>0$. Subtracting $\bx_i(t)$ on both sides, dividing by $\delta t$ and taking the limit $\delta t \downarrow 0$ we obtain the differential equation
\begin{equation}\label{eq:flow routing pds}
\frac{d}{dt}\bx_i(t)= \Pi_{T_{\Delta_i}(\bx_i(t))}\!\big(-\eta_i \cdot \bc_i(t)\big)
\quad\text{with}\quad
\bc_i(t) =\operatorname{ext}_0\!\left(
       \big(c_{ij}(N_j(t-\tau_{ij}))\big)_{j\in\DS^+(i)}\right).
\end{equation}
where $\Pi_{T_{\Delta_i}(\bx_i)}(\cdot)$ is Euclidean projection onto the tangent cone, and $\operatorname{ext}_0$ pads the outgoing-arc cost vector with zeros on coordinates $j\notin\DS^+(i)$.
This padding is only a dimensional convention that embeds the cost vector in the same ambient space $\mathbb R^{|\DS|}$ as $\bx_i$; it does not make non-arcs available at zero cost, because the definition of $\Delta_i$ fixes $x_{ij}=0$ for $j\notin\DS^+(i)$ and its tangent cone fixes the corresponding velocities at zero.

 Throughout, we consider feasible solutions of \eqref{eq:dynamics-workloads}--\eqref{eq:flow routing pds}
that are absolutely continuous for $t\geq0$, with $\bN(t)\in\mathbb R_{\geq0}^{n}$ and $\bx_i(t)\in\Delta_i$.  To initialize the workload dynamics, a feasible solution also includes arbitrary Lebesgue-measurable routing values $\bx_i(t)\in\Delta_i$ for $t\in[-\bar\tau,0]$, agreeing with the initial routing vector at zero.

\begin{assumption}[Admissible workload history]
\label{assume:fsr-history-draft}
The prescribed workload history $\bN:[-\bar\tau,0]\to\mathbb R_{\geq0}^{n}$ is absolutely continuous.  With
\begin{equation*}
    v=\max_{j\in\DS}\max\left\{
       \sum_{i\in\DS^-(j)}\lambda_i,\bar\mu_j\right\},
\end{equation*}
we assume $\left|\frac{d}{dt}N_j(t)\right|\leq v$ almost everywhere on $[-\bar\tau,0]$.
\end{assumption}

The workload equation itself implies the same bound for nonnegative times: both its arrival and service terms lie in $[0,v]$, and hence \begin{equation}
    |N_j(t)-N_j(s)|\leq v|t-s|,\quad s,t\geq-\bar\tau.
    \label{eq:fsr-workload-speed-draft}
\end{equation}

The following lemma establishes that stationary points of the dynamics
satisfy the definition of equilibrium points in Definition~\ref{defn:eq}.

\begin{lemma}\label{lemma:equilibrium-points} Suppose Assumptions~\ref{assume:fsr-service-draft} and~\ref{assume:fsr-cost-draft}  hold, and consider either discrete routing or flow routing.
If there exists a point $(\bN^*,\bx^*)$ and a time $t > 0$ such that for all arcs $(i,j) \in \AS$ and times $s\in[t - \tau_{ij}, t]$ we have $N_j(s)= N_j^*$ and $x_{ij}(s)= x_{ij}^*$, then $(\bN^*,\bx^*)$ is an equilibrium point (Definition~\ref{defn:eq}).
\end{lemma}
\begin{proof}
We argue that stationary solutions are equilibrium points. Fix a stationary solution $(\bN^*,\bx^*)$ and a time $t > 0$ such that for all arcs $(i,j) \in \AS$ and times $s\in[t - \tau_{ij}, t]$ we have $N_j(s)= N_j^*$ and $x_{ij}(s)= x_{ij}^*$. We need to show that $(\bN^*,\bx^*)$ satisfies flow balance and complementary slackness.

For flow balance, plugging the stationary trajectory into~\eqref{eq:dynamics-workloads} yields, for each $j\in\DS$,
\[
0 = \frac{d}{dt}N_j(t)
= \sum_{i\in\DS^-(j)} \lambda_i x_{ij}(t-\tau_{ij}) - \mu_j(N_j(t))
=\sum_{i\in\DS^-(j)} \lambda_i x_{ij}^* - \mu_j(N_j^*).
\]
Hence $\sum_{i\in\DS^-(j)} \lambda_i x_{ij}^* = \mu_j(N_j^*)$ for all $j$, i.e., the flow balance condition in Definition~\ref{defn:eq} is satisfied.

\paragraph{Part 1 (Discrete routing).} We argue that the complementary slackness condition holds for discrete routing. Under discrete routing, for each $i\in\DS$ and each $t$, the routing vector satisfies
\[
\bx_i(t)\in \arg\min_{\bz\in\Delta_i} \sum_{j\in\DS^+(i)} c_{ij}\!\left(N_j(t-\tau_{ij})\right)z_j = \arg\min_{\bz\in\Delta_i} \sum_{j\in\DS^+(i)} c_{ij}(N_j^*)z_j,
\]
where the last equation follows from stationarity. Let $c_i = \min_{j\in\DS^+(i)} c_{ij}(N_{j}^*)$ be the lowest cost observed by cell $i\in\DS$. By definition, we must have $c_{ij}(N_j^*)\ge c_i$ for all $j\in\DS^+(i)$ with equality whenever $x_{ij}^*>0$, implying the complementary-slackness condition in Definition~\ref{defn:eq}. Combining the two parts, $(\bN^*,\bx^*)$ is an equilibrium point in the case of discrete routing.

\paragraph{Part 2 (Flow routing).} We next move to flow routing. At the stationary solution, we have that $d\bx_i/dt(t) = 0$ for every $i\in \DS$ and \eqref{eq:flow routing pds} gives
\[
\boldsymbol{0} = \Pi_{T_{\Delta_i}(\bx_i^*)}\!\big(-\eta_i \bc_i^*\big) \quad \text{with} \quad \bc_i^* := \operatorname{ext}_0\!\left((c_{ij}(N_j^*))_{j\in\DS^+(i)}\right)\,.
\]
By Lemma~4 in \cite{balseiro2025load}, there exists some $c_i$ such that $c_{ij} = c_i$ for all $(i,j) \in \AS$ with $x_{ij}^* > 0$ and $c_{ij} \ge c_i$ otherwise. The result follows.
\end{proof}

\subsection{Proof of Lemma~\ref{lemma:existence-eq-points}}
\label{app:equilibrium-points}

\begin{proof}[Proof of Lemma~\ref{lemma:existence-eq-points}]
Before proving existence and uniqueness, we give a variational representation of finite equilibria.
Because $\mu_j$ is continuous, bounded, and strictly increasing with $\mu_j(0)=0$, it maps $[0,\infty)$ bijectively onto $[0,\bar\mu_j)$.  Its inverse is continuous and strictly increasing on this interval, and $\mu_j^{-1}(s)\to\infty$ as $s\uparrow\bar\mu_j$.  The composition $f_j\circ\mu_j^{-1}$ is therefore continuous and strictly increasing. Define
\[
p_j(z):=
\begin{cases}
\displaystyle\int_0^z f_j\big(\mu_j^{-1}(s)\big)\,ds,
   & 0\leq z\leq\bar\mu_j,\\[2mm]
+\infty, & z>\bar\mu_j,
\end{cases}
\]
which is a lower-semicontinuous convex function, strictly convex on $[0,\bar\mu_j)$ and differentiable with $p_j'(s)=f_j\big(\mu_j^{-1}(s)\big)$, which is the workload component of the cost of a cell carrying inflow $s$.  Define the Beckmann, McGuire, and Winsten
potential
\begin{equation}
\Psi(\bx):=\sum_{j\in\DS}p_j\left(\sum_{i\in\DS^-(j)}\lambda_i x_{ij}\right)
 +\sum_{(i,j)\in\AS}\lambda_i g_{ij}x_{ij}.
\label{eq:closed-potential-program-objective}
\end{equation}
Consider the feasible flows $\mathcal X = \{x_{ij}\ge 0: \sum_{j\in\DS^+(i)}  x_{ij}= 1,x_{ij}=0\ \text{if }(i,j)\notin\AS\}$. The following lemma establishes the equivalence.
\begin{lemma}\label{lemma:varitional-representation}
A finite pair $(\bN^*,\bx^*)$ is an equilibrium if and only if $\bx^*$ solves
\begin{align}
    \min_{\bx \in \mathcal X} \Psi(\bx) \label{eq:potential-program}
\end{align}
and, for every $j$, its induced inflow satisfies $\sum_{i\in\DS^-(j)}\lambda_i x^*_{ij}<\bar\mu_j$ and $N_j^*=\mu_j^{-1}(\sum_{i\in\DS^-(j)}\lambda_i x^*_{ij})$.
\end{lemma}

\begin{proof}
We first argue that an optimal solution $\bx^*$ with $N^*_j=\mu_j^{-1}(\sum_{i\in \DS^-(j)} \lambda_i x_{ij}^*)$ is an equilibrium. Flow balance in Definition~\ref{defn:eq} follows by construction because $\mu_j(N_j^*) = \sum_{i\in \DS^-(j)} \lambda_i x_{ij}^*$. We next argue that the complementary-slackness condition in Definition~\ref{defn:eq} holds. Because the optimization problem $\min_{\bx \in \mathcal X} \Psi(\bx)$ has a differentiable objective and linear constraints, the KKT conditions are necessary for optimality~\cite{bertsekas1997nonlinear}. The KKT conditions are the following. Introduce Lagrange multipliers $c_i$ for the constraints $\sum_{j\in\DS^+(i)}  \lambda_i x_{ij} = \lambda_i$ (we multiplied the constraint of each cell by its arrival rate) and $\alpha_{ij} \ge 0$ for the non-negativity constraints $x_{ij} \ge 0$. The Lagrangian of the optimization problem \eqref{eq:potential-program} is
\[
L(\bx, \boldsymbol{c}, \boldsymbol{\alpha}) = \Psi(\bx) - \sum_{i \in \DS} c_i \left(\sum_{j \in \DS^+(i)} \lambda_i x_{ij} - \lambda_i \right) - \sum_{(i,j) \in \AS} \alpha_{ij} x_{ij}\,.
\]
The first-order optimality conditions imply that the partial derivative of the Lagrangian with respect to $x_{ij}$ should be zero, which gives that
\[
    \frac{\partial \Psi (\bx^*)}{\partial x_{ij}}  - c_i^* \lambda_i - \alpha_{ij}^* = 0\,,
\]
where
\[
    \frac{\partial \Psi (\bx)}{\partial x_{ij}} = \lambda_i p_j'\left(\sum_{k \in \DS^-(j)} \lambda_k x_{kj}\right) + \lambda_i g_{ij} = \lambda_i f_j\left(\mu_j^{-1}\left(\sum_{k \in \DS^-(j)} \lambda_k x_{kj}\right)\right) + \lambda_i g_{ij}\,.
\]
If $x_{ij}^* > 0$, the KKT complementary slackness condition implies that $\alpha_{ij}^* = 0$, which gives that $\lambda_i f_j(N_j^*) + \lambda_i g_{ij} - c_i^* \lambda_i = 0$. Canceling the arrival rate $\lambda_i > 0$ and re-arranging, we obtain that $f_j(N_j^*) + g_{ij} = c_i^*$. If $x_{ij}^* = 0$, we obtain using a similar argument that $f_j(N_j^*) + g_{ij} \ge c_i^*$ because $\alpha_{ij}^* \ge 0$. The claim follows because the cost function satisfies $c_{ij}(N_j^*) = f_j(N_j^*) + g_{ij}$.

We now argue that for every equilibrium point $(\bN^*, \bx^*)$ the optimal routing probabilities $\bx^*$ are optimal for the optimization problem~\eqref{eq:potential-program}. The potential function $\Psi(\bx)$ is convex because $f_j$ and $\mu_{j}^{-1}$ are increasing, their composition is increasing and the integral is convex, together with the fact that the sum of convex functions is convex. Therefore, the KKT conditions are sufficient for optimality. The claim follows because the KKT conditions of this program are exactly the complementary-slackness conditions in Definition~\ref{defn:eq}, with multipliers $(c_i)_{i\in\DS}$, as we argued before.
\end{proof}

\paragraph{Part 1 (Existence of an equilibrium).} We establish existence by proving that the optimization problem~\eqref{eq:potential-program} has an optimal solution. Let $\bx'$ be the matrix of routing probabilities in Assumption~\ref{assume:stability}. First, note that $\Psi(\bx') < \infty$. This follows because $f_j(\mu_j^{-1}(s)) < \infty$ for all $s \in [0, z_j']$ with $z_j' = \sum_{i\in\DS^-(j)} \lambda_i x_{ij}'$ since $\mu_j^{-1}(z_j') = N_j' < \infty$ and $f_j \circ \mu_j^{-1}$ is monotonically increasing. Because monotonically increasing functions are Riemann integrable, we conclude that $p_j(z_j') < \infty$, proving the claim. By the monotone convergence theorem, we have that $p_j(z_j)$ is lower semicontinuous. Because composition with a continuous function preserves lower semicontinuity, we conclude that $\Psi(\bx)$ is lower semicontinuous. The Weierstrass theorem~\cite{luenberger1997optimization} implies that $\min_{\bx \in \mathcal X} \Psi(\bx)$ admits an optimal solution because the objective is lower semicontinuous, and the feasible set is closed, non-empty (since $\bx' \in \mathcal X)$ and bounded. Any optimal solution must have inflows strictly below capacity, since otherwise shifting slightly toward the strictly feasible routing $\bx'$ would lower the objective by relieving cells with unbounded marginal costs. Lemma~\ref{lemma:varitional-representation} implies existence of an equilibrium point.

\paragraph{Part 2 (Uniqueness of equilibrium workloads).}  We argue that the optimization problem has unique optimal workload levels. If cost functions $\boldsymbol{c}$ are strictly increasing, then $p_j$ is strictly convex since $f_j \circ \mu_j^{-1}$ is strictly increasing. Let $z_j = \sum_{i\in\DS^-(j)} \lambda_i x_{ij}$ be the inflow to cell $j\in\DS$. Strict convexity of $p_j$ for $j \in \DS$ implies that the optimal inflow vector $\boldsymbol{z}^*$ is unique (otherwise, taking the midpoint of two solutions with different $z_j^*$ for some cell $j\in \DS$ would strictly decrease the objective value). Finally, since each $\mu_j$ is strictly increasing, $N_j^*=\mu_j^{-1}(z_j^*)$ is unique for every $j$,
so the equilibrium workload vector $\bN^*$ is unique.
\end{proof}

For the cost functions used by gradient-based routing,
a corollary of Lemma~\ref{lemma:varitional-representation} is that
equilibrium points are optimal solutions to the following centralized static routing problem:
\begin{align}\label{eq:opt}
\min_{\bx_i \in \Delta_i, \bN \in \mathbb R_{\geq 0}^{|\DS|}} &\; \sum_{j \in \DS} N_j + 2 \sum_{(i,j) \in \AS} \lambda_i x_{ij}  \tau_{ij} \tag{OPT}\\
\text{s.t.} & \; \sum_{i \in \DS^-(j)} \lambda_i x_{ij} = \mu_j(N_j) \,, \forall j \in \DS\,.\nonumber
\end{align}
The first term in the objective captures, by Little's Law, the steady state serving latency of all requests in the system while the second term measures the total network latency. The constraint imposes flow balance at the cells, i.e., the inflow of requests should be equal to the outflow processed at a cell.

 \subsection{Proof of Theorem~\ref{thm:fsr-raw-cumulative-draft}}
\label{app:global-performance}
Throughout this proof, we use the shorthand $D_j(u):=D_j(u,N_j^*)$.
For routing probabilities $\bx_i \in \Delta_i$ for cell $i \in \DS$, define the \emph{complementary slackness error}:
\begin{equation*}
    S_i(\bx_i)=\sum_{j\in\DS^+(i)}c_{ij}(N_j^*)
       \big(x_{ij}-x_{ij}^*\big).
\end{equation*}
The equilibrium complementary slackness gives $S_i(\bx_i)\geq0$ because
$\bx_i\in\Delta_i$.  In particular, the current-time gap
$S_i(\bx_i(t))$ is nonnegative.  We will also encounter the mixed-delay
quantity
\begin{equation*}
    S_i^\tau(t)=\sum_{j\in\DS^+(i)}c_{ij}(N_j^*)
       \big(x_{ij}(t-\tau_{ij})-x_{ij}^*\big).
\end{equation*}
Its coordinates can come from different times and therefore may fail to form a
vector in $\Delta_i$; consequently, $S_i^\tau(t)$ can be negative.

We will use two Lyapunov functions in our analysis.  The first involves the
routing probabilities and is given by
\begin{equation}\label{eq:fsr-routing-potential}
V(\bx)=\sum_{i\in\DS}\frac{\lambda_i}{2\eta_i}
       \|\bx_i-\bx_i^*\|_2^2.
\end{equation}
The second involves the cell workloads and is given by
\begin{equation*}
    \Phi_j(u)=\int_{N_j^*}^{u}\big(f_j(y)-f_j(N_j^*)\big)\,dy,
       \qquad \Phi(\bN)=\sum_{j\in\DS}\Phi_j(N_j).
\end{equation*}
The potential $V$ is the traffic-weighted squared routing error, scaled by the
inverse stepsizes, while $\Phi$ accumulates the marginal workload-cost
deviation from equilibrium.  Both are nonnegative and vanish only at their
respective equilibrium states.  %

Let $d_{ij}(t)=x_{ij}(t)-x_{ij}^*$ and define the \emph{cost--routing delay
mismatch}
\begin{equation}
    \Gamma_{ij}(t)=
       c_{ij}(N_j(t))d_{ij}(t-\tau_{ij})
       -c_{ij}(N_j(t-\tau_{ij}))d_{ij}(t),
    \label{eq:fsr-Gamma-draft}
\end{equation}
which compares current cost paired with delayed routing against
delayed cost paired with current routing.  It vanishes when $\tau_{ij}=0$.
The following fundamental lemma bounds the instantaneous drift of the
Lyapunov functions in terms of the workload optimality gap, mixed-delay
complementary slackness error, and cost--routing delay mismatch.

\begin{lemma}[Composite Lyapunov drift]
\label{lem:fsr-composite-drift-draft}
For almost every $t\geq0$,
\begin{equation}
\begin{aligned}
\frac{d}{dt}V(\bx(t))+\frac{d}{dt}\Phi(\bN(t))
\leq-\sum_{j\in\DS}D_j(N_j(t))-\sum_{i\in\DS}\lambda_iS_i^\tau(t)+\sum_{(i,j)\in\AS}\lambda_i\Gamma_{ij}(t),
\end{aligned}
    \label{eq:fsr-composite-drift-draft}
\end{equation}
\end{lemma}

\begin{proof}
Along the absolutely continuous feasible solution, continuity of each $f_j$
makes $\Phi$ continuously differentiable.  Hence the chain rule applies almost
everywhere.  Fix a time $t\geq0$ at which the derivatives exist and the
dynamics hold.
Differentiating the routing-probability Lyapunov function gives
\begin{align}
\frac{d}{dt}V(\bx(t))
&=\sum_{i\in\DS}\frac{\lambda_i}{\eta_i}
   \sum_{j\in\DS^+(i)}
   \big(x_{ij}(t)-x_{ij}^*\big)\frac{d}{dt}x_{ij}(t) \nonumber\\
&\leq-\sum_{i\in\DS}\lambda_i
   \sum_{j\in\DS^+(i)}
   \big(x_{ij}(t)-x_{ij}^*\big)
   c_{ij}(N_j(t-\tau_{ij})) \nonumber\\
&=-\sum_{(i,j)\in\AS}\lambda_i
   c_{ij}(N_j(t-\tau_{ij}))d_{ij}(t).
    \label{eq:fsr-V-drift-draft}
\end{align}
For the inequality, apply the first projection inequality in
Lemma~\ref{lem:fsr-projection-draft} with
$\bz_i=-\eta_i\bc_i(t)$ and
$\bv_i=\frac{d}{dt}\bx_i(t)$. The equality follows from the definition of
$d_{ij}$. %

We next differentiate the workload Lyapunov function.  The chain rule and \eqref{eq:dynamics-workloads} give
\begin{align}
\frac{d}{dt}\Phi(\bN(t))
&=\sum_{j\in\DS}
   \big(f_j(N_j(t))-f_j(N_j^*)\big)\frac{d}{dt}N_j(t) \nonumber\\
&=\sum_{j\in\DS}
   \big(f_j(N_j(t))-f_j(N_j^*)\big)
\cdot\left(\sum_{i\in\DS^-(j)}\lambda_i
      x_{ij}(t-\tau_{ij})-\mu_j(N_j(t))\right) \nonumber\\
&=\sum_{j\in\DS}
   \big(f_j(N_j(t))-f_j(N_j^*)\big)
\cdot\big(\mu_j(N_j^*)-\mu_j(N_j(t))\big) \nonumber\\
&\quad+\sum_{j\in\DS}
   \big(f_j(N_j(t))-f_j(N_j^*)\big)
\cdot\sum_{i\in\DS^-(j)}\lambda_i
   \big(x_{ij}(t-\tau_{ij})-x_{ij}^*\big).
    \label{eq:fsr-P-drift-first-draft}
\end{align}
The last equality adds and subtracts the equilibrium inflow
$\sum_{i\in\DS^-(j)}\lambda_i x_{ij}^*=\mu_j(N_j^*)$ from
\eqref{eq:equilibrium-flow-balance}.  By the definition of $D_j$, the first
sum on the right-hand side is $-\sum_{j\in\DS}D_j(N_j(t))$; see
\eqref{eq:semimetric-workload}.

For the remaining sum, exchange the order of summation and use the definition of $d_{ij}$ to obtain
\begin{align*}
&\sum_{j\in\DS}
   \big(f_j(N_j(t))-f_j(N_j^*)\big)
   \sum_{i\in\DS^-(j)}\lambda_i
   \big(x_{ij}(t-\tau_{ij})-x_{ij}^*\big) \nonumber\\
&=\sum_{i\in\DS}\lambda_i\sum_{j\in\DS^+(i)}
   \big(f_j(N_j(t))-f_j(N_j^*)\big)
   d_{ij}(t-\tau_{ij}) \nonumber\\
&=\sum_{i\in\DS}\lambda_i\sum_{j\in\DS^+(i)}
   \big(f_j(N_j(t))+g_{ij}\big)d_{ij}(t-\tau_{ij}) \nonumber\\
&\quad-\sum_{i\in\DS}\lambda_i\sum_{j\in\DS^+(i)}
   \big(f_j(N_j^*)+g_{ij}\big)d_{ij}(t-\tau_{ij}) \nonumber\\
&=\sum_{(i,j)\in\AS}\lambda_i
   c_{ij}(N_j(t))d_{ij}(t-\tau_{ij})
   -\sum_{i\in\DS}\lambda_i S_i^\tau(t).
\end{align*}
The second equality adds and subtracts $g_{ij}$, and the last equality uses
the separable-cost identity~\eqref{eq:fsr-raw-cost-draft} together with the
definition of $S_i^\tau(t)$.  Substituting this identity into
\eqref{eq:fsr-P-drift-first-draft} yields
\begin{equation}
\begin{aligned}
\frac{d}{dt}\Phi(\bN(t))
&=-\sum_{j\in\DS}D_j(N_j(t))
-\sum_{i\in\DS}\lambda_iS_i^\tau(t)
+\sum_{(i,j)\in\AS}\lambda_i
   c_{ij}(N_j(t))d_{ij}(t-\tau_{ij}).
\end{aligned}
    \label{eq:fsr-P-drift-final-draft}
\end{equation}

Finally, add~\eqref{eq:fsr-V-drift-draft} and
\eqref{eq:fsr-P-drift-final-draft} and use the definition~\eqref{eq:fsr-Gamma-draft} to obtain the result.
\end{proof}

\paragraph{\textbf{Roadmap}.}
We begin by providing a summary of the main steps of our analysis.  If network
delays were zero, then $S_i^\tau(t)=S_i(\bx_i(t))\geq0$ and
$\Gamma_{ij}(t)=0$.  Lemma~\ref{lem:fsr-composite-drift-draft} would then
imply that the drift of the Lyapunov functions is nonpositive.  Because of
heterogeneous delays, however, the mixed-delay routing gap might not be
nonnegative and the cost--routing mismatch might not vanish.  The main
challenge is to control these two terms after integration.

The main steps of our analysis are as follows:
\begin{enumerate}
\item In Subsection~\ref{sec:fsr-routing-gap-control-draft}, we control the cumulative mixed-delay routing gap $\mathcal S_T^\tau$.  Although $\big(x_{ij}(t-\tau_{ij})\big)_{j\in\DS^+(i)}$ might not lie in $\Delta_i$, shifting each coordinate in time replaces its integral by the nonnegative current-time complementary slackness error.  After the shift, only short
intervals near $0$ and $T$, each of length at most $\bar\tau$, remain.  Their total contribution is bounded by a constant independent of $T$ and the stepsizes.

\item In Subsection~\ref{sec:fsr-mismatch-control-draft}, we control the cost--routing delay mismatch $\Gamma_{ij}(t)$. The basic idea is to use the routing-speed bound from Lemma~\ref{lem:fsr-projection-draft}: when stepsizes are small, decisions do not change much during a round-trip delay window, so the two cross terms approximately cancel after integration.  Because costs can be unbounded, we first split each cost exactly into a bounded component and an unbounded tail.  For the bounded component, the projected dynamics bound changes in routing probabilities by the stepsize times a cost upper bound.  With a sufficiently small stepsize, we can bound the bounded-component delay mismatch by a small fraction of the cumulative workload optimality gap. For the unbounded tail, \eqref{eq:fsr-workload-speed-draft} ensures that the workloads differ by at most $2v\bar\tau$.  Additive slow variation (Assumption~\ref{assume:fsr-cost-draft}) makes the tail mismatch small relative to the workload optimality gap, while the contributions from the short intervals near $0$ and $T$ are bounded separately using the prescribed workload history and the workload Lyapunov function at time $T$, $\Phi(\bN(T))$.
\end{enumerate}

We conclude by integrating
\eqref{eq:fsr-composite-drift-draft}. Therefore, although the composite potential
is not guaranteed to decrease at every instant, its cumulative drift yields
the desired performance bound.
After division by $T$, the initial routing potential contributes
$1/(T\underline\eta)$, the bounded-component time shift contributes $\bar\eta$,
and the remaining contributions from the workload history and the beginning
and end of the time horizon contribute $1/T$.

\subsubsection{The Mixed-Delay Routing Gap}
\label{sec:fsr-routing-gap-control-draft}

Introduce the cumulative workload gap, mixed-delay routing gap, and
current-time routing gap
\begin{equation*}
\begin{aligned}
\mathcal D_T
   &=\sum_{j\in\DS}\int_0^T D_j(N_j(t))\,dt,\\
\mathcal S_T^\tau
   &=\sum_{i\in\DS}\lambda_i\int_0^T S_i^\tau(t)\,dt,\\
\mathcal S_T
   &=\sum_{i\in\DS}\lambda_i\int_0^T S_i(\bx_i(t))\,dt.
\end{aligned}
\end{equation*}
We have
\begin{equation}
    \mathcal D_T\geq0,
       \qquad \mathcal S_T\geq0.
    \label{eq:fsr-cumulative-gaps-draft}
\end{equation}
The integrand defining $\mathcal S_T^\tau$ can be negative: heterogeneous
delays can make
$(x_{ij}(t-\tau_{ij}))_{j\in\DS^+(i)}$ fail to lie in $\Delta_i$.  The next
lemma shows that the cumulative delayed gap can nevertheless be replaced by
the nonnegative current-time gap $\mathcal S_T$, up to an additive constant arising from the
short intervals near $0$ and $T$.

\begin{lemma}[Mixed-delay routing-gap comparison]
\label{lem:fsr-mixed-slackness-draft}
For every $T\geq\bar\tau$,
\begin{equation*}
\begin{aligned}
&\mathcal S_T^\tau
\geq\mathcal S_T
   -\bar\tau\sum_{i\in\DS}\lambda_i
      \max_{j\in\DS^+(i)}c_{ij}(N_j^*).
\end{aligned}
\end{equation*}
\end{lemma}

\begin{proof}
For one origin cell, linearity and a change of variables give
\begin{equation*}
\begin{aligned}
&\int_0^T\sum_{j\in\DS^+(i)}c_{ij}(N_j^*)
   \big(x_{ij}(t-\tau_{ij})-x_{ij}(t)\big)\,dt\\
&=\sum_{j\in\DS^+(i)}c_{ij}(N_j^*)
   \left(\int_{-\tau_{ij}}^0x_{ij}(s)\,ds
        -\int_{T-\tau_{ij}}^Tx_{ij}(s)\,ds\right).
\end{aligned}
\end{equation*}
The first integral is nonnegative.  Since the current routing vector lies in
the simplex, the second integral is at most
\begin{equation*}
    \max_{j\in\DS^+(i)}c_{ij}(N_j^*)
       \int_{T-\bar\tau}^{T}
       \sum_{j\in\DS^+(i)}x_{ij}(s)\,ds
       =\bar\tau\max_{j\in\DS^+(i)}c_{ij}(N_j^*).
\end{equation*}
Multiplying by $\lambda_i$ and summing over $i$ proves the result.\end{proof}

\subsubsection{The Cost--Routing Delay Mismatch}
\label{sec:fsr-mismatch-control-draft}

We now control the cumulative contribution of the cost--routing mismatch terms
$\Gamma_{ij}(t)$.  For finite thresholds $\bar N_j>N_j^*$ to be selected, we
begin with the exact cost decomposition $c_{ij}(u)=b_{ij}(u)+e_j(u)$, where
\begin{equation*}
\begin{aligned}
    b_{ij}(u)&=g_{ij}+\min\{f_j(u),f_j(\bar N_j+2v\bar\tau)\},\\
    e_j(u)&=\big[f_j(u)-f_j(\bar N_j+2v\bar\tau)\big]_+.
\end{aligned}
\end{equation*}
The bounded component $b_{ij}$ is at most
$f_j(\bar N_j+2v\bar\tau)+g_{ij}$.  The unbounded tail $e_j$ is shared by all
arcs entering cell $j$. The extra term $2v\bar{\tau}$ appears because the time-shift argument compares $N_j(t)$ and $N_j(t-2\tau_{ij})$, whose distance is at most $2v\tau_{ij}\leq 2v\bar{\tau}$ by \eqref{eq:fsr-workload-speed-draft}.

The proof relies on two estimates used to control the two parts of this
decomposition.  The first provides a global envelope for the raw costs.

\begin{lemma}[Global cost--gap bound]
\label{lem:fsr-global-cost-gap-draft}
There exist finite thresholds $(\bar N_j)_{j\in\DS}$, with
$\bar N_j>N_j^*$ for every $j$, such that, for every workload $u\geq0$ and
every arc $(i,j)\in\AS$,
\begin{equation}
    c_{ij}(u)\leq
       \max_{(k,\ell)\in\AS}
          \big(f_\ell(\bar N_\ell)+g_{k\ell}\big)
       +
       \frac{2D_j(u)}{\bar\mu_j-\mu_j(N_j^*)}.
    \label{eq:fsr-global-cost-gap-draft}
\end{equation}
\end{lemma}
Lemma~\ref{lem:fsr-global-cost-gap-draft} is used in the
bounded-component argument (Lemma~\ref{lem:fsr-core-shift-draft}) to replace
each raw cost in the routing-speed bound by a fixed term plus a multiple of
$D_j$.

Proposition~\ref{prop:fsr-threshold-existence-draft} shows that the thresholds
in Lemma~\ref{lem:fsr-global-cost-gap-draft} can be chosen so that the
associated tail components also satisfy the following local estimate.  We fix
such a common threshold choice for the remainder of the proof.

\begin{lemma}[Local tail mismatch]
\label{lem:fsr-tail-mismatch-draft}
For every cell $j$, if $u,u'\geq0$ and
$|u-u'|\leq2v\bar\tau$, then
\begin{equation*}
    |e_j(u)-e_j(u')|\leq
       \frac{D_j(u)+D_j(u')}
            {8\sum_{i\in\DS^-(j)}\lambda_i}.
\end{equation*}
\end{lemma}
The shift $2v\bar\tau$ in the definition of $e_j$ matches the largest possible
workload change between the two arguments compared in the time-shift argument.
Thus, if one of the two tail values is nonzero, the smaller workload is above
$\bar N_j$, where additive slow variation makes the change in $f_j$, and
hence in $e_j$, small relative to the workload gap.  Consequently,
Lemma~\ref{lem:fsr-tail-mismatch-draft} charges the tail mismatch to the
workload gaps at its two endpoints; this is the estimate used in
Lemma~\ref{lem:fsr-tail-shift-draft}.

The threshold construction and the proofs of both lemmas are deferred to
Subsection~\ref{sec:fsr-supporting-results-draft}.

The cost decomposition induces the same decomposition of the cumulative
cost--routing mismatch.  Define the bounded-component and tail contributions by
\begin{equation*}
\begin{aligned}
    \Gamma_b(T)
       &=\sum_{(i,j)\in\AS}\lambda_i\int_0^T
          b_{ij}(N_j(t))d_{ij}(t-\tau_{ij})\,dt
       -\sum_{(i,j)\in\AS}\lambda_i\int_0^T
          b_{ij}(N_j(t-\tau_{ij}))d_{ij}(t)\,dt,
\end{aligned}
\end{equation*}
and
\begin{equation*}
\begin{aligned}
    \Gamma_e(T)
       &=\sum_{(i,j)\in\AS}\lambda_i\int_0^T
          e_j(N_j(t))d_{ij}(t-\tau_{ij})\,dt
       -\sum_{(i,j)\in\AS}\lambda_i\int_0^T
          e_j(N_j(t-\tau_{ij}))d_{ij}(t)\,dt.
\end{aligned}
\end{equation*}
By linearity,
\begin{equation*}
    \sum_{(i,j)\in\AS}\lambda_i\int_0^T\Gamma_{ij}(t)\,dt
       =\Gamma_b(T)+\Gamma_e(T).
\end{equation*}
We proceed to bound $\Gamma_b(T)$ and $\Gamma_e(T)$ separately below.

\paragraph{The Bounded Component}

The bounded component can be handled by shifting the routing variables in time.
\begin{lemma}[Bounded-component time shift]
\label{lem:fsr-core-shift-draft}
Let
\begin{equation*}
    K_1=\max_{j\in\DS}
       \frac{2}{\bar\mu_j-\mu_j(N_j^*)}
       \sum_{\substack{i\in\DS^-(j)\\ \ell\in\DS^+(i)}}
       \lambda_i\big(f_\ell(\bar N_\ell+2v\bar\tau)+g_{i\ell}\big)
       \tau_{i\ell},
\end{equation*}
and set
\begin{equation}\label{eq:fsr-eta-choice-draft}
    \eta_0(\boldsymbol\tau)=
    \begin{cases}
       1/(16K_1),&\bar\tau>0,\\
       +\infty,&\bar\tau=0.
    \end{cases}
\end{equation}

For $T\geq2\bar\tau$ and
$\bar\eta\leq\eta_0(\boldsymbol\tau)$, we have
\begin{equation*}
\begin{aligned}
    \Gamma_b(T)&\leq 3K_2
       +4K_2\bar\eta T
          \max_{(k,\ell)\in\AS}
             \big(f_\ell(\bar N_\ell)+g_{k\ell}\big)
       +\frac14\left(H_-+
          \mathcal D_T\right),
\end{aligned}
\end{equation*}
where
\begin{equation*}
    K_2=\sum_{(i,j)\in\AS}\lambda_i
       \big(f_j(\bar N_j+2v\bar\tau)+g_{ij}\big)\tau_{ij},
\end{equation*}
and
\begin{equation*}
    H_-=\sum_{j\in\DS}\int_{-\bar\tau}^{0}
       D_j(N_j(s))\,ds
\end{equation*}
denotes the contribution of the initial workload history.
\end{lemma}

\begin{proof}
Fix an arc, abbreviate $\delta=\tau_{ij}$,
$d(t)=d_{ij}(t)$, and $b(t)=b_{ij}(N_j(t))$.  A direct change of variables gives
the exact identity
\begin{equation*}
\begin{aligned}
&\int_0^T\big(b(t)d(t-\delta)-b(t-\delta)d(t)\big)dt
=\int_0^{T-\delta}b(t)\big(d(t-\delta)-d(t+\delta)\big)dt\\
&\quad+\int_{T-\delta}^{T}b(t)d(t-\delta)dt
       -\int_{-\delta}^{0}b(t)d(t+\delta)dt.
\end{aligned}
\end{equation*}
The last two integrals each have absolute value at most
$\big(f_j(\bar N_j+2v\bar\tau)+g_{ij}\big)\delta$, because $|d(t)|\leq1$
and $b_{ij}$ is bounded above by
$f_j(\bar N_j+2v\bar\tau)+g_{ij}$.

For the first term, on $[0,\delta]$, both routing
coordinates lie in $[0,1]$, so
$|d(t-\delta)-d(t+\delta)|\leq1$, contributing at most another
$\big(f_j(\bar N_j+2v\bar\tau)+g_{ij}\big)\delta$.
On $[\delta,T-\delta]$, both routing times are nonnegative, and the
equilibrium term cancels from the difference:
\begin{equation*}
\begin{aligned}
    d(t-\delta)-d(t+\delta)
       &=x_{ij}(t-\delta)-x_{ij}(t+\delta)
       =-\int_{t-\delta}^{t+\delta}
          \frac{d}{ds}x_{ij}(s)\,ds.
\end{aligned}
\end{equation*}
Lemma~\ref{lem:fsr-projection-draft} gives the routing-speed bound
\begin{equation*}
    \left|\frac{d}{dt}x_{ij}(t)\right|\leq2\eta_i
       \max_{k\in\DS^+(i)}c_{ik}(N_k(t-\tau_{ik})).
\end{equation*}

Fubini's theorem counts each routing time for at most $2\delta$ units of the
outer variable.  Combining these observations gives
\begin{equation}
\begin{aligned}
&\int_0^T\big(b(t)d(t-\delta)-b(t-\delta)d(t)\big)dt
\leq3\big(f_j(\bar N_j+2v\bar\tau)+g_{ij}\big)\delta\\
&\quad+\big(f_j(\bar N_j+2v\bar\tau)+g_{ij}\big)
       \int_{\delta}^{T-\delta}
       \int_{t-\delta}^{t+\delta}
       \left|\frac{d}{ds}x_{ij}(s)\right|\,ds\,dt\\
&\leq3\big(f_j(\bar N_j+2v\bar\tau)+g_{ij}\big)\delta
+4\big(f_j(\bar N_j+2v\bar\tau)+g_{ij}\big)\delta\eta_i
       \int_0^T\max_{k\in\DS^+(i)}
          c_{ik}(N_k(s-\tau_{ik}))\,ds.
\end{aligned}
    \label{eq:fsr-core-arc-bound-draft}
\end{equation}

The global cost--gap bound in
Lemma~\ref{lem:fsr-global-cost-gap-draft} implies
\begin{equation*}
\begin{aligned}
&\max_{k\in\DS^+(i)}c_{ik}(N_k(s-\tau_{ik}))
\leq \max_{(p,q)\in\AS}
          \big(f_q(\bar N_q)+g_{pq}\big)
+\sum_{\ell\in\DS^+(i)}
       \frac{2D_\ell(N_\ell(s-\tau_{i\ell}))}
            {\bar\mu_\ell-\mu_\ell(N_\ell^*)}.
\end{aligned}
\end{equation*}
For each delayed term, a change of variables yields
\begin{equation*}
\begin{aligned}
   \int_0^T D_\ell(N_\ell(s-\tau_{i\ell}))\,ds
       &\leq\int_{-\bar\tau}^{0}D_\ell(N_\ell(u))\,du
       +\int_0^T D_\ell(N_\ell(u))\,du.
\end{aligned}
\end{equation*}
Multiplying \eqref{eq:fsr-core-arc-bound-draft} by $\lambda_i$,
summing over arcs, and using the definitions of $K_2$ and $K_1$ gives
\begin{equation*}
\begin{aligned}
    \Gamma_b(T)&\leq 3K_2
       +4K_2\bar\eta T
          \max_{(k,\ell)\in\AS}
             \big(f_\ell(\bar N_\ell)+g_{k\ell}\big)
       +4K_1\bar\eta\left(H_-+
          \mathcal D_T\right).
\end{aligned}
\end{equation*}
Finally, $\bar\eta\leq\eta_0(\boldsymbol\tau)$ and
\eqref{eq:fsr-eta-choice-draft} imply $4K_1\bar\eta\leq1/4$, which proves the
stated bound.
\end{proof}

\paragraph{The Unbounded Tail}

The tail can be unbounded, so we control its cumulative contribution using
its change over two delay windows and the workload Lyapunov function at time
$T$.  The following lemma gives the complete bound used in the theorem proof.

\begin{lemma}[Tail time shift]
\label{lem:fsr-tail-shift-draft}

For every $T\geq2\bar\tau$,
\begin{equation*}
\begin{aligned}
    \Gamma_e(T)&\leq K_3
       +\frac14\mathcal D_T
       +\frac18H_-+\frac12\Phi(\bN(T))+K_4,
\end{aligned}
\end{equation*}
Here $H_-$ is defined in Lemma~\ref{lem:fsr-core-shift-draft}.  The
beginning-of-horizon constant is
\begin{equation*}
    K_3=\sum_{(i,j)\in\AS}\lambda_i\tau_{ij}
       e_j\big(N_j(0)+v\bar\tau\big),
\end{equation*}
and the final-time constant is
\begin{equation}
    K_4=\sum_{j\in\DS}\sup_{u\geq0}\left[
       2\bar\tau\sum_{i\in\DS^-(j)}\lambda_i
       f_j(u+2v\bar\tau)-\frac12\Phi_j(u)\right]_+.
    \label{eq:fsr-final-time-constant-draft}
\end{equation}
Proposition~\ref{prop:fsr-final-time-constant-draft} proves that
$0\leq K_4<\infty$.
\end{lemma}

\begin{proof}
Fix an arc and again abbreviate $\delta=\tau_{ij}$ and
$d(t)=d_{ij}(t)$.  A backward change of variables gives the exact identity
\begin{equation}
\begin{aligned}
&\int_0^T\left[e_j(N_j(t))d(t-\delta)
       -e_j(N_j(t-\delta))d(t)\right]dt
=\int_0^\delta e_j(N_j(t))d(t-\delta)dt\\
&\quad+\int_\delta^T
   \left[e_j(N_j(t))-e_j(N_j(t-2\delta))\right]d(t-\delta)dt\\
&\quad-\int_{T-\delta}^T e_j(N_j(t-\delta))d(t)dt.
\end{aligned}
    \label{eq:fsr-tail-identity-draft}
\end{equation}
\begin{itemize}

\item {Initial term.}
For $0\leq t\leq\delta$, monotonicity of $e_j$ and
\eqref{eq:fsr-workload-speed-draft} give
$e_j(N_j(t))\leq e_j(N_j(0)+v\bar\tau)$.  Since $|d|\leq1$, the first
term in \eqref{eq:fsr-tail-identity-draft} is at most
$\delta e_j(N_j(0)+v\bar\tau)$.  After multiplying by $\lambda_i$ and
summing over arcs, these terms give $K_3$.

\item{Middle term.}
The two workload arguments differ by at most $2v\bar\tau$ by
\eqref{eq:fsr-workload-speed-draft}.  Since $|d|\leq1$,
Lemma~\ref{lem:fsr-tail-mismatch-draft} bounds the middle term for this arc by
\begin{equation*}
    \frac{1}{8\sum_{k\in\DS^-(j)}\lambda_k}\int_\delta^T
       \left[D_j(N_j(t))+D_j(N_j(t-2\delta))\right]dt.
\end{equation*}
The change of variables $s=t-2\delta$ gives
\begin{equation*}
\begin{aligned}
    \int_\delta^T D_j(N_j(t-2\delta))\,dt
       &\leq\int_{-\bar\tau}^{0}D_j(N_j(s))\,ds
       +\int_0^T D_j(N_j(s))\,ds.
\end{aligned}
\end{equation*}
Thus each incoming arc contributes at most
\begin{equation*}
    \frac{1}{8\sum_{k\in\DS^-(j)}\lambda_k}\left(
       2\int_0^T D_j(N_j(t))\,dt
       +\int_{-\bar\tau}^0D_j(N_j(t))\,dt\right).
\end{equation*}
After multiplying by $\lambda_i$ and summing over incoming arcs, their
weights cancel the denominator.  Summing over cells therefore gives
\begin{equation*}
\begin{aligned}
    &\sum_{(i,j)\in\AS}\lambda_i
       \int_{\tau_{ij}}^T
       \big|e_j(N_j(t))-e_j(N_j(t-2\tau_{ij}))\big|dt
    \leq\frac14\mathcal D_T
       +\frac18H_-.
\end{aligned}
\end{equation*}

\item {Final term.}
Because $e_j\geq0$ and $-d(t)\leq|d(t)|\leq1$, the sum of the last terms in
\eqref{eq:fsr-tail-identity-draft} is at most
\begin{equation*}
    \sum_{(i,j)\in\AS}\lambda_i
       \int_{T-\tau_{ij}}^T e_j(N_j(t-\tau_{ij}))\,dt.
\end{equation*}
For an arc with delay $\delta$, a change of variables maps this integral to
$[T-2\delta,T-\delta]$, which is contained in
$[T-2\bar\tau,T]$.  Since $T\geq2\bar\tau$,
\eqref{eq:fsr-workload-speed-draft} gives, throughout this interval,
\begin{equation*}
    e_j(N_j(s))\leq f_j(N_j(s))
       \leq f_j(N_j(T)+2v\bar\tau).
\end{equation*}
Summing the incoming arc weights for cell $j$ and enlarging each integration
interval to $[T-2\bar\tau,T]$ gives
\begin{equation*}
\begin{aligned}
&\sum_{i\in\DS^-(j)}\lambda_i
   \int_{T-\tau_{ij}}^T e_j(N_j(t-\tau_{ij}))\,dt
\leq 2\bar\tau\sum_{i\in\DS^-(j)}\lambda_i
   f_j(N_j(T)+2v\bar\tau)\\
&\leq\frac12\Phi_j(N_j(T))
+\sup_{u\geq0}\left[
       2\bar\tau\sum_{i\in\DS^-(j)}\lambda_i
       f_j(u+2v\bar\tau)-\frac12\Phi_j(u)\right]_+.
\end{aligned}
\end{equation*}
Summing over $j$ bounds the final contribution by
$\frac12\Phi(\bN(T))+K_4$.

\end{itemize}

Adding the bounds for the three terms proves
the claim.
\end{proof}

\subsubsection{Putting the Estimates Together}
\label{sec:fsr-final-assembly-draft}
\begin{proof}[Proof of Theorem~\ref{thm:fsr-raw-cumulative-draft}]
Integrating the inequality in
Lemma~\ref{lem:fsr-composite-drift-draft} from $0$ to $T$ and moving the
workload and delayed-routing terms to the left gives
\begin{equation}
\begin{aligned}
&V(\bx(T))+\Phi(\bN(T))+\mathcal D_T+\mathcal S_T^\tau
\leq V(\bx(0))+\Phi(\bN(0))
   +\sum_{(i,j)\in\AS}\lambda_i\int_0^T\Gamma_{ij}(t)\,dt.
\end{aligned}
    \label{eq:fsr-integrated-composite-drift-draft}
\end{equation}
Applying Lemma~\ref{lem:fsr-mixed-slackness-draft}, Lemma~\ref{lem:fsr-core-shift-draft}, and Lemma~\ref{lem:fsr-tail-shift-draft}, we obtain
\begin{equation*}
\begin{aligned}
&V(\bx(T))+\Phi(\bN(T))
   +\mathcal D_T+\mathcal S_T
\leq V(\bx(0))+\Phi(\bN(0))\\
&\quad+\bar\tau\sum_{i\in\DS}\lambda_i
      \max_{j\in\DS^+(i)}c_{ij}(N_j^*)
   +3K_2+K_3+K_4
+4K_2\bar\eta T
   \max_{(k,\ell)\in\AS}
      \big(f_\ell(\bar N_\ell)+g_{k\ell}\big)\\
&\quad+\frac12\mathcal D_T+\frac38H_-
   +\frac12\Phi(\bN(T)).
\end{aligned}
\end{equation*}
Since $\bN$ is continuous and $D_j$ and $\Phi_j$ are continuous,
$\mathcal D_T$, $H_-$, and $\Phi(\bN(T))$ are finite for every finite $T$.
Thus we may subtract $\frac12\mathcal D_T$ and
$\frac12\Phi(\bN(T))$ from both sides.  This leaves
$V(\bx(T))+\frac12\Phi(\bN(T))$ on the left; discarding those two
nonnegative final-time quantities gives
\begin{equation}
\begin{aligned}
&\frac12\mathcal D_T+\mathcal S_T
\leq V(\bx(0))+\Phi(\bN(0))\\
&\quad+\bar\tau\sum_{i\in\DS}\lambda_i
      \max_{j\in\DS^+(i)}c_{ij}(N_j^*)
   +3K_2+K_3+K_4
+\frac38H_-
   +4K_2\bar\eta T
      \max_{(k,\ell)\in\AS}
         \big(f_\ell(\bar N_\ell)+g_{k\ell}\big).
\end{aligned}
    \label{eq:fsr-after-final-absorption-draft}
\end{equation}
Set
\begin{equation*}
\begin{aligned}
    G_1(\boldsymbol\tau)
       &=8K_2\max_{(k,\ell)\in\AS}
          \big(f_\ell(\bar N_\ell)+g_{k\ell}\big),\\
    G_2\!\left(
       \boldsymbol\tau,\left.\bN\right|_{[-\bar\tau,0]}
    \right)
       &=2\Phi(\bN(0))
       +2\bar\tau\sum_{i\in\DS}\lambda_i
          \max_{j\in\DS^+(i)}c_{ij}(N_j^*)
       +6K_2+2K_3+2K_4+\frac34H_-.
\end{aligned}
\end{equation*}
Using the definition of $V(\bx(0))$ in
\eqref{eq:fsr-routing-potential} and multiplying both sides of \eqref{eq:fsr-after-final-absorption-draft} by $2/T$, we obtain
\begin{equation}
\begin{aligned}
&\frac{\mathcal D_T}{T}+\frac{2\mathcal S_T}{T}
\leq\frac1T\sum_{i\in\DS}\frac{\lambda_i}{\eta_i}
      \|\bx_i(0)-\bx_i^*\|_2^2
   +G_1(\boldsymbol\tau)\bar\eta
   +\frac{G_2\!\left(
      \boldsymbol\tau,\left.\bN\right|_{[-\bar\tau,0]}
   \right)}{T}.
\end{aligned}
    \label{eq:fsr-main-bound-draft}
\end{equation}
The result thus follows as $\|\bx_i(0)-\bx_i^*\|_2^2\leq2$ and
$1/\eta_i\leq1/\underline\eta$ and $\mathcal S_T\geq0$ by
\eqref{eq:fsr-cumulative-gaps-draft}.
\end{proof}

\subsubsection{{Supporting Results}}
\label{sec:fsr-supporting-results-draft}

\begin{lemma}\label{lem:fsr-projection-draft} Fix a cell $i \in \DS$ and a probability vector $\bx_i \in \Delta_i$. Let $\bv = \Pi_{T_{\Delta_i}(\bx_i)} \left(\bz_i \right)$ be the projection to the tangent cone of $\Delta_i$ at $\bx_i$ of the vector $\bz_i \in \mathbb R^{|\DS|}$. The following holds.
\begin{enumerate}
    \item Contraction property. For every $\bx_i' \in \Delta_i$ we have  $(\bx_i - \bx_i')^\top \bv  \le (\bx_i - \bx_i')^\top \bz_i$.

    \item Boundedness. We have that $\| \bv \|_\infty \le 2\|\bz_i\|_\infty$.
\end{enumerate}
\end{lemma}
\begin{proof}
We prove each part at a time.

\paragraph{Part 1 (Contraction Property).}
Recall that the tangent cone is given by
\[
 T_{\Delta_i}(\bx_i) = \left\{\bv \in \mathbb{R}^{|\DS|}: \sum_{j\in \DS} v_j = 0,  v_j \geq 0  \text{ if } x_{ij}  = 0, v_j = 0 \text{ for all } j \not\in \DS^+(i) \right\}\,.
 \]

Since $T_{\Delta_i}(\bx_i)$ is a closed convex cone, the projection $\bv$ is uniquely characterized by the property that for any $\boldsymbol{u} \in T_{\Delta_i}(\bx_i)$,
\begin{equation*}
    (\bv - \bz_i)^\top (\boldsymbol{u} - \bv) \ge 0. %
\end{equation*}
Taking $\boldsymbol u=\boldsymbol0$ and $\boldsymbol u=2\bv$ in this characterization gives
$(\bv-\bz_i)^\top\bv=0$.  Because $T_{\Delta_i}(\bx_i)$ is a cone, substitution of any $\boldsymbol u\in T_{\Delta_i}(\bx_i)$ and this orthogonality identity then give
\begin{equation}
    (\bv - \bz_i )^\top \boldsymbol{u} \ge 0, \quad \forall \boldsymbol{u} \in T_{\Delta_i}(\bx_i)\,. \label{eq:polar_cond}
\end{equation}

Let $\bx'_i \in \Delta_i$. Consider the vector $\boldsymbol{u} = \bx'_i - \bx_i$. It is straightforward to verify that $\boldsymbol{u} \in T_{\Delta_i}(\bx_i)$ because $\boldsymbol{u} = \bx'_i - \bx_i$ is a feasible direction in the tangent cone. By the projection property \eqref{eq:polar_cond}, we have:
\[
(\bv - \bz_i)^\top (\bx'_i - \bx_i) \ge 0\,,
\]
and the result follows from re-arranging terms.

\paragraph{Part 2 (Boundedness).} The projection $v$ is the solution to the optimization problem:
\begin{equation*}
\begin{aligned}
& \min_{v} & & \frac{1}{2} \sum_{j \in \DS^+(i)} (v_j - z_{ij})^2 \\
& \text{subject to} & & \sum_{j \in \DS^+(i)} v_j = 0, \\
& & & v_j \ge 0, \quad \forall j \in \DS^+(i) : x_{ij} = 0\,,
\end{aligned}
\end{equation*}
where we removed from the objective all $j \notin \DS^+(i)$ because the tangent cone enforces $v_j = 0$ for those indices, so their contribution $(v_j - z_{ij})^2 = z_{ij}^2$ is constant with respect to $\bv$ and does not affect the optimal solution.

Let $I_0 = \{j \in \DS^+(i) : x_{ij} = 0\}$. The Lagrangian is:
\[
\mathcal{L}(\bv, \beta, \alpha) = \frac{1}{2} \sum_{j \in \DS^+(i)} (v_j - z_{ij})^2 - \beta \sum_{j \in \DS^+(i)} v_j - \sum_{j \in I_0} \alpha_j v_j,
\]
where $\beta$ is the Lagrange multiplier of the equality constraint and $\alpha_j \ge 0$ are the Lagrange multipliers of the non-negativity constraints. The first-order optimality conditions are necessary and sufficient for optimality because the problem is convex with linear constraints, and differentiable. The first-order condition with respect to $v_j$ is:
\[
\frac{\partial \mathcal L}{\partial v_j} = v_j - z_{ij} - \beta - \mathbf 1\{j \in I_0\} \alpha_j = 0 \implies v_j = z_{ij} + \beta + \mathbf 1\{j \in I_0\} \alpha_j\,.
\]
We also have the complementary slackness condition $\alpha_j v_j = 0$ for $j \in I_0$.
\begin{itemize}
    \item If $j \in I_0$ and the constraint is active ($v_j=0$), then trivially $|v_j| = 0 \le \|\bz_i\|_\infty$.
    \item If $j \notin I_0$ or the constraint is inactive ($\alpha_j=0$), then $v_j = z_{ij} + \beta$.
\end{itemize}
The constant $\beta$ is determined by the constraint $\sum_{j \in \DS^+(i)} v_j = 0$. Let $S = \DS^+(i) \setminus \{j \in I_0 \mid v_j = 0\}$ be the set of indices where the non-negativity constraint is inactive.  This set is nonempty because $\bx_i$ has at least one positive component. Then:
\[
\sum_{j \in S} (z_{ij} + \beta) = 0 \implies \beta = -\frac{1}{|S|} \sum_{j \in S} z_{ij}\,.
\]
Thus, for non-zero components,
\begin{align*}
    |v_j| &= \left| z_{ij} - \frac{1}{|S|} \sum_{j' \in S} z_{ij'} \right|
    \le \|\bz_i\|_\infty + \frac{1}{|S|} \sum_{j' \in S} \|\bz_i\|_\infty = 2 \|\bz_i\|_\infty\,,
\end{align*}
where the inequality follows from the triangle inequality together with $|z_{ij'}| \le \|\bz_i\|_\infty$. Thus, $\|\bv\|_\infty \le 2\|\bz_i\|_\infty$.
\end{proof}

 \begin{lemma}[Two consequences of additive slow variation]
\label{lem:fsr-consequences-draft}
For every cell $j$ and every fixed $H\geq0$,
\begin{equation}
    \sup_{u\geq r}\ \sup_{0\leq h\leq H}
       \left(\frac{f_j(u+h)}{f_j(u)}-1\right)
       \longrightarrow0\quad\text{as }r\to\infty.
    \label{eq:fsr-modulus-draft}
\end{equation}
Moreover,
\begin{equation}
    \frac{\Phi_j(u)}{f_j(u)}\longrightarrow\infty
       \quad\text{as }u\to\infty.
    \label{eq:fsr-p-over-f-draft}
\end{equation}
\end{lemma}

\begin{proof}
For $0\leq h\leq H$, monotonicity gives
\begin{equation*}
    0\leq\frac{f_j(u+h)}{f_j(u)}-1
       \leq\frac{f_j(u+H)}{f_j(u)}-1.
\end{equation*}
The definition of a limit then makes the right-hand side uniformly small for
all sufficiently large $u$, proving \eqref{eq:fsr-modulus-draft}.

Fix an arbitrary $L>0$.  For all sufficiently large $u$,
\begin{equation*}
    \Phi_j(u)\geq
       \int_{u-L}^{u}\big(f_j(y)-f_j(N_j^*)\big)\,dy
       \geq L\big(f_j(u-L)-f_j(N_j^*)\big).
\end{equation*}
Applying additive slow variation at $u-L$ with shift $L$, and then taking
reciprocals, gives $f_j(u-L)/f_j(u)\to1$.  Coercivity also gives
$f_j(N_j^*)/f_j(u)\to0$.  Hence
$\liminf_{u\to\infty}\Phi_j(u)/f_j(u)\geq L$.  Since $L$ is arbitrary,
\eqref{eq:fsr-p-over-f-draft} follows.
\end{proof}

We now justify the common threshold choice used in
Lemmas~\ref{lem:fsr-global-cost-gap-draft} and
\ref{lem:fsr-tail-mismatch-draft}.

\begin{proposition}[Common workload thresholds]
\label{prop:fsr-threshold-existence-draft}
For every cell $j\in\DS$, there exists a finite threshold
$\bar N_j>N_j^*$ with the following two properties.  Above $\bar N_j$,
$f_j$ is at most $2/(\bar\mu_j-\mu_j(N_j^*))$ times the workload gap $D_j$.
Moreover, starting from any workload at or above $\bar N_j$, every increment
of length at most $2v\bar\tau$ increases $f_j$ by at most
$(\bar\mu_j-\mu_j(N_j^*))/(16\sum_{i\in\DS^-(j)}\lambda_i)$ as a fraction of
its initial value.
Any larger threshold has the same two properties.
\end{proposition}

\begin{proof}
The two quantities are positive because $\mu_j$ is strictly increasing,
$N_j^*$ is finite, every cell has an in-neighbor, and every arrival rate is
positive.  We verify the two properties separately.

\paragraph{First property.}
For every sufficiently large $\bar N_j$, we have
\begin{equation*}
    f_j(u)\leq \frac{2D_j(u)}{\bar\mu_j-\mu_j(N_j^*)},
    \qquad u\geq\bar N_j.
\end{equation*}
For all sufficiently large $u$, coercivity gives $f_j(u)>0$, and the
definition of $D_j$ gives
\begin{equation*}
\begin{aligned}
    \lim_{u\to\infty}\frac{D_j(u)}{f_j(u)}
    &=\lim_{u\to\infty}
       \left(1-\frac{f_j(N_j^*)}{f_j(u)}\right)
       \left(\mu_j(u)-\mu_j(N_j^*)\right)
    =\bar\mu_j-\mu_j(N_j^*).
\end{aligned}
\end{equation*}
Indeed, coercivity gives $f_j(N_j^*)/f_j(u)\to0$, while bounded monotonicity
gives $\mu_j(u)\to\bar\mu_j$.
Thus $D_j(u)/f_j(u)\geq(\bar\mu_j-\mu_j(N_j^*))/2$ beyond some finite workload level.
Rearranging gives the displayed estimate beyond that level.

\paragraph{Second property.}
For every sufficiently large $\bar N_j$, we also have
\begin{equation*}
    \frac{2}{\bar\mu_j-\mu_j(N_j^*)}
    \sup_{u\geq\bar N_j}\ \sup_{0\leq h\leq2v\bar\tau}
       \left(\frac{f_j(u+h)}{f_j(u)}-1\right)
       \leq\frac{1}{8\sum_{i\in\DS^-(j)}\lambda_i}.
\end{equation*}
Apply Lemma~\ref{lem:fsr-consequences-draft} with the fixed displacement
$H=2v\bar\tau$.  By \eqref{eq:fsr-modulus-draft}, the nested supremum in the
display tends to zero as $\bar N_j\to\infty$.  It is therefore at most
$(\bar\mu_j-\mu_j(N_j^*))/(16\sum_{i\in\DS^-(j)}\lambda_i)$ for every sufficiently large choice of
$\bar N_j$; multiplying by $2/(\bar\mu_j-\mu_j(N_j^*))$ proves the displayed estimate.

Choose one finite $\bar N_j>N_j^*$ large enough for both estimates, and do
this for every cell.  The proof of Lemma~\ref{lem:fsr-global-cost-gap-draft}
below uses the first property, while the proof of
Lemma~\ref{lem:fsr-tail-mismatch-draft} uses both.  Hence this single
collection of thresholds gives both conclusions simultaneously.
Enlarging any threshold only shrinks the workload ranges on which the two
properties must hold, so both remain valid.
\end{proof}

\begin{proof}[Proof of Lemma~\ref{lem:fsr-global-cost-gap-draft}]
Fix a workload $u\geq0$ and an arc $(i,j)\in\AS$.  If
$u<\bar N_j$, monotonicity gives
\begin{equation*}
    c_{ij}(u)=f_j(u)+g_{ij}
       \leq f_j(\bar N_j)+g_{ij}
       \leq\max_{(k,\ell)\in\AS}
          \big(f_\ell(\bar N_\ell)+g_{k\ell}\big).
\end{equation*}
If $u\geq\bar N_j$, the same maximum bounds $g_{ij}$, while the first property
of Proposition~\ref{prop:fsr-threshold-existence-draft} gives
\begin{equation*}
    f_j(u)\leq
       \frac{2D_j(u)}{\bar\mu_j-\mu_j(N_j^*)}.
\end{equation*}
In the first case, the claimed inequality follows because $D_j(u)\geq0$.  In
the second case, adding the bounds on $g_{ij}$ and $f_j(u)$ proves
\eqref{eq:fsr-global-cost-gap-draft}.
\end{proof}

\begin{proof}[Proof of Lemma~\ref{lem:fsr-tail-mismatch-draft}]
Recall $e_j(u)=\big[f_j(u)-f_j(\bar N_j+2v\bar\tau)\big]_+$.
If $\max\{u,u'\}\leq\bar N_j+2v\bar\tau$, then both tail values are zero.
Otherwise, suppose without loss of generality that $u\geq u'$ and
$u>\bar N_j+2v\bar\tau$.  Then $e_j(u)\geq e_j(u')$ and
$u'\geq u-2v\bar\tau>\bar N_j$.  Moreover,
\begin{equation*}
\begin{aligned}
   e_j(u)-e_j(u')
       &\leq f_j(u)-f_j(u')
       \leq\frac{\bar\mu_j-\mu_j(N_j^*)}
          {16\sum_{i\in\DS^-(j)}\lambda_i}f_j(u')
       \leq\frac{D_j(u')}
          {8\sum_{i\in\DS^-(j)}\lambda_i}.
\end{aligned}
\end{equation*}
For the first inequality, equality holds when
$u'\geq\bar N_j+2v\bar\tau$.  If
$u'<\bar N_j+2v\bar\tau$, then $e_j(u')=0$ and
\begin{equation*}
    e_j(u)=f_j(u)-f_j(\bar N_j+2v\bar\tau)
       \leq f_j(u)-f_j(u').
\end{equation*}
The second inequality uses $u'>\bar N_j$, $u-u'\leq2v\bar\tau$, and the
second property of Proposition~\ref{prop:fsr-threshold-existence-draft}; the
third uses its first property.  Interchanging $u$ and $u'$ handles the other
ordering.  Since $D_j\geq0$, either ordering implies the stated symmetric
bound.
\end{proof}

We finally verify that the constant $K_4$ used above is finite.

\begin{proposition}[Finiteness of the final-time constant]
\label{prop:fsr-final-time-constant-draft}
The quantity $K_4$ in \eqref{eq:fsr-final-time-constant-draft} satisfies
$0\leq K_4<\infty$.  For a fixed network, collection of arrival rates,
service and cost functions, equilibrium, and delay profile, $K_4$ is
independent of the horizon $T$, the trajectory, the workload history, and the
stepsizes.
\end{proposition}

\begin{proof}
Lemma~\ref{lem:fsr-consequences-draft}, additive slow variation with
$H=2v\bar\tau$, and coercivity imply, for every cell $j$,
\begin{equation*}
    \frac{\Phi_j(u)}{f_j(u+2v\bar\tau)}
       =\frac{\Phi_j(u)}{f_j(u)}
          \frac{f_j(u)}{f_j(u+2v\bar\tau)}
       \longrightarrow\infty.
\end{equation*}
Since the coefficient $2\bar\tau\sum_{i\in\DS^-(j)}\lambda_i$ is finite,
the preceding limit gives, for all sufficiently large $u$,
\begin{equation*}
    2\bar\tau\sum_{i\in\DS^-(j)}\lambda_i
       f_j(u+2v\bar\tau)
       \leq\frac12\Phi_j(u).
\end{equation*}
Therefore, the expression inside the positive part in
\eqref{eq:fsr-final-time-constant-draft} is nonpositive for all sufficiently
large $u$.  It is continuous in $u$, so its positive part attains a finite
maximum on the remaining compact interval.  Each supremum is thus finite,
and their sum is finite because $\DS$ is finite.  Nonnegativity follows from
the positive part.  Finally, the definition of $K_4$ involves only the fixed
problem data listed in the proposition, which proves the claimed
independence.  In particular, if $\bar\tau=0$, then every summand is
$\sup_{u\geq0}[-\Phi_j(u)/2]_+=0$, so $K_4=0$.
\end{proof}

\end{document}